\documentclass[11pt]{amsart}
\usepackage{geometry}                
\usepackage{graphicx}
\usepackage{epstopdf}
\usepackage{color}
\usepackage{amstext, amsmath,latexsym,amsbsy,amssymb}

\newtheorem{theorem}{Theorem}[section]
\newtheorem{proposition}{Proposition}[section]
\newtheorem{lemma}{Lemma}[section]
\newtheorem{corollary}{Corollary}[section]
\newtheorem{remark}{Remark}[section]

\newcommand{\at}{\tilde{a}}

\begin{document}

\title[Exponentially-tailed regularity and decay rate for Boltzmann]{Exponentially-tailed regularity and decay rate to equilibrium for the Boltzmann equation}

\maketitle

\begin{center}
Ricardo Alonso*, Irene M. Gamba**, Maja Taskovi\'{c} ***\\
*Department of Mathematics \\
Texas A\&M University at Qatar\\
email: ricardo.alonso@qatar.tamu.edu\\
**Department of Mathematics  and Oden Institute\\
University of Texas at Austin\\
email: gamba@math.utexas.edu\\
***Department of Mathematics \\
Emory University\\
email: maja.taskovic@emory.edu
\end{center}

\begin{abstract}
After revisiting the existence and uniqueness theory of solutions to the homogeneous Boltzmann equation whose transition probabilities (or collision kernels) \cite{ AGbams, MiscWenn} are given by Maxwell type and hard intramolecular potentials, under just integrability condition for the angular scattering kernel, we present in this manuscript several new results.  We start by showing the  Lebesgue and Sobolev propagation of the exponential tails for such solutions. Previous results required stronger angular scattering kernel integrability conditions \cite{ATh, GPV}. { We point out that one of the novel tools for obtaining these results includes pointwise (i.e. strong) commutators between fractional derivatives and the collision operator.  The paper includes the analysis for the critical case of Maxwell interactions corresponding to propagation of tails rather than generation. } 
In addition, we show new estimates giving  $L^{p}$-integrability  generation of exponential tails in the case of hard potential interactions in the range $p\in[1,\infty]$, exponentially-fast convergence rate  to thermodynamical equilibrium (under rather general physical initial data), and regularization in the sense of exponential attenuation of singularities. In many ways, this work is an improvement and an extension of several classical works in the area \cite{AGbams,ATh,A,BG,GPV,MV,Wennberg-stability}.  We, both,  use known techniques and introduce new and flexible ideas that achieve the proofs in a rather  elementary manner.
\end{abstract}
\vspace{.8cm}


\section{Introduction}
We study qualitative and quantitative aspects of the exponentially weighted Lebesgue $L^{p}_{exp}$  integrability, propagation and generation, in the sense of tails, and the exponentially weighted Sobolev $H^{k}_{exp}$  propagation for the  homogeneous Boltzmann equation with hard potentials and Maxwell molecules for integrable angular scattering kernel.  In addition, we study the exponentially fast asymptotic convergence of solutions towards thermodynamical equilibrium under general physical initial data, propagation of regularity using pointwise (strong) commutators between the collision operator and fractional differentiation, and exponential attenuation of singularities.  The proof is valid for any dimension $d\geq 2$, integrability $p\in[1,\infty]$, regularity $k\geq0$, and integrability assumption on the scattering kernel.

More precisely, in this paper we study the Cauchy problem for   the  spatially homogeneous Boltzmann equation with hard potentials and Maxwell molecules for integrable angular scattering kernel
\begin{equation}\label{HBE}
\begin{aligned}
\begin{cases}
&\partial_{t}f(t,v) = Q_{\gamma,b}(f,f)(t,v)\,,\quad (t,v)\in\mathbb{R}^{+}\times\mathbb{R}^{d},\\
&f(0,v) = f_0(v),
\end{cases}
\end{aligned}
\end{equation}
where the collision operator is given by
\begin{align*}
&Q_{\gamma,b}(f,g)(v)=Q^{+}_{\gamma,b}(f,g) - Q^{-}_{\gamma,b}(f,g)\\
&\hspace{1cm}:=\int_{\mathbb{R}^{d}}\int_{\mathbb{S}^{d-1}}f(v')g(v'_{*})B(|u|,\hat{u}\cdot w)\text{d}w\text{d}v_{*}
- f(v) \int_{\mathbb{R}^{d}}\int_{\mathbb{S}^{d-1}}g(v_{*})B(|u|,\hat{u}\cdot w)\text{d}w\text{d}v_{*}\,,
\end{align*}
and where the collisional variables are defined as
\begin{equation*}
v':=v - u^{-}, \qquad v'_{*} := v_{*} + u^{-}, \qquad u:=v-v_{*},\qquad u^{\pm}:= \frac{u \pm |u|w}{2}\,.
\end{equation*}
The scattering angle $\theta$ is simply defined as $\cos(\theta):=\hat{u}\cdot w$, while the scattering kernel is assumed to have the form
\begin{align*}
B(|u|,\hat{u}\cdot w) = |u|^\gamma b(\hat{u}\cdot w), \; \text{ with }\,\;\gamma\in[0,2]\,,
\end{align*}
where $b$ is assumed to be integrable over the sphere $w\in \mathbb{S}^{d-1}$, and we can set, without loss of generality, $\|b\|_{L^{1}(\mathbb{S}^{d-1})}=1$. Another way of representing the kernel is via
\begin{equation*}
B(x, y) = x^{\gamma}\,b(y)\,,\; \text{ with }\;\; x\geq0,\;\;y\in[0,1],\;\;\gamma\in[0,2].
\end{equation*}
 We note that the support of $b$ is assumed to be in $[0,1]$ thanks to a symmetrization argument (thus, $b$ is assumed to be symmetrized \cite[pg. 3]{GPV}).
When the dependence of the angular scattering kernel $b$ is important, say in the constants involved in the estimates, it will be explicitly stated. 

We will commonly use, for technical reasons, the  following decomposition of the angular scattering kernel to estimate distinct parts of the collision operator
\begin{equation}\label{DAK}
\setlength{\jot}{.3cm}
\begin{aligned}
b\big(\cos(\theta)\big) &= b\big(\cos(\theta)\big) \big(\text{1}_{|\sin(\theta)| \geq \varepsilon} + \text{1}_{|\sin(\theta)| < \varepsilon}\big)=:b^{\varepsilon}_{1}\big(\cos(\theta)\big) + b^{\varepsilon}_{2}\big(\cos(\theta)\big)\,,\\
&\hspace{3cm}\text{ or, in terms of } y \\
b(y) &= b(y) \big(\text{1}_{ y\leq\sqrt{1-\varepsilon^{2}} } + \text{1}_{ y>\sqrt{1-\varepsilon^{2}} }\big)=:b^{\varepsilon}_{1}(y) + b^{\varepsilon}_{2}(y)\,.
\end{aligned}
\end{equation}
 This decomposition is motivated by the Young's inequality for the gain collision operator (see   Theorem \ref{ApT-1} for the statement), and is designed to separate grazing collisions (close to $\theta =0$) for which a singularity appears in the definition of the constant \eqref{C-1}. Different estimates will be applied to singular part of the kernel $ b^{\varepsilon}_{2}$ and the part corresponding to the remainder $ b^{\varepsilon}_{1}$  (see Lemma \ref{APq+}).

The initial data considered in this document is \textit{nonnegative} and has finite mass and energy,
\begin{equation}\label{iconditions}
\int_{\mathbb{R}^{d}}f_0\langle v \rangle^{2}\text{d}v<\infty,\qquad \text{ where }\;\langle v \rangle:= \sqrt{1+|v|^{2}},
\end{equation} 
moreover, without loss of generality we can set the initial mass equal to unity.  In this case, equation \eqref{HBE} has a unique solution $f(t,v)\geq0$ that conserves initial mass, momentum and energy, see \cite{MiscWenn}.
\\

Before starting with the technical details, let us mention that propagation of $L^{p}$ integrability for the Boltzmann equation, with different degrees of cut-off and weights, has been studied for quite some time now. In the case of polynomial weights, classical papers for the hard potential case are \cite{C,A} for $p=\infty$.   The case  $p\in(1,\infty)$ has been studied in \cite{G,MV,DM} and  $p=1$ in \cite{MiscWenn,D}, while the  Maxwell molecules model for the whole range  $p\in[1,\infty]$ is studied in \cite[Prop. 1.4]{Wennberg-stability1}.  In particular, the $L^{\infty}$-norm propagation has been shown for $b$ assumed bounded.  Interestingly in all these works, entropy plays a subtle role in the argument: it is used to find a suitable \textit{a priori} uniform lower bound for the \textit{loss} collision operator.  In fact, we show that for hard potentials entropy will not play such a relevant job as far as propagation of $L^{p}$-norms is concerned, instead, the propagation of a higher moment will be key.  Yet, for the Maxwell model, it plays a central part when estimating the \textit{gain} collision operator in a suitable way.

The propagation of $L^p$-exponentially weighted norms is more uncommon in the literature, however, reference \cite{GPV} is an illuminating example of it, where pointwise ($L^\infty$) gaussian estimates are shown to be propagated in the case of hard potentials using a comparison principle.  These techniques have been applied later for the Maxwell case in \cite{BG} after nontrivial modifications.  In these works, more restrictive assumptions on the angular function $b$ are assumed. In particular, $b$ has to behave in a specific way at the grazing angle $\theta=0$. 
Exponentially weighted $L^1$-norms are well understood in the general case $b\in L^1$ and have been extensively studied in \cite{ B, GPV, Mo, ACGM} for $\gamma>0$ and in \cite{B,BG}  for $\gamma=0$. The work in \cite{AGbams} emphasizes on the dependance of the characterization of the exponential rates depending on the coerciveness of the Boltzmann flow in $L^1_\mu(\mathbb{R}^d)$-norms.

For previous results on propagation of Sobolev norms we refer to \cite{MV} with polynomial weights, and \cite{ATh} with gaussian weights. These references are restricted to the hard potential case with cut-off condition $b\in L^{2}(\mathbb{S}^{d-1})$.  For the Maxwell molecules model, the propagation of smoothness and its relation with the relaxation towards equilibrium is studied in \cite{CGT}.  In this document, the analysis of propagation of regularity and relaxation are independent giving a more robust and general approach.

One of the main contributions of the present document with respect to  $L^{p}$-integrability and   propagation of Sobolev norms is to unify all previous works with a relatively simple line of reasoning that includes all ranges of integrability $p\in[1,\infty]$ and weights, polynomial and exponential, including gaussian.  It also includes both, Maxwell and hard interactions. Another novelty of our paper is that we  extend  the range of $p$ for the propagation of $L^p$ exponentially weighted norms from $p=1$ and $p=\infty$ to the entire range $p\in [1,\infty]$. We also relax the assumption on the angular kernel for the $L^\infty$ results by asking that the kernel is merely integrable (as opposed to being bounded in the previous works). Finally, in the context of Sobolev norms we relax the assumption  from $b\in L^2(\mathbb S^{d-1})$ to $b\in L^1(\mathbb S^{d-1})$,  we are able to prove propagation of fractional regularity by introducing \textit{pointwise commutators} between fractional differentiation and the collision operator, and we consider Maxwell molecules.

 The program to prove Lebesgue exponential tail propagation/generation consists of 3 main steps, \vspace{6pt}

 { 1)} prove propagation and generation  of exponential moments \cite{B, B1, BGP, ACGM, BG, MM, AGbams}, \vspace{6pt}
  
  { 2)}prove a so-called ``gain of integrability" inequality for the gain collision operator in the spirit of \cite{MV,AG,AGbams} and, \vspace{6pt}
 
  {3)} use Young's inequality \cite{G,ACG,AGbams} for the gain collision operator for  the case $p=\infty$. \vspace{6pt}

  In addition to these main steps, we will need an explicit lower bound for the negative part of the collision operator which seems to be classical in the literature, at least when finite initial entropy is assumed, see Lemma \ref{APq-} below.  Contrary to the hard potential case, the critical case of Maxwell interactions will need the \textit{a priori} propagation of entropy.  That is, the additional assumption
\begin{equation*}
\int_{\mathbb{R}^{d}}f_{0}\ln(f_0)\text{d}v<\infty
\end{equation*}
will be required on the initial condition.  This is harmless in our context since we will impose more restrictive conditions on $f_0$, namely, $f_{0}\in\big(L^{1}_{2}\cap L^{p}\big)(\mathbb{R}^{d})$ for some $p>1$, when proving propagation of Lebesgue norms.

 This short summary ends by stating a few words about the relaxation to thermal equilibrium of solutions to the Boltzmann equation.  The literature is ample in this respect, notable examples are \cite{A1,Mo,MV,Wennberg-stability,CGT} and the references therein.  It is well established that such convergence will occur with exponential rate in $L^{p}$-norms, even Sobolev norms, provided enough regularity is satisfied by the angular scattering kernel and the initial data.  This holds also for the case of Maxwell molecules \cite{CGT} for initial data with enough tail decay and regularity.  In all aforementioned references such angular scattering kernel must be at least bounded and the initial condition is usually taken at least in $L^{2}(\mathbb{R}^{d})$.  In references \cite{A1,Wennberg-stability} a clever approach is introduced using a dyadic splitting of the collision operator.  In \cite{Mo,MV} the analysis relies on the decomposition theorem, that is, propagation of smoothness and roughness, and the entropic methods, see for instance \cite{V}.  In both approaches spectral theory is also key, in particular, reference \cite{Mo} introduced a powerful quantitative technique known as enlargement of the spectral functional space. 

The contributions of this document in this topic is given, for hard potentials, in terms of the requirement of the scattering angle kernel and the generality of the initial data.  In particular, we will only assume
\begin{equation}\label{Initial-Data}
\text{Initial data:}\qquad\int_{\mathbb{R}^{d}}f_0\langle v \rangle^{2}\text{d}v<\infty,\qquad \int_{\mathbb{R}^{d}} f_0\ln(f_0)\text{d}v<\infty\,.
\end{equation}
\begin{equation}\label{Grad-cutoff-lower}
\text{Angular kernel:}\qquad b\in L^{1}(\mathbb{S}^{d-1})\,,\qquad b \geq b_o>0\,,\quad\text{ and }
\end{equation}
{We note that the assumption $b \ge b_o>0$ will be needed only for the entropy dissipation estimate in Theorem \ref{Dissipation-entropy} to fulfil the condition \eqref{K_B}
\begin{equation*}
B(x, y) = x^{\gamma}\,b(y) \geq x^{\gamma} b_o>0\,,
\end{equation*}
that is, in our particular case $K_B=b_o$ in that theorem.  This condition is not needed in previous sections.}  The strategy follows the entropic methods.  One of main contributions is the relaxation of the result given in \cite{V} with respect to the dissipation of entropy for hard potentials.  This will allow us to prescind of the decomposition theorem.  This is important since the  most general proof of the decomposition theorem available requires $b\in L^{2}(\mathbb{S}^{d-1})$, see \cite{MV}.  Once this is achieved, the result will follow after a fine decomposition of the Boltzmann linearized operator and the spectral enlargement result given in \cite{GMM}.  As an application of this convergence result, we prove the decomposition theorem under just assumptions \eqref{Initial-Data} and \eqref{Grad-cutoff-lower}.

\subsection{Outline of the paper}
In Section 2, we list several lemmas that will be used throughout the manuscript. They include the lower bound of the loss operator and the upper bound on the gain operator corresponding to the parts $b_1^\varepsilon$ and $b_2^\varepsilon$  of the angular kernel.  We state a result on the $L^\infty$-norm propagation for solutions which is based on these two estimates and the classical $L^{2}$-norm propagation.

Section 3 concerns the propagation and generation of exponentially weighted $L^p$ norms  of the solution to the homogeneous Boltzmann equation.    The case of hard potentials and the case of Maxwell molecules are studied separately. The latter one requires the \textit{a priori} boundedness of the entropy. We remark that the generation results holds only for hard potentials.

In  Section 4, we prove exponentially-tailed Sobolev regularity propagation for solutions of the homogeneous Boltzmann equation.  Propagation of regularity has been discussed before in \cite{MV, ATh}.  One of the central ingredients for the proof of propagation of Sobolev norms is the estimate \cite[Theorem 2.1]{BD} which is valid for any angular transition under the assumption  $b\in L^{2}(\mathbb{S}^{d-1})$.  This result was used in \cite{MV} to prove propagation of regularity in the case of hard potentials with polynomial weights and later, in \cite{ATh} for exponential weights.  Three extensions are given in this section with respect to \cite{MV, ATh} enable us, 

 1) to relax the assumption on the angular transition to be  merely $b\in L^{1}(\mathbb{S}^{d-1})$, 
 
 2) to able to prove propagation of fractional regularity by introducing pointwise/strong commutators, and 
 
 3) to develop regularity for Maxwell molecules cases,  in the context of exponential tails.  
 
 Although our main goal is to prove all previous extensions in the context of exponential weights in the spirit of \cite{ATh}, similar results follow with the same line of reasoning for polynomial weights of any order.

Finally, in Section 5 we show the exponentially fast convergence towards thermodynamical equilibrium and prove, as a corollary, a decomposition theorem for propagation of smoothness and roughness in the case of hard potentials.  The effort is focused in two fronts: 
\vspace{6pt}

1) provide a proof in the context of integrable scattering kernel hypothesis, and \vspace{6pt}

2) work with general physical data, that is,
 initial data having only finite mass, energy, and entropy.  
\vspace{6pt}

These results generalize classical references in the topic such as \cite{A1,Mo,MV,Wennberg-stability} at the level of the model and the initial data.  An entropic method \cite{V} together with the technique of spectral space enlargement \cite{Mo,GMM} will lead to the desired results.  In contrast to the usual argument made in the proof of exponential convergence that uses a decomposition theorem first and then an entropy method, our argument eliminates the need of the decomposition theorem.  This considerably reduces the technicalities here and in other contexts where entropic methods are used in Boltzmann-like equations.   
\subsection{Notation} We work with classical Lebesgue spaces $L^{p}(\mathbb{R}^{d})$ for $p\in[1,\infty]$.  The addition of polynomial or exponential weights are central throughout the manuscript.  No particular notation will be used, however, in some places we adopt the following standard notation for convenience
\begin{align*}
L^{p}_{\mu}(\mathbb{R}^{d}):&=\Big\{f\;\text{measurable}\,\big|\, \|f\|_{L^{p}_{\mu}(\mathbb{R}^{d})}:=\|f\langle\cdot\rangle^{\mu}\|_{L^{p}(\mathbb{R}^{d})}<+\infty\Big\}\,,\;\;\text{and}\\
L^{p}_{E\alpha,r}(\mathbb{R}^{d}):&=\Big\{f\;\text{measurable}\,\big|\, \|f\|_{L^{p}_{E\alpha,r}(\mathbb{R}^{d})}:=\|f\,e^{r\langle \cdot\rangle^{\alpha}}\|_{L^{p}(\mathbb{R}^{d})}<+\infty\Big\}\,,
\end{align*} 
for some $\mu\geq0$, $r>0$, and $\alpha>0$.  We restrict ourself to the Sobolev spaces $H^{k}(\mathbb{R}^{d})$, with $k\geq0$, defined as
\begin{equation*}
H^{k}(\mathbb{R}^{d}):=\Big\{f\in L^{2}(\mathbb{R}^{d})\,\big|\, \big\|\big(1+(-\Delta)\big)^{\frac{k}{2}}f\big\|_{L^{2}(\mathbb{R}^{d})}<+\infty\Big\}\,.
\end{equation*}
Here, the operator $\big(1+(-\Delta)\big)^{\frac{k}{2}}$ is defined using the Fourier transform $\mathcal{F}$,
\begin{equation*}
\mathcal{F}\Big\{\big(1+(-\Delta)\big)^{\frac{k}{2}}f\Big\}(\xi) = \langle \xi \rangle^{k}\mathcal{F}\big\{f\big\}(\xi)\,,
\end{equation*}
where, we recall, the brackets stand for $\langle \xi \rangle := \sqrt{1+|\xi|^{2}}$.  Polynomial and exponential weights will also be used in these spaces,
\begin{equation*}
H^{k}_{\mu}(\mathbb{R}^{d}):=\Big\{f\in H^{k}(\mathbb{R}^{d})\,\big|\, \big\|\langle \cdot\rangle^{\mu} \big(1+(-\Delta)\big)^{\frac{k}{2}}f\big\|_{L^{2}(\mathbb{R}^{d})}<+\infty\Big\}\,,\;\mu\geq0,
\end{equation*}
and,
\begin{equation*}
H^{k}_{exp}(\mathbb{R}^{d}):=\Big\{f\in H^{k}(\mathbb{R}^{d})\,\big|\, \big\|e^{r\langle\cdot\rangle^{\alpha}}\big(1+(-\Delta)\big)^{\frac{k}{2}}f\big\|_{L^{2}(\mathbb{R}^{d})}<+\infty\Big\}\,,
\end{equation*}
for some positive $r$ and $\alpha$.  These are the spaces of functions enjoying Sobolev regularity with polynomial and exponential tails respectively.  Observe that weights were chosen to be outside the differentiation operator.  Special care has to be made with this choice when fractional differentiation is performed in the particular case of exponential weights.  Finally, we will commonly use the norm
\begin{equation*}
\| \cdot \|_{(L^{p}\cap L^{q})(\mathbb{R}^{d})} := \max\big\{ \| \cdot \|_{L^{p}(\mathbb{R}^{d})},\| \cdot \|_{L^{q}(\mathbb{R}^{d})}\big\}
\end{equation*}
for the intersection space $L^{p}(\mathbb{R}^{d})\cap L^{q}(\mathbb{R}^{d})$, with $p,\,q\in[1,\infty]$.

\section{Important lemmas and $L^{\infty}$-propagation}
We begin this section with a classical inequality used to obtain a lower estimate for the loss operator $Q^{-}$.  This estimate is well know in the community when entropy is assumed to be finite, see for example \cite{A}, and it has been used in a central way in the analysis of moments \cite{WennbergMP} and the propagation of $L^{p}$-norms \cite{MV}.  The case where no finite entropy is assumed is presented in \cite{GPV}, however, the proof is strongly based in the fact that $f(t,v)$ solves the Boltzmann equation.  Refer to \cite[Lemma 27]{AGbams} and \cite[Lemma 4.2]{AGH} for an elementary proof based only on functional arguments.

\begin{lemma}[Lower bound]\label{APq-}
Fix $\gamma\in[0,2]$, and assume $0\leq \{f(t)\}_{t\geq0}\subset L^{1}_{2}(\mathbb{R}^{d})$ satisfies
\begin{equation*}
C\geq \int_{\mathbb{R}^{d}}f(t,v)dv \geq c\,,\quad  C\geq\int_{\mathbb{R}^{d}}f(t,v)|v|^{2}dv \geq c\,,\quad \int_{\mathbb{R}^{d}}f(t,v)\,v\,dv=0\,,
\end{equation*}
for some positive constants $C$ and $c$.  Assume also the boundedness of some $2^{+}$ moment
\begin{equation*}
\int_{\mathbb{R}^{d}}f(t,v)|v|^{2^{+}}dv\leq B.
\end{equation*}
Then, there exists $c_o:=c_o(B,C,c)>0$ such that
\begin{equation*}
\big(f(t,\cdot)\ast|\cdot|^{\gamma}\big)(v)\geq c_o\langle v \rangle^{\gamma}\,.
\end{equation*}
\end{lemma}
\begin{remark}
The symbol $a^{\pm}$, with $a>0$, denotes a fixed real number bigger (+) or smaller (-) than $a$.
\end{remark}

The following lemma provides estimates on the gain operator, which are essential for proving propagation of $L^p$-exponential tails in Theorem \ref{T1}. 
The novel argument in Theorem \ref{T1} is for $p=\infty$ and relies on weighted $L^1$ and $L^2$ propagation results and  on the Young's inequality  (see    Theorem \ref{ApT-1}). In order to ensure that the constant in the Young's inequality is finite, one needs to take special care of the singular point close to grazing collisions (see Remark \ref{remark on C}). This is why the decomposition \eqref{DAK} of the kernel is introduced. In the lemma below, the singular part (with $b^{\varepsilon}_{2}$) is estimated in terms of weighed $L^1$ norms. Such estimate produces a finite Young's inequality constant, which also can be made small for small $\varepsilon$. The remainder part (with $b^{\varepsilon}_{1}$) is estimated in terms of weighed $L^2$ norms which are  propagated.
 
\begin{lemma}[Upper bound]\label{APq+} Write $b=b^{\varepsilon}_{1}+b^{\varepsilon}_{2}$ as in \eqref{DAK}.  Then, for every $\gamma\geq0$,
\begin{align}\label{upper}
\begin{split}
\| Q^{+}_{o,b^{\varepsilon}_{1} }\big(f,f\langle\cdot\rangle^{\gamma}\big) \|_{L^{\infty}(\mathbb{R}^{d})} &\leq \varepsilon^{-\frac{d}{2}}\,C(b) \|f\|_{L^{2}(\mathbb{R}^{d})}\| f \langle \cdot \rangle^{\gamma}\|_{L^{2}(\mathbb{R}^{d})}\,,\\
\| Q^{+}_{o,b^{\varepsilon}_{2} }\big(f,f\langle\cdot\rangle^{\gamma}\big) \|_{L^{\infty}(\mathbb{R}^{d})} &\leq \mathfrak{m}(b^{\epsilon}_{2})\| f \|_{L^{\infty}(\mathbb{R}^{d})}\| f \langle \cdot \rangle^{\gamma} \|_{L^{1}(\mathbb{R}^{d})}\,.
\end{split}
\end{align}
The constants are such that $C(b)\sim\|b\|_{L^{1}(\mathbb{S}^{d-1})}$ and $\mathfrak{m}(b^{\epsilon}_{2})\sim\|b^{\varepsilon}_{2}\|_{L^{1}(\mathbb{S}^{d-1})}$.  In particular, $\lim_{\varepsilon\rightarrow0}\mathfrak{m}(b^{\epsilon}_{2})=0$.
\end{lemma}
\begin{proof}
Both estimates are a direct consequence of Young's inequality for the gain part of the collision operator \cite[Theorem 1]{ACG} (see Theorem \ref{ApT-1} in the appendix for a clear statement of such theorem).  Indeed, recalling that $b$ has support in $[0,1]$, we use for the first estimate in \eqref{upper} the case $(p,q,r)=(2,2,\infty)$ with constant \eqref{C-1}
\begin{align*}
C&=K\left(\int^{1}_{0}\Big(\frac{1-s}{2}\Big)^{-\frac{d}{2}}\big(1-s^{2}\big)^{\frac{d-3}{2}}b^{\varepsilon}_{1}(s)\text{d}s\right)^{\frac{1}{2}}\left(\int^{1}_{0}\Big(\frac{1+s}{2}\Big)^{-\frac{d}{2}}\big(1-s^{2}\big)^{\frac{d-3}{2}}b^{\varepsilon}_{1}(s)\text{d}s\right)^{\frac{1}{2}}\\
&\leq \varepsilon^{-\frac{d}{2}}2^{\frac{3d}{4}}K\int^{1}_{0}\big(1-s^{2}\big)^{\frac{d-3}{2}}b(s)\text{d}s=\varepsilon^{-\frac{d}{2} }2^{\frac{3d}{4}+\gamma+2}\|b\|_{L^{1}(\mathbb{S}^{d-1})}=:\varepsilon^{-\frac{d}{2}}C(b),
\end{align*}
where to obtain the inequality we used the fact that $(1-s)^{-1}\leq 2 \varepsilon^{-2}$ in the support of $b^{\varepsilon}_{1}\leq b$.  

For the second estimate we use the case the case $(p,q,r)=(\infty,1,\infty)$ with constant \eqref{C-2} 
\begin{align*}
C=K\int^{1}_{0}\Big(\frac{1+s}{2}\Big)^{-\frac{d}{2}}\big(1-s^{2}\big)^{\frac{d-3}{2}}b^{\varepsilon}_{2}(s)\text{d}s\leq 2^{\frac{d}{2}}K\int^{1}_{0}\big(1-s^{2}\big)^{\frac{d-3}{2}}b^{\varepsilon}_{2}(s)\text{d}s =: \mathfrak{m}(b^{\epsilon}_{2})\,.
\end{align*}
The fact that $\lim_{\varepsilon\rightarrow0}\mathfrak{m}(b^{\epsilon}_{2})=0$ is a direct consequence of the monotone convergence theorem.
\end{proof}
\begin{remark} In the sequel the symbol $\mathfrak{m}(b)$ will be interpreted, more generally, as a quantity proportional to $\|b\|_{L^{1}(\mathbb{S}^{d-1})}$ having a constant of proportionality that depends only on the dimension $d\geq2$ and $\gamma\geq0$.  Such symbol will be reserved to a quantity that will be taken sufficiently small at some point in the argument in question.
\end{remark}
A direct application of these two lemmas gives an elementary proof of the $L^{\infty}$-norm propagation for the homogeneous Boltzmann equation, see \cite[Theorem 31]{AGbams}.
\begin{theorem}[Propagation of $L^{\infty}$]\label{Tinfty}
Take $\gamma\in(0,1]$, $b\in L^{1}(\mathbb{S}^{d-1})$ be the angular kernel (with mass normalized to unity) and
\begin{equation*}
\|f_0\|_{(L^{1}_{2^{+}}\cap L^{\infty})(\mathbb{R}^{d})} = C_{o}\,,
\end{equation*} 
for some positive constant $C_{o}$.  Then, there exist constant $C(f_0)>0$ depending on $C_o$, $\gamma$ and $b$ such that
\begin{equation*}
\|f(t,\cdot)\|_{L^{\infty}(\mathbb{R}^{d})} \leq C(f_0)\,,\quad t\geq 0\,,
\end{equation*}
for the solution $f(t,v)$ of the Boltzmann equation.
\end{theorem}
\section{Propagation and generation of $L^{p}$-exponential tail integrability for cut-off Boltzmann}
In this section we study the propagation and generation of $L^{p}$-exponential tails.  The argument for the case $p\in[1,\infty)$ follows closely the standard theory used in the literature for polynomial weights.  The reasoning for the case $p=\infty$ is novel and follows the one given in previous section.  We divide the proof in the hard potentials $\gamma\in(0,2]$ and Maxwell molecules $\gamma=0$ cases as they are slightly different.  For instance, as opposed to the hard potentials model, the Maxwell molecules model does not create tail only propagates it.  We begin the section by considering the hard potentials case, one can get propagation of $L^{p}$-exponential tails for $p\in[1,\infty)$.

\subsection{Hard potential case}

For the hard potentials case, we first address the question of propagation  of $L^{p}$-exponential tails. The idea is to multiply the Boltzmann equation by $e^{r\langle\cdot\rangle^{\alpha}}$ in order to obtain a differential inequality \eqref{inequalityg1} for $g:=f\,e^{r\langle\cdot\rangle^{\alpha}}$. Then by Lemma \ref{APq-} and techniques of \cite{AGbams}  we have the following result.

\begin{theorem}[Propagation of $L^{p}$-exponential tails]\label{T1}
Let $\gamma\in(0,2]$ be the potential exponent, $b\in L^{1}(\mathbb{S}^{d-1})$ be the angular kernel (with mass normalized to unity), and
\begin{equation*}
\|f_0(\cdot)\,e^{a_{o}\langle\cdot\rangle^{\alpha}}\|_{(L^{1}\cap L^{p})(\mathbb{R}^{d})} = C_{o} <\infty \,,
\end{equation*} 
for some $\alpha\in(0,2]$, $p\in[1,\infty]$ and positive constants $a_o$ and $C_{o}$.  Then, there exist positive constants $a$ and $C$ depending on the initial mass, energy, $a_o$, $C_o$, $\gamma$ and $b$ such that
\begin{equation*}
\|f(t,\cdot)\,e^{a\langle\cdot\rangle^{\alpha}}\|_{L^{p}(\mathbb{R}^{d})} \leq C\,,\quad t\geq 0\,,
\end{equation*}
for the solution $f(t,v)$ of the Boltzmann equation.
\end{theorem}
\begin{proof}
We first recall that thanks to the propagation of moments  of \cite[Theorem 2]{ACGM} for $p=1$ we have
\begin{equation}\label{momentpropagation}
\|f(t,\cdot)\,e^{\at\langle\cdot\rangle^{\alpha}}\|_{L^{1}(\mathbb{R}^{d})} \leq C\,,\quad t\geq 0\,,
\end{equation}
for some positive $\at$ and $C$ with dependence as stated.  

Our next goal is to prove the result for $p\in(1,\infty)$.
Note that for any $\alpha\in[0,2]$
\begin{align*}
\langle v \rangle^{\alpha} = \big(1+|v|^{2}\big)^{\frac{\alpha}{2}} & \leq\big(1+ |v|^{2}+|v_{*}|^{2}\big)^{\frac{\alpha}{2}} = \big(1 + |v'|^{2} + |v'_{*}|^{2}\big)^{\frac{\alpha}{2}}\\
&\leq \big(1+|v'|^{2})^{\frac{\alpha}{2}} + |v'_{*}|^{\alpha}\leq \langle v' \rangle^{\alpha} + \langle v'_{*} \rangle^{\alpha}\,,
\end{align*}
and therefore, $e^{r\langle v\rangle^{\alpha}} \leq e^{r\langle v'\rangle^{\alpha}} e^{r\langle v'_{*}\rangle^{\alpha}}$ for any $r>0$.  As a consequence
\begin{align*}
Q(f,f)(v)\,e^{r\langle v\rangle^{\alpha}} &= Q^{+}(f,f)(v)\,e^{r\langle v\rangle^{\alpha}} - Q^{-}(f,f)(v)\,e^{r\langle v\rangle^{\alpha}} \\
& \leq Q^{+}(fe^{r\langle \cdot \rangle^{\alpha}},fe^{r\langle\cdot\rangle^{\alpha}})(v) - f(v)e^{r\langle v\rangle^{\alpha}}\,\big(f\ast|\cdot|^{\gamma}\big)(v)\,. 
\end{align*}
Thus, defining $g:=f\,e^{r\langle\cdot\rangle^{\alpha}}$, the Boltzmann equation implies that
\begin{equation}\label{inequalityg1}
\partial_{t}g(v) \leq Q^{+}(g,g)(v) - g(v)\,\big(f\ast|\cdot|^{\gamma}\big)(v)\,. 
\end{equation}
As a consequence of Lemma \ref{APq-}, estimate \eqref{inequalityg1} implies that
\begin{equation}\label{inequalityg2}
\partial_{t}g(v) \leq Q^{+}(g,g)(v) - c_o\,g(v)\langle v \rangle^{\gamma}\,.
\end{equation} 
Estimate \eqref{inequalityg2} suffices to conclude using the techniques of \cite[Theorem 4.1]{MV} or \cite[Corollary 1.1]{AG} that
\begin{equation}\label{casepfinite0}
\sup_{t\geq0}\|g(t,\cdot)\|_{L^{p}(\mathbb{R}^{d})} \leq C_{p}(g_0)\,,\quad \text{ with }\, p\in(1,\infty)\,,
\end{equation}
where $C_{p}(g_0)$ depends on an upper bound of $\|g_0\|_{L^{p}} + \sup_{t\geq0}\|g(t,\cdot)\|_{L^{1}_{2}}$.  Of course, such upper bound is finite for any $r<\min\{\at,a_o\}$ thanks to \eqref{momentpropagation} and the weighted $L^{p}$ integrability of $f_0$.  So, setting $a<\min\{\at,a_o\}$ proves the result for any $p\in(1,\infty)$.\\

\noindent
Let us prove the case $p=\infty$.  

Observe that $|u|\leq \sqrt{2}\,|u^{+}| \leq \sqrt{2}\langle v'_{*} \rangle \langle v \rangle$ by symmetrization of $b$.  Then,
\begin{equation*}
Q^{+}(f,f)(v)\leq 2^{\frac{\gamma}{2}}Q^{+}_{o,b}\big(f,f\langle\cdot\rangle^{\gamma}\big)(v)\,\langle v \rangle^{\gamma}\,.
\end{equation*}
Using estimate \eqref{inequalityg2}, we conclude that
\begin{equation}\label{estimateinf}
\partial_{t}g (v) \leq \Big(2^{\frac{\gamma}{2}}Q^{+}_{o,b}\big(g,g\langle \cdot \rangle^{\gamma}\big)(v) - c_o\,g(v)\Big)\langle v \rangle ^{\gamma}\,.
\end{equation}
Furthermore, Lemma \ref{APq+} leads to
\begin{align}\label{estimateinf1}
\begin{split}
\| Q^{+}_{o,b^{\varepsilon}_{1} }\big(g,g\langle\cdot\rangle^{\gamma}\big) \|_{L^{\infty}(\mathbb{R}^{d})} &\leq \varepsilon^{-\frac{d}{2}}\,C(b) \|g\|_{L^{2}(\mathbb{R}^{d})}\| g \langle \cdot \rangle^{\gamma}\|_{L^{2}(\mathbb{R}^{d})}\,,\\
\| Q^{+}_{o,b^{\varepsilon}_{2} }\big(g,g\langle\cdot\rangle^{\gamma}\big) \|_{L^{\infty}(\mathbb{R}^{d})} &\leq \mathfrak{m}(b^{\epsilon}_{2})\| g \|_{L^{\infty}(\mathbb{R}^{d})}\| g \langle \cdot \rangle^{\gamma} \|_{L^{1}(\mathbb{R}^{d})}\,.
\end{split}
\end{align}
In addition, estimates \eqref{momentpropagation} and \eqref{casepfinite0} ensure that
\begin{equation}\label{estimateinf2}
\sup_{t\geq0}\Big(\| g(t) \langle \cdot \rangle^{\gamma} \|_{L^{1}(\mathbb{R}^{d})} + \| g(t) \langle \cdot \rangle^{\gamma}\|_{L^{2}(\mathbb{R}^{d})}\Big) \leq C(g_0)\,, \quad \text{ for any }\, r< \min\{\at,a_o\}\,.
\end{equation}
Using again the splitting \eqref{DAK} as $Q^{+}_{o,b}=Q^{+}_{o,b^{\varepsilon}_{1}}+Q^{+}_{o,b^{\varepsilon}_{2}}$, the conclusion from the estimates \eqref{estimateinf}, \eqref{estimateinf1} and \eqref{estimateinf2} is that for some constants $C_{i}(g_0, b)$, $i=1,2$, depending on $g_0$ through the norm $\|g_0\|_{L^{1}\cap L^{2}(\mathbb{R}^{d})}$ it holds that
\begin{align}\label{estimateinf3}
\begin{split}
\partial_{t}g(t,v) &\leq \Big( \varepsilon^{-\frac{d}{2}} C_{1}(g_0) + \mathfrak{m}(b^{\epsilon}_{2})\,C_{2}(g_0)\|g(t)\|_{L^{\infty}(\mathbb{R}^{d})}\Big)\langle v \rangle^{\gamma} - c_o\,g(t,v)\langle v \rangle^{\gamma}\\
&= \bigg( \tilde{C}_{1}(g_0) + \frac{c_o}{4}\|g(t)\|_{L^{\infty}(\mathbb{R}^{d})}\bigg)\langle v \rangle^{\gamma} - c_o\,g(t,v)\langle v \rangle^{\gamma}\,,
\end{split}
\end{align}
where, for the latter, we simply took $\varepsilon>0$ sufficiently small such that $\mathfrak{m}(b^{\epsilon}_{2})C_{2}(g_0)\leq \frac{c_o}{4}$.  Finally, let us integrate estimate \eqref{estimateinf3} to conclude that
\begin{equation}\label{estimateinffinal}
\sup_{t\geq0}\|g(t)\|_{L^{\infty}(\mathbb{R}^{d})}\leq \frac{4}{3}\Big(\|g_0\|_{L^{\infty}(\mathbb{R}^{d})} +\frac{\tilde{C}_{1}(g_0)}{c_o}\Big)=:G(g_0)\,. 
\end{equation}
Indeed, bear in mind that $\partial_{t}g$ exist a.e. in $\mathbb{R}^{+}\times\mathbb{R}^{d}$ since $f$ is solution of the Boltzmann equation.  Thus, we can integrate estimate \eqref{estimateinf3} in $[0,t]$ to obtain
\begin{align*}
g(t,v) &\leq g_{0}(v)e^{-c_{o}\langle v \rangle^{\gamma} t} + \int^{t}_{0}e^{-c_o\langle v \rangle^{\gamma}(t-s)}\Big( \tilde{C}_{1}(g_0) + \frac{c_o}{4}\|g\|_{L^{\infty}(\mathbb{R}^{d})}\Big)\langle v \rangle^{\gamma}\text{d}s\\
&\hspace{-.5cm}\leq \|g_{0}\|_{L^{\infty}(\mathbb{R}^{d})} + \Big(\tilde{C}_{1}(g_0) + \frac{c_o}{4}\sup_{0\leq s \leq t}\|g(s)\|_{L^{\infty}(\mathbb{R}^{d})}\Big)\langle v \rangle^{\gamma}\int^{t}_{0}e^{-c_o\langle v \rangle^{\gamma}(t-s)}\text{d}s\\
&\leq \frac{3}{4}G(g_0)+\frac{1}{4}\sup_{0\leq s \leq t}\|g(s)\|_{L^{\infty}(\mathbb{R}^{d})}\,,\quad \text{ a.e. in } v\in\mathbb{R}^{d}\,.
\end{align*}
As a consequence, for any $0\leq t \leq T$ it holds
\begin{equation}\label{estimateinf4}
g(t,v) \leq \frac{3}{4}G(g_0)+\frac{1}{4}\sup_{0\leq s \leq T}\|g(s)\|_{L^{\infty}(\mathbb{R}^{d})}\,,\quad \text{ a.e. in } v\in\mathbb{R}^{d}\,.
\end{equation}
Estimate \eqref{estimateinffinal} readily follows after computing the essential supremum in $v\in\mathbb{R}^{d}$ and, then, the supremum in $t\in[0,T]$ in estimate \eqref{estimateinf4}.   So, setting $a<\min\{\at,a_o\}$ proves the result for $p = \infty$.
\end{proof}
It is well know that instantaneous generation of tails at the level of moments occurs for the Boltzmann equation for hard potentials.  This is at odds with the Maxwell molecules model where only propagation of tails is possible.  This was first noticed in \cite{WennbergMP}.  More recently, $L^{1}$-exponential tail generation has been originally studied in \cite{ACGM}. 

In this manuscript, we follow the techniques and results from the recent work  \cite{AGbams}  by obtaining sharper  $L^{1}_{\mu}$-polynomial and  $L^{1}_{\text{exp}}$-exponential tail estimates  induced  by  new ordinary differential inequalities needed for the Lebesgue Sobolev spaces, yielding smaller upper bounds for larger times.  

We show here that from $L^{1}_{\text{exp}}$, it is possible to deduce $L^{p}_{\text{exp}}$-exponential  tail generation following the argument presented previously for propagation of exponential integrability.  We should stress that integrability is not created only the tails are.
\begin{theorem}[Generation of $L^{p}$-tails]\label{T1/Creation}
Let $\gamma\in(0,2)$ be the potential exponent, $b\in L^{1}(\mathbb{S}^{d-1})$ be the angular kernel (with mass normalized to unity), and assume that for some $p\in(1,\infty)$
\begin{equation*}
\|f_0(\cdot)\|_{(L^{1}_{2}\cap L^{p})(\mathbb{R}^{d})} = C_{o} <\infty \,,
\end{equation*} 
for some constant $C_{o}>0$.  Then, there exist positive constants $a$ and $C$ depending on the initial mass, energy, $C_o$, $\gamma$ and $b$ such that
\begin{equation}\label{E/Creation}
\|f(t,\cdot)\,e^{a\min\{1,t\}\langle\cdot\rangle^{\gamma}}\|_{L^{p}(\mathbb{R}^{d})} \leq C\,,\quad t\geq 0\,,
\end{equation}
for the solution $f(t,v)$ of the Boltzmann equation.  This estimate is also true in the case $\gamma=2$ if a moment $2^{+}$ is assumed to be finite for $f_0$.\\\\
\noindent
For the  case $p=+\infty$, with $\gamma\in(0,2)$, assume
\begin{equation*}
\|f_0(\cdot)\|_{(L^{1}_{2}\cap L^{2}_{2} \cap L^{\infty})(\mathbb{R}^{d})} = C_{o} <\infty \,.
\end{equation*} 
Then, estimate \eqref{E/Creation} holds true.  Again, this estimate holds for $\gamma=2$ provided the $L^{1}_{2^{+}}\cap L^{2}_{2^{+}}$ norm is assumed finite for the initial datum.
\end{theorem}
\begin{proof}
Let us handle first the case $p\in(1,\infty)$.  Note that for any $r>0$ and $\gamma\in(0,2]$ it follows that
\begin{equation*}
e^{r\,\min\{1,t\}\langle v\rangle^{\gamma}}\partial_{t}f(t,v) = \partial_{t}g(t,v) - r\,\chi_{[0,1)}(t)\langle v \rangle^{\gamma} g(t,v)\,,\quad t\geq0\,,
\end{equation*}
where $g(t,v):=f(t,v)e^{r\,\min\{1,t\}\langle v\rangle^{\gamma}}$.  Thus, arguing as in the proof of Theorem \ref{T1}, it follows that
\begin{equation*}
\partial_{t}g(v) \leq Q^{+}(g,g)(v) - (c_o-r)\,g(v)\langle v \rangle^{\gamma}\,.
\end{equation*}
Moreover, in \cite[Theorem 1]{ACGM} it is proved that if $f_0\in L^{1}_{2}(\mathbb{R}^{d})$ then
\begin{equation*}
\|f(t,\cdot)\,e^{a_o\min\{1,t\}\langle\cdot\rangle^{\gamma}}\|_{L^{1}(\mathbb{R}^{d})} \leq C_o,\quad t\geq0\,,
\end{equation*}
for some positive constants $a_o$ and $C_o$ depending on the initial mass, energy, $\gamma$ and $b$.  Choosing $r<\min\{a_o,c_o\}$ sufficiently small, this estimate suffices to conclude, using the techniques of \cite[Theorem 4.1]{MV} or \cite[Corollary 1.1]{AG}, that
\begin{equation}\label{casepfinite}
\sup_{t\geq0}\|g(t,\cdot)\|_{L^{p}(\mathbb{R}^{d})} \leq C_{p}(g_0)\,,\quad \text{ with }\, \gamma\in(0,2]\,,
\end{equation}
where $C_{p}(g_0)$ depends on an upper bound of $\|g_0\|_{L^{p}} + \sup_{t\geq0}\|g(t,\cdot)\|_{L^{1}_{\gamma^{+}}}<+\infty$. This concludes case $p\in(0,\infty)$ by choosing $a<\min\{a_o,c_o\}$ sufficiently small.\\

\noindent The case $p=\infty$ is obtained following the respective argument of Theorem \ref{T1} for this case.   Indeed, one easily arrives to the estimate
\begin{equation*}
\partial_{t}g (v) \leq \Big(2^{\frac{\gamma}{2}}Q^{+}_{o}\big(g,g\langle \cdot \rangle^{\gamma}\big)(v) - (c_o-r)\,g(v)\Big)\langle v \rangle ^{\gamma}\,.
\end{equation*}
We will choose $r=a$ with $a<c_o$ sufficiently small, thus, we only need to guarantee the finiteness of the term
\begin{equation*}
\sup_{t\geq0}\Big(\big\|g(t)\langle \cdot \rangle^{\gamma} \|_{L^{1}} + \big\|g(t)\langle \cdot \rangle^{\gamma} \|_{L^{2}}\Big)\,.
\end{equation*}
This holds for any $\gamma\in(0,2)$ (recall that $f_0\in L^{1}_{2}\cap L^{2}_{2}$) due to the propagation of polynomial integrability \cite[Theorem 4.1]{MV} and generation of exponential integrability Theorem \ref{T1/Creation} for the case $p=2$.  The case $\gamma=2$ holds true by assuming $f_0\in L^{1}_{2^{+}}\cap L^{2}_{2^{+}}$ and invoking the same theorems. 
\end{proof}
\subsection{Maxwell molecules case}
Maxwell molecules, $\gamma=0$, is a critical case for uniform propagation of $L^{p}$ integrability for general initial data.  Indeed, as soon as $\gamma<0$, i.e. soft potential case, the uniform propagation of moments is lost for general initial data, refer to \cite{CCL} for an interesting discussion.  In order to compensate for this issue, we will use propagation of entropy.  Thus, we implicitly have the additional requirement on the initial data
\begin{equation*}
\int_{\mathbb{R}^{d}}f_0(v)\ln\big(f_0(v)\big)\text{d}v<\infty\,,
\end{equation*}
that is, we work with initial data having finite initial entropy.  Clearly, this is harmless since more restrictive conditions on $f_0$ are imposed in our context, namely, $f_0\in \big(L^{1}_{2}\cap L^{p}\big)(\mathbb{R}^{d})$, for $p>1$.  As a consequence, initial  entropy is finite since for any $1<p<\frac{d}{d-1}\leq 2$
\begin{align*}
\bigg|\int_{\mathbb{R}^{d}} &f_0(v)\ln\big(f_0(v)\big)\text{d}v \bigg|  \leq \int_{\{f_0\leq1\}} \big|f_0(v)\ln\big(f_0(v)\big)\big|\text{d}v  +  \int_{\{f_0\geq1\}} |f_0(v)\ln\big(f_0(v)\big)|\text{d}v\\
& \leq C_{p}\int_{\{f_0\leq1\}} \big|f_0(v)\big|^{\frac{1}{p}}\text{d}v + \int_{\{f_0\geq1\}} |f_0(v)|^{p}\text{d}v\\
& \leq C_{p,d}\bigg(\int_{\{f_0\leq1\}} f_0(v)\langle v \rangle^{p}\text{d}v\bigg)^{\frac{1}{p}} + \int_{\{f_0\geq1\}} |f_0(v)|^{p}\text{d}v =: C\Big(d,\|f_0\|_{(L^{1}_{2}\cap L^{p})(\mathbb{R}^{d})}\Big)<+\infty.
\end{align*}
In the second estimate we used that $|x\ln(x)| \leq C_{p}\,x^{\frac{1}{p}}$, for any $p>1$ and $x\in[0,1]$.  Furthermore, it is well known that using energy conservation and entropy dissipation it follows for $f(t,v)$, the solution of the homogeneous Boltzmann equation (see \cite[page 329]{DL} or more recently \cite[Lemma A.1]{ALods}), that
\begin{equation}\label{entropycontrol}
\sup_{t\geq0}\int_{\mathbb{R}^{d}}f(t,v)\big|\ln\big(f(t,v)\big)\big|\text{d}v \leq C\bigg(\int f_0\ln(f_0),\int f_0| \cdot |^{2}\bigg)\,.
\end{equation}
We begin this section by proving an $L^p$ estimate on the gain operator for Maxwell molecules by again relying on the Young's inequality.
\begin{lemma}\label{L1MM}
For any $K>1$, $\varepsilon>0$ and $p\in[1,\infty]$ one has
\begin{align*}
\|Q^{+}_{o,b}(f,f)\|_{ L^{p}(\mathbb{R}^{d}) } & \leq \frac{ C(b) }{ \ln(K) }  \|f \ln(f) \|_{L^{1}(\mathbb{R}^{d}) } \|f\|_{L^{p}(\mathbb{R}^{d})} \\
&\hspace{-1cm} + \varepsilon^{-\frac{d}{2p'}}\,K^{\frac{1}{2p'}}C(b)\|f\|^{1+\frac{1}{2p}}_{L^{1}(\mathbb{R}^{d}) }\|f\|^{\frac{1}{2}}_{L^{p}(\mathbb{R}^{d})} + \mathfrak{m}(b^{\epsilon}_{2})\|f\|_{L^{1}(\mathbb{R}^{d}) } \|f\|_{L^{p}(\mathbb{R}^{d})}\,.
\end{align*}
\end{lemma}
\begin{proof}
Use the usual decomposition \eqref{DAK} to write $Q^{+}_{o,b}(f,f)= Q^{+}_{o,b^{\varepsilon}_{1}}(f,f) + Q^{+}_{o,b^{\varepsilon}_{2}}(f,f)$.  On the one hand,  from \cite[Theorem 1]{ACG} or Theorem \ref{ApT-1} we know that
\begin{equation*}
\| Q^{+}_{o,b^{\varepsilon}_{2}}(f,f) \|_{L^{p}(\mathbb{R}^{d})} \leq \mathfrak{m}(b^{\epsilon}_{2})\| f \|_{L^{p}(\mathbb{R}^{d})} \|f\|_{L^{1}(\mathbb{R}^{d})}\,.
\end{equation*}
On the other hand, note that bilinearity implies
\begin{align*}
Q^{+}_{o,b^{\varepsilon}_{1}}(f,f) &= Q^{+}_{o,b^{\varepsilon}_{1}}(f,f\,\text{1}_{\{f\leq K\}}) + Q^{+}_{o,b^{\varepsilon}_{1}}(f,f\,\text{1}_{\{f > K\}})\\
&\leq Q^{+}_{o,b^{\varepsilon}_{1}}(f,f\,\text{1}_{\{f\leq K\}}) + \ln(K)^{-1}\,Q^{+}_{o,b^{\varepsilon}_{1}}(f,f\,\ln(f))\,.
\end{align*}
Moreover, Young's inequality for the gain collision operator, see \cite[Theorem 1]{ACG} or Theorem \ref{ApT-1}, implies
\begin{align*}
\| Q^{+}_{o,b^{\varepsilon}_{1}}(f,f\,\ln(f)) \|_{L^{p}(\mathbb{R}^{d})} &\leq C(b)\| f \|_{L^{p}(\mathbb{R}^{d})} \|f\,\ln(f)\|_{L^{1}(\mathbb{R}^{d})}\,,\\
\| Q^{+}_{o,b^{\varepsilon}_{1}}(f,f\,\text{1}_{\{f\leq K\}}) \|_{L^{p}(\mathbb{R}^{d})} &\leq \varepsilon^{-\frac{d}{2p'}}C(b)\|f\|_{L^{\frac{2p}{p+1}}(\mathbb{R}^{d})}\|f\,\text{1}_{\{f\leq K\}}\|_{L^{\frac{2p}{p+1}}(\mathbb{R}^{d})}\\
&\leq \varepsilon^{-\frac{d}{2p'}}\,K^{\frac{1}{2p'}}\,C(b)\|f\|_{L^{\frac{2p}{p+1}}(\mathbb{R}^{d})}\|f\|^{\frac{p+1}{2p}}_{L^{1}(\mathbb{R}^{d})}\,.
\end{align*}
The estimate follows using, in the last inequality, the interpolation
\begin{equation}\label{lasteq}
\|f\|_{L^{\frac{2p}{p+1}}(\mathbb{R}^{d})}\leq\|f\|^{\frac{1}{2}}_{L^{1}(\mathbb{R}^{d})}\|f\|^{\frac{1}{2}}_{L^{p}(\mathbb{R}^{d})}\,.
\end{equation}
\end{proof}
Thanks to the previous lemma, we next prove propagation of $L^p$ norms without weights.
\begin{proposition}\label{P1MM}
Consider Maxwell molecules $\gamma=0$ with angular kernel $b\in L^{1}(\mathbb{S}^{d-1})$ (with mass normalized to unity).  Assume that the initial data satisfies
\begin{equation*}
\|f_0(\cdot)\|_{(L^{1}_{2}\cap L^{p})(\mathbb{R}^{d})} <\infty\,,
\end{equation*} 
for some $p\in[1,\infty]$.  Then, there exist positive constant $C$ depending on the initial mass, energy, entropy, $\|f_0\|_{L^{p}(\mathbb{R}^{d})}$, $\gamma$ and $b$ such that
\begin{equation*}
\|f(t,\cdot)\|_{L^{p}(\mathbb{R}^{d})} \leq C\,,\quad t\geq 0\,,
\end{equation*}
for the solution $f(t,v)$ of the Boltzmann equation.
\end{proposition}
\begin{proof}
Without loss of generality assume $\|f(t,\cdot)\|_{L^{1}(\mathbb{R}^{d})}=1$.  The case $p\in[1,\infty)$ is a direct consequence of Lemma \ref{L1MM}.  Indeed, multiply the Boltzmann equation by $f^{p-1}$ and integrate in velocity to obtain
\begin{align*}
\frac{1}{p}\frac{\text{d}}{\text{d}t}\|f\|^{p}_{L^{p}(\mathbb{R}^{d})} = \int_{\mathbb{R}^{d}}&Q^{+}(f,f)(v)f^{p-1}(v)\text{d}v - \|f\|^{p}_{L^{p}(\mathbb{R}^{d})} \\
& \leq \|Q^{+}(f,f)\|_{L^{p}(\mathbb{R}^{d})}\|f\|^{p-1}_{L^{p}(\mathbb{R}^{d})} - \|f\|^{p}_{L^{p}(\mathbb{R}^{d})}\,.
\end{align*}
Taking $K>1$ sufficiently large and $\varepsilon>0$ sufficiently small in Lemma \ref{L1MM}, we simply obtain that
\begin{equation*}
\frac{1}{p}\frac{\text{d}}{\text{d}t}\|f\|^{p}_{L^{p}(\mathbb{R}^{d})} \leq C(f_0)\|f\|^{p-\frac{1}{2}}_{L^{p}(\mathbb{R}^{d})} - \frac{1}{2}\|f\|^{p}_{L^{p}(\mathbb{R}^{d})}\,,
\end{equation*}
where the cumulative constant $C(f_0)$ depends on mass, energy, entropy and scattering kernel $b$.  This is enough to conclude that
\begin{equation}\label{estimateLpmax}
\sup_{t\geq0}\|f(t)\|_{L^{p}(\mathbb{R}^{d})}\leq \max\big\{ \|f_0\|_{L^p(\mathbb{R}^{d})}, 4C(f_0)^{2} \big\}\,.
\end{equation}
The case $p=\infty$ is straightforward from here.  Indeed, using \eqref{estimateinf1} which is valid for $\gamma=0$, it follows that
\begin{align*}
\partial_{t}f(v) &\leq \varepsilon^{-\frac{d}{2}}C(b)\|f\|^{2}_{L^{2}(\mathbb{R}^{d})} + \mathfrak{m}(b^{\epsilon}_{2})\|f\|_{L^{1}(\mathbb{R}^{d})}\|f\|_{L^{\infty}(\mathbb{R}^{d})} - f(v) \\
&\leq C(f_0) + \frac{1}{4}\|f\|_{L^{\infty}(\mathbb{R}^{d})} - f(v)\,.
\end{align*}
We argue as in Theorem \ref{T1} to conclude that 
\begin{equation*}
\sup_{t\geq 0}\|f(t)\|_{L^{\infty}(\mathbb{R}^{d})} \leq \frac{4}{3}\Big(\|f_0\|_{L^{\infty}(\mathbb{R}^{d})} + C(f_0)\Big)\,.
\end{equation*}
\end{proof}
\begin{remark}  Since Lemma \ref{L1MM} is valid for $p=\infty$, estimate \eqref{estimateLpmax} does not degenerate as $p\rightarrow\infty$.  Thus, the $L^{\infty}$-estimate can be obtained by simply sending $p\rightarrow\infty$ in \eqref{estimateLpmax}.  This is at odds with the usual estimates for hard potentials which use interpolation and degenerate as $p$ increases to infinity.  The explanation for  such difference is that Lemma \ref{L1MM} explicitly uses the propagation of entropy which was avoided in the context of hard potentials.  This approach of using the entropy  functional leads to an alternative argument for propagation of $L^{p}$-norms for hard potentials as long as propagation of entropic moments is at hand.  Such propagation of entropic moments  will be shown in the last section.
\end{remark}
We next work towards adding exponential weights. To that end, we first obtain an $L^p$ estimate on the gain operator evaluated at $f \, e^{r\langle\cdot\rangle^{\alpha}}$.

\begin{lemma}\label{L2MM}
Define $g:=f \, e^{r\langle\cdot\rangle^{\alpha}}$, and set $a>0$, $r\in(0,a)$, $\alpha\in (0,2]$.  Then, for any $R>0$, $\varepsilon>0$ and $p\in[1,\infty]$ it holds that
\begin{align*}
\|Q^{+}(g,g)\|_{L^{p}(\mathbb{R}^{d})}  & \leq e^{ -(a-r)R^{\alpha} }C(b)\|f\,e^{ a\langle \cdot\rangle^{\alpha} }\|_{ L^{1}(\mathbb{R}^{d}) }\|g\|_{ L^{p}(\mathbb{R}^{d}) }\\
&\hspace{-.6cm} + e^{rR^{\alpha}}\varepsilon^{-\frac{d}{2p'}}C(b)\sqrt{\|f\|_{1}\|f\|_{p}\|g\|_{1}\|g\|_{p}} + \mathfrak{m}(b^{\epsilon}_{2})\|g\|_{L^{1}(\mathbb{R}^{d})}\|g\|_{L^{p}(\mathbb{R}^{d})}\,.
\end{align*}
\end{lemma}
\begin{proof}
Note that for any $R>0$ we can write
\begin{equation*}
Q^{+}(g,g) = Q^{+}(g,g\,\text{1}_{|\cdot| \leq R}) + Q^{+}(g,g\,\text{1}_{|\cdot| > R})\,.
\end{equation*}
On the one hand, for the latter one simply estimates
\begin{align*}
\|Q^{+}(g,g\,\text{1}_{|\cdot| > R})\|_{L^{p}(\mathbb{R}^{d})} &\leq C(b)\|g\|_{L^{p}(\mathbb{R}^{d})} \|g\,\text{1}_{|\cdot| > R}\|_{L^{1}(\mathbb{R}^{d})}\\
& \leq e^{-(a-r)R^{\alpha}}C(b)\|g\|_{L^{p}(\mathbb{R}^{d})} \|f\,e^{a\langle\cdot\rangle^{\alpha}}\|_{L^{1}(\mathbb{R}^{d})}\,.
\end{align*}
On the other hand, for the former one uses the decomposition \eqref{DAK}
\begin{equation*}
Q^{+}(g,g\,\text{1}_{|\cdot| \leq R}) = Q^{+}_{b^{1}_{\varepsilon}}(g,g\,\text{1}_{|\cdot| \leq R})+Q^{+}_{b^{2}_{\varepsilon}}(g,g\,\text{1}_{|\cdot| \leq R})\,.
\end{equation*}
Each of the terms on the right side is easily controlled by the Young's inequality for the gain collision operator.  Indeed, for the operator with $b^{1}_{\varepsilon}$
\begin{align*}
\|Q^{+}_{b^{1}_{\varepsilon}}(g,g\,\text{1}_{|\cdot| \leq R})\|_{L^{p}(\mathbb{R}^{d})} &\leq \varepsilon^{-\frac{d}{2p'}} C(b) \|g\|_{L^{\frac{2p}{p+1}}(\mathbb{R}^{d})}\|g\,\text{1}_{|\cdot|\leq R}\|_{L^{\frac{2p}{p+1}}(\mathbb{R}^{d})}\\
& \leq e^{rR^{\alpha}}\varepsilon^{-\frac{d}{2p'}} C(b) \|g\|_{L^{\frac{2p}{p+1}}(\mathbb{R}^{d})}\|f\|_{L^{\frac{2p}{p+1}}(\mathbb{R}^{d})}\\
&\hspace{-1cm} \leq e^{rR^{\alpha}}\varepsilon^{-\frac{d}{2p'}} C(b) \sqrt{\|g\|_{L^{1}(\mathbb{R}^{d})}\|g\|_{L^{p}(\mathbb{R}^{d})}\|f\|_{L^{1}(\mathbb{R}^{d})}\|f\|_{L^{p}(\mathbb{R}^{d})}}\,,
\end{align*}
where the last inequality follows from \eqref{lasteq}.  Similarly, for the operator with $b^{2}_{\varepsilon}$
\begin{align*}
\|Q^{+}_{b^{2}_{\varepsilon}}(g,g\,\text{1}_{|\cdot| \leq R})\|_{L^{p}(\mathbb{R}^{d})} \leq \mathfrak{m}(b^{\epsilon}_{2})\|g\|_{L^{p}(\mathbb{R}^{d})}\|g\,\text{1}_{|\cdot| \leq R}\|_{L^{1}(\mathbb{R}^{d})} \leq \mathfrak{m}(b^{\epsilon}_{2})\|g\|_{L^{p}(\mathbb{R}^{d})}\|g\|_{L^{1}(\mathbb{R}^{d})}\,.
\end{align*}
The result follows after gathering these estimates.
\end{proof}
Finally, we prove propagation of exponentially weighted $L^p$ norms for the solution to the homogeneous Boltzmann equation for Maxwell molecules.
\begin{theorem}\label{T2}
Consider Maxwell molecules $\gamma=0$, let $b\in L^{1}(\mathbb{S}^{d-1})$ be the angular kernel (with mass normalized to unity) and suppose
\begin{equation*}
\|f_0(\cdot)\,e^{a_{o}\langle\cdot\rangle^{\alpha}}\|_{(L^{1}\cap L^{p})(\mathbb{R}^{d})} = C_{o}<\infty\,,
\end{equation*} 
for some $\alpha\in(0,2]$, $p\in[1,\infty]$ and positive constants $a_o$ and $C_{o}$.  Then, there exist positive constants $a$ and $C$ depending on the initial mass, energy, entropy, $a_o$, $C_o$ and $b$ such that
\begin{equation*}
\|f(t,\cdot)\,e^{a\langle\cdot\rangle^{\alpha}}\|_{L^{p}(\mathbb{R}^{d})} \leq C\,,\quad t\geq 0\,,
\end{equation*}
for the solution $f(t,v)$ of the Boltzmann equation.
\end{theorem}
\begin{proof}
Again, our first step consists in noticing the propagation of exponential moments for Maxwell molecules \cite[Theorem 4.1]{MM} 
\begin{equation}\label{momentpropagationmax}
\|f(t,\cdot)\,e^{\at \langle\cdot\rangle^{\alpha}}\|_{L^{1}(\mathbb{R}^{d})} \leq C\,,\quad t\geq 0\,,
\end{equation}
for some positive $\at$ and $C$ with dependence as stated.  Now, following the notation and argument of Theorem \ref{T1}, we arrive to the equivalent of equation \eqref{inequalityg2}
\begin{equation}\label{inequalityg2max}
\partial_{t}g(v)\leq Q^{+}(g,g)(v) - g(v)\,,\quad \text{recall that } g:=fe^{a\langle \cdot\rangle^{\alpha}}.
\end{equation}  
After multiplying \eqref{inequalityg2max} by $g^{p-1}$ and integrating in velocity one concludes, as usual, that
\begin{equation}\label{maxg3}
\frac{1}{p}\frac{\text{d}}{\text{d}t}\|g\|_{L^{p}(\mathbb{R}^{d})} \leq \|Q^{+}(g,g)\|_{L^{p}(\mathbb{R}^{d})}\|g\|^{p-1}_{L^{p}(\mathbb{R}^{d})} - \|g\|^{p}_{L^{p}(\mathbb{R}^{d})}\,.
\end{equation}
Furthermore, the norms $\|f\,e^{r\langle\cdot\rangle^{\alpha}}\|_{L^{1}(\mathbb{R}^{d})}$, with $r\in(0,a]$, are uniformly bounded by \eqref{momentpropagationmax}, and the norm $\|f\|_{L^{p}(\mathbb{R}^{d})}$ is uniformly bounded by Proposition \ref{P1MM}.  As a consequence, using Lemma \ref{L2MM} for $R>0$ sufficiently large and $\varepsilon>0$ sufficiently small, we can find a cumulative constant depending only on $g_0$ such that
\begin{equation}\label{maxg4}
\|Q^{+}(g,g)\|_{L^{p}(\mathbb{R}^{d})} \leq C(g_0)\|g\|^{\frac{1}{2}}_{L^{p}(\mathbb{R}^{d})} + \frac{1}{2} \|g\|_{L^{p}(\mathbb{R}^{d})}\,.
\end{equation}
From \eqref{maxg3} and \eqref{maxg4} it readily follows that
\begin{equation}\label{maxfinal}
\sup_{t\geq0}\|g(t)\|_{L^{p}(\mathbb{R}^{d})} \leq \max\big\{ \|g_0\|_{L^{p}(\mathbb{R}^{d})}\, ,\, 4C(g_0)^{2} \big\}\,.
\end{equation}
Finally, the case $p=\infty$ goes in the same manner as in Theorem \ref{T1}, or simply by taking $p\rightarrow\infty$ in \eqref{maxfinal}.
\end{proof}

\section{Propagation of exponentially-tailed Sobolev regularity}

Key ingredients for proving the propagation of exponentially-tailed Sobolev regularity are commutator estimates and the regularization  property of the gain operator.  In order to provide an elementary proof of the commutator estimates, we will rely on a generalization of the Bobylev's identity.  Namely, in the sequel, the Fourier transform $\mathcal{F}$ is denoted by $\widehat{f} := \mathcal{F}(f)$ for any tempered distribution $f$.  We continue using the shorthand notation $\hat{\xi}:=\xi/|\xi|$, with $\xi\in\mathbb{R}^{d}$, to denote unitary vectors since it should not present any confusion with that of the Fourier transform.  Next we state a central identity in the analysis is \cite[equation (2.15)]{BD} which is a generalization of that of Bobylev for Maxwell molecules \cite{B}.

\begin{lemma}\label{bobylev identity}
For any $(v,v_{*})$-integrable function $F_{x}(v,v_{*})$, with $x\in[-1,1]$, it holds that
\begin{equation}\label{BD}
\mathcal{F}\Big\{\int_{\mathbb{R}^{d}}\int_{\mathbb{S}^{d-1}}F_{\hat{u}\cdot\sigma}(v',v'_{*})\text{d}\sigma\text{d}v_{*}\Big\}(\xi) = \int_{\mathbb{S}^{d-1}}\widehat{F_{\hat{\xi}\cdot\sigma}}(\xi^{+},\xi^{-})\text{d}\sigma\,,\quad \xi^{\pm}:=\frac{\xi\pm|\xi|\sigma}{2}\,.
\end{equation}
In formula \eqref{BD}, $\widehat{F_{x}}$ is the Fourier transform of $F_{x}$ in the variables $(v,v_{*})$ keeping $x\in[-1,1]$ fixed.  
\end{lemma}

Such formula is applied to $F_{\hat{u}\cdot\sigma}(v',v'_{*})=f(v')g(v'_{*})B(|u|,\hat{u}\cdot\sigma)$ in \cite{BD}, however, in our context we will need its full generality.  Let us give a short proof of identity \eqref{BD}.
\begin{proof} Using the weak formulation
\begin{equation*}
\int_{\mathbb{R}^{d}}\bigg(\int_{\mathbb{R}^{d}}\int_{\mathbb{S}^{d-1}}F_{\hat{u}\cdot\sigma}(v',v'_{*})\text{d}\sigma\text{d}v_{*}\bigg) e^{-iv\cdot\xi}\text{d}v = \int_{\mathbb{R}^{d}}\int_{\mathbb{R}^{d}}\int_{\mathbb{S}^{d-1}}F_{\hat{u}\cdot\sigma}(v,v_{*})e^{-iv'\cdot\xi}\text{d}\sigma\text{d}v_{*}\text{d}v\,.
\end{equation*}
Now, for any $(v,v_{*},\xi)$ fixed, it follows that
\begin{align*}
\int_{\mathbb{S}^{d-1}}&F_{\hat{u}\cdot\sigma}(v,v_{*})e^{-iv'\cdot\xi}\text{d}\sigma =\int_{\mathbb{S}^{d-1}}F_{\hat{u}\cdot\sigma}(v,v_{*})e^{iu^{-}\cdot\xi}\text{d}\sigma\,e^{-iv\cdot\xi}\\
&=\int_{\mathbb{S}^{d-1}}F_{\hat{\xi}\cdot\sigma}(v,v_{*})e^{iu\cdot\xi^{-}}\text{d}\sigma\,e^{-iv\cdot\xi} = \int_{\mathbb{S}^{d-1}}F_{\hat{\xi}\cdot\sigma}(v,v_{*})e^{-iv_{*}\cdot\xi^{-}}e^{-iv\cdot\xi^{+}}\text{d}\sigma\,.
\end{align*}
We used in the second equality the change of variables $\sigma\rightarrow\text{R}\,\sigma$ where $\text{R}$ is the reflection that interchanges the unit vectors $\hat{u}$ and $\hat{\xi}$.  As a consequence,
\begin{align*}
\int_{\mathbb{R}^{d}}\bigg(\int_{\mathbb{R}^{d}}&\int_{\mathbb{S}^{d-1}}F_{\hat{u}\cdot\sigma}(v',v'_{*})\text{d}\sigma\text{d}v_{*}\bigg) e^{-iv\cdot\xi}\text{d}v \\
&= \int_{\mathbb{S}^{d-1}}\bigg(\int_{\mathbb{R}^{d}}\int_{\mathbb{R}^{d}}F_{\hat{\xi}\cdot\sigma}(v,v_{*})e^{-i(v,v_{*})\cdot(\xi^{+},\xi^{-})}\text{d}v_{*}\text{d}v\bigg)\text{d}\sigma\,,
\end{align*}
which is the desired identity.
\end{proof}

\noindent
Let us mention, in addition, that the proofs in this section are explicitly written for propagation of Sobolev regularity with exponential weights.  Similar results, shown with analogous arguments, hold for polynomial weights.  Also, we will restrict ourselves only to the physical range $\gamma\in[0,1]$; this is not central in the argument, but it simplifies some statements.  More important is the restriction $\alpha\in(0,1]$ in the exponential tail, central for the validity of some estimates with fractional commutators. 
\subsection{Commutators}
With the identity \eqref{BD} at our disposal, we set out to prove in this section that fractional differentiation commutes with the collision operator up to a lower order remainder.  This result is typical of convolution like operators and happens for both, the gain and loss Boltzmann operators.  These commutator estimates are new in the context of  cut-off Boltzmann equation, however, the idea has been used before in the context of Boltzmann equation without cut-off to develop an $L^{2}(\mathbb{R}^{d})$ regularity theory, see for example \cite[Section 3]{AMUXI}.  The technique used to prove these estimates is usually pseudo-differential calculus; here we take a more elementary approach exploiting formula \eqref{BD}.  This method is flexible and it could be adapted to study propagation in general $L^{p}(\mathbb{R}^{d})$ spaces because our commutator identities hold in the pointwise sense.  We do not explore this path and content ourselves with the presentation of the $L^{2}(\mathbb{R}^{2})$ theory.  We start by commutator estimates for the loss operator, followed by commutator estimates for the gain operator.
\begin{lemma}[Commutator for the loss operator]\label{L2Hard}
Take any $\gamma\in[0,1]$, $s\in(0,1]$, $r\in[0,\frac{1}{2})$ and $\alpha\in(0,1]$.  Then,
\begin{equation*}
\big(1+(-\Delta)\big)^{\frac{s}{2}}Q^{-}_{\gamma,b}\big(f,g\big) = Q^{-}_{\gamma,b}\big(\big(1+(-\Delta)\big)^{ \frac{s}{2} }f,g\big) + \mathcal{I}^{-}(f,g)\,,
\end{equation*}
where
\begin{equation*}
\Bigg\{
\begin{array}{cl}
\mathcal{I}^{-}(f,g) = 0 & \text{ for }\; \gamma=0\,, \\
\|\mathcal{I}^{-}(f,g)\,e^{r\langle\cdot\rangle^{\alpha}}\|_{L^{2}(\mathbb{R}^{d})}\leq C\|b\|_{L^{1}(\mathbb{S}^{d-1})}\| f\, e^{r\langle\cdot\rangle^{\alpha}} \|_{L^{2}(\mathbb{R}^{d})}\|g\|_{L^{2}_{\frac{d^{+}-(1-\gamma)}{2}}(\mathbb{R}^{d})} & \text{ for }\; \gamma\in(0,1]\,,
\end{array}
\end{equation*} 
with constant $C:=C(d,s,\gamma,r,\alpha)$.
\end{lemma}
\begin{proof}
The case $\gamma=0$ is trivial, thus, consider $\gamma\in(0,1]$.  Using Lemma \ref{ApL1} one has
\begin{align*}
\big(1+(-\Delta)\big)^{\frac{s}{2}}Q^{-}_{\gamma,b}\big(f,g\big)(v) = Q^{-}_{\gamma,b}\big(\big(1+(-\Delta)\big)^{\frac{s}{2}}f,g\big)(v) + \int_{\mathbb{R}^{d}}g(v_{*})\mathcal{R}_{v_*}(f)(v)\text{d}v_{*}\,. 
\end{align*}
Then, we define
\begin{align}\label{L2Hard*1}
\begin{split}
\mathcal{I}^{-}(f,g)(v) &:= \int_{\mathbb{R}^{d}}g(v_{*})\mathcal{R}_{v_*}(f)(v)\text{d}v_{*}\\
&= s\int_{\mathbb{R}^{d}}f(v-x)\nabla\varphi(x)\cdot\int_{\mathbb{R}^{d}}g(v_{*})\bigg(\int^{1}_{0}\nabla|\cdot|^{\gamma}(v-v_{*} - \theta\,x)\text{d}\theta\bigg)\text{d}v_{*}\text{d}x\,.
\end{split}
\end{align}
Recall that $\varphi := \mathcal{F}^{-1}\big\{ \langle \cdot \rangle^{s-2} \big\}$  is the inverse Fourier transform of the Bessel kernel of order $2-s$.  Now, let us first consider the case $s\in(0,1)$.  Using Lemma \ref{ApL0} we control the integral in $\theta$ in \eqref{L2Hard*1}
\begin{align*}
\big| \mathcal{I}^{-}(f,g)(v)\big| &\leq (2-\gamma)s\int_{\mathbb{R}^{d}}\big| f(v-x) \big| \big| \nabla\varphi(x)\big|\bigg(\int_{\mathbb{R}^{d}}\big|g(v_{*})\big||v-v_{*}|^{-(1-\gamma)}\text{d}v_{*}\bigg)\text{d}x\\
& \leq s\,C_{d,\gamma}\|g\|_{L^{2}_{\frac{d^{+}-(1-\gamma)}{2}}(\mathbb{R}^{d})} \int_{\mathbb{R}^{d}}\big| f(v-x) \big| \big| \nabla\varphi(x)\big|\text{d}x\,.
\end{align*}
In the last inequality we simply used Cauchy-Schwarz inequality in the $v_{*}$-integral,
\begin{align*}
\int_{\mathbb{R}^{d}}\frac{|g(v_{*})|}{|v-v_{*}|^{1-\gamma}}\text{d}v_{*}&\leq
\bigg(\int_{\mathbb{R}^{d}} g(v_{*}) |^{2}\langle v_{*} \rangle^{d^{+}-(1-\gamma)}\text{d}v_{*}\bigg)^{\frac{1}{2}}
\bigg(\int_{\mathbb{R}^{d}}\frac{\langle v_{*} \rangle^{-d^{+}+(1-\gamma)}}{|v-v_{*}|^{2(1-\gamma)}}\text{d}v_{*}\bigg)^{\frac{1}{2}}\\
&\leq C_{d,\gamma}\|g\|_{L^{2}_{\frac{d^{+}-(1-\gamma)}{2}}(\mathbb{R}^{d})}\,.
\end{align*}
Additionally, $|v|^{\alpha}\leq |v-x|^{\alpha} + |x|^{\alpha}$ for any $\alpha\in(0,1]$.  As a consequence,
\begin{align}\label{L2Harde1}
\big| \mathcal{I}^{-}(f,g)(v)\big|e^{r\langle v \rangle^{\alpha}} \leq s\,C_{d,\gamma}\|g\|_{L^{2}_{\frac{d^{+}-(1-\gamma)}{2}}(\mathbb{R}^{d})}\int_{\mathbb{R}^{d}}\big| f(v-x) e^{r\langle v-x \rangle^{\alpha}} \big| \big| \nabla\varphi(x) e^{r\langle x \rangle^{\alpha}}\big|\text{d}x\,,
\end{align}
and, invoking Young's inequality for convolutions yields
\begin{equation*}
\| \mathcal{I}^{-}(f,g)\,e^{r\langle\cdot\rangle^{\alpha}}\|_{L^{2}(\mathbb{R}^{d})} \leq s\,C_{d,\gamma}\|g\|_{L^{2}_{\frac{d^{+}-(1-\gamma)}{2}}(\mathbb{R}^{d})} \|f\,e^{r\langle\cdot\rangle^{\alpha}}\|_{L^{2}(\mathbb{R}^{d})}\|\nabla\varphi\,e^{r\langle\cdot\rangle^{\alpha}}\|_{L^{1}(\mathbb{R}^{d})}\,.
\end{equation*}
It is well known (see for instance \cite[Proposition 6.1.5]{GF}) that $\varphi$ is smooth, except at the origin, and decaying as $e^{-\frac{|x|}{2}}$.  Moreover,
\begin{equation}\label{bpb0}
\varphi(x) = \frac{1}{\Gamma(2-s)}|x|^{(2-s)-d} + 1 + O(|x|^{(2-s)-d+2})\,,\quad x\approx 0\,.
\end{equation}
Thus, $e^{r\langle\cdot\rangle^{\alpha}} \nabla\varphi \in L^{1}(\mathbb{R}^{d})$ for any $s\in(0,1)$ and $r\in[0,\frac{1}{2})$.  This proves the lemma in this case.  The case $s=1$ needs a special treatment, but the essential idea remains intact.  Indeed, using Remark \ref{R1Ap} in \eqref{L2Hard*1} it readily follows that
\begin{align*}
\big| \mathcal{I}^{-}(f,g)(v)\big| 
\leq & \bigg|\int_{\mathbb{R}^{d}}g(v_{*})\nabla|\cdot|^{\gamma}(v-v_{*})\text{d}v_{*}\bigg|\,\big|\big(\nabla\varphi\ast f\big)(v)\big|\\
&+C\int_{\mathbb{R}^{d}}\big| f(v-x) \big| \big| |x|^{\varepsilon}\nabla\varphi(x)\big|\bigg(\int_{\mathbb{R}^{d}}\big|g(v_{*})\big||v-v_{*}|^{-(1 - \gamma +\varepsilon)}\text{d}v_{*}\bigg)\text{d}x\,.
\end{align*}
The lemma follows from here using previous arguments and the fact that
\begin{equation*}
\|e^{r\langle\cdot\rangle^{\alpha}}\,\nabla\varphi\ast f\|_{L^{2}(\mathbb{R}^{d})} \leq C(s,r,\alpha)\|f\,e^{r\langle\cdot\rangle^{\alpha}}\|_{L^{2}(\mathbb{R}^{d})}\,.
\end{equation*}
\end{proof}
\begin{lemma}[Commutator for the gain operator]\label{L1Hard}
Let the potential be $\gamma\in[0,1]$, $s\in(0,1]$, $r\in[0,\frac{1}{4})$ and $\alpha\in(0,1]$.  Then,
\begin{equation*}
\big(1+(-\Delta)\big)^{\frac{s}{2}}Q^{+}_{\gamma,b}\big(f,g\big) = Q^{+}_{\gamma,b_{s}}\big(\big(1+(-\Delta)\big)^{ \frac{s}{2} }f,g\big) + \mathcal{I}^{+}(f,g)\,,
\end{equation*}
where $b_{s}(\cdot)=\Big(\frac{ 2 }{ 1+\, \cdot }\Big)^{\frac{s}{2}}b(\cdot)$, and
\begin{equation*}
\|\mathcal{I}^{+}(f,g)\,e^{r\langle\cdot\rangle^{\alpha}}\|_{L^{2}(\mathbb{R}^{d})}\leq C\,\|b\|_{L^{1}(\mathbb{S}^{d-1})}\|f\,e^{r\langle\cdot\rangle^{\alpha}}\|_{L^{2}_{\gamma}(\mathbb{R}^{d})}\|g\,e^{r\langle\cdot\rangle^{\alpha}}\|_{L^{2}_{\gamma + \frac{d^{+}}{2}}(\mathbb{R}^{d})}\,,
\end{equation*}
with constant $C:=C(d,s,\gamma,r,\alpha)$.
\end{lemma}
\begin{proof}
Using formula \eqref{BD} with $F(v',v'_{*})=f(v')|v'-v'_{*}|^{\gamma}g(v'_{*})$ one has
\begin{align}\label{L1Harde1}
\begin{split}
\mathcal{F}\big\{\big( 1 + (-\Delta) \big)^{\frac{s}{2}}&Q^{+}_{\gamma,b}(f,g)\big\}(\xi) = \langle \xi \rangle^{s}\int_{\mathbb{S}^{d-1}} \widehat{F}(\xi^{+},\xi^{-})b(\hat{\xi}\cdot\sigma)\text{d}\sigma\\
&= \int_{\mathbb{S}^{d-1}} \langle \xi^{+} \rangle^{s} \widehat{F}(\xi^{+},\xi^{-})b_{s}(\hat{\xi}\cdot\sigma)\text{d}\sigma + \widehat{\mathcal{I}^{+}_{1}(f,g)}(\xi)\,,
\end{split}
\end{align}
where
\begin{equation*}
\widehat{\mathcal{I}^{+}_{1}(f,g)}(\xi) :=  \int_{\mathbb{S}^{d-1}} \Big( \big( \tfrac{1+\hat{\xi}\cdot\sigma}{2} +  |\xi^{+}|^{2} \big)^{\frac{s}{2}} - \big( 1 + |\xi^{+}|^{2} \big)^{\frac{s}{2}} \Big) \widehat{F}(\xi^{+},\xi^{-})b_{s}(\hat{\xi}\cdot\sigma)\text{d}\sigma\,.
\end{equation*}
Now,
\begin{equation*}
\langle \xi^{+} \rangle^{s} \widehat{F}(\xi^{+},\xi^{-}) = \mathcal{F}\big\{\big(1+(-\Delta)\big)^{\frac{s}{2}} (f\, \tau_{v'_{*}}|\cdot|^{\gamma})(v')g(v'_{*})\big\}(\xi^{+},\xi^{-})\,.
\end{equation*}
We continue by using Lemma \ref{ApL1}
\begin{align}\label{L1Harde2}
\begin{split}
& \langle \xi^{+} \rangle^{s} \widehat{F}(\xi^{+},\xi^{-})  \\
& =\mathcal{F}\big\{\big(1+(-\Delta)\big)^{\frac{s}{2}}f(v')\times|v'-v'_{*}|^{\gamma}g(v'_{*})\big\}(\xi^{+},\xi^{-}) + \mathcal{F}\big\{\mathcal{R}_{v'_{*}}(f)(v')g(v'_{*})\big\}(\xi^{+},\xi^{-})\,.
\end{split}
\end{align}
The conclusion of \eqref{L1Harde1} and \eqref{L1Harde2} is
\begin{align*}
\mathcal{F}\big\{\big( 1 + (-\Delta) & \big)^{\frac{s}{2}}Q^{+}_{\gamma,b}(f,g)\big\}(\xi) = \widehat{\mathcal{I}^{+}_{1}(f,g)}(\xi) + \widehat{\mathcal{I}^{+}_{2}(f,g)}(\xi)\\
&+ \int_{\mathbb{S}^{d-1}} \mathcal{F}\big\{\big(1+(-\Delta)\big)^{\frac{s}{2}}f(v')\times|v'-v'_{*}|^{\gamma}g(v'_{*})\big\}(\xi^{+},\xi^{-}) b_{s}(\hat{\xi}\cdot\sigma)\text{d}\sigma\,,
\end{align*}
where
\begin{equation*}
\widehat{\mathcal{I}^{+}_{2}(f,g)}(\xi) := \int_{\mathbb{S}^{d-1}}\mathcal{F}\big\{\mathcal{R}_{v'_{*}}(f)(v')g(v'_{*})\big\}(\xi^{+},\xi^{-})b_{s}(\hat{\xi}\cdot\sigma)\text{d}\sigma\,.
\end{equation*} 
As a consequence, taking the inverse Fourier transform
\begin{equation}\label{L1Harde3}
\big(1 + (-\Delta)\big)^{\frac{s}{2}}Q^{+}_{\gamma,b}(f,g) = Q^{+}_{\gamma,b_{s}}\big( \big(1 + (-\Delta)\big)^{\frac{s}{2}}f,g\big)(v) + \mathcal{I}^{+}_{1}(f,g)(v) + \mathcal{I}^{+}_{2}(f,g)(v)\,.
\end{equation}
\textit{Controlling the term $\mathcal{I}^{+}_{2}$:}  Assume $\gamma>0$, otherwise, this term is zero.  Also, assume first $s\in(0,1)$ since the case $s=1$ needs special treatment.  Thanks to formula \eqref{BD}, the remainder term $\mathcal{I}^{+}_{2}$ is
\begin{align*}
\mathcal{I}^{+}_{2}(f,g)(v)&=\int_{\mathbb{R}^{d}}\int_{\mathbb{S}^{d-1}}\mathcal{R}_{v'_{*}}(f)(v')g(v'_{*})b_{s}(\hat{u}\cdot\sigma)\text{d}\sigma\text{d}v_{*}\,.
\end{align*}
Thus, using the pre-post collisional change of variables and Lemmas \ref{ApL1} and \ref{ApL0} one obtains that for any test function $\phi$
\begin{align*}
\bigg|&\int_{\mathbb{R}^{d}}\mathcal{I}^{+}_{2}(f,g)(v)\phi(v)\text{d}v\bigg|=\bigg|\int_{\mathbb{R}^{d}}\int_{\mathbb{R}^{d}}\int_{\mathbb{S}^{d-1}}\mathcal{R}_{v_{*}}(f)(v)g(v_{*})\phi(v')b_{s}(\hat{u}\cdot\sigma)\text{d}\sigma\text{d}v_{*}\text{d}v\bigg|\\
&\leq (2-\gamma)s\int_{\mathbb{R}^{d}}|\nabla\varphi(x)|\bigg(\int_{\mathbb{S}^{d-1}}\int_{\mathbb{R}^{d}}\int_{\mathbb{R}^{d}}\frac{|\tau_{x}f(v)||\phi(v')|}{|v-v_{*}|^{1-\gamma}}|g(v_{*})|b(\hat{u}\cdot\sigma)\text{d}v_{*}\text{d}v\text{d}\sigma\bigg)\text{d}x\,.
\end{align*}
Now, let us introduce the exponential weight by choosing $\phi(v) =e^{r\langle v\rangle^{\alpha}}\tilde{\phi} (v)$.  Since $|v'|^{2}\leq |v|^{2} + |v_{*}|^{2}$, if follows that
\begin{equation*}
e^{r\langle v'\rangle^{\alpha}} \leq e^{r\langle x\rangle^{\alpha}}e^{r\langle v-x\rangle^{\alpha}}e^{r\langle v_{*}\rangle^{\alpha}},\qquad \alpha\in(0,1],\quad r\geq0\,.
\end{equation*}
As a consequence,
\begin{align}\label{L1Harde4}
\begin{split}
\bigg|&\int_{\mathbb{R}^{d}}\mathcal{I}^{+}_{2}(f,g)(v)\phi(v)\text{d}v\bigg|\\
&\hspace{-.6cm}\leq (2-\gamma)s\int_{\mathbb{R}^{d}}|\tilde\varphi(x)|\bigg(\int_{\mathbb{S}^{d-1}}\int_{\mathbb{R}^{d}}\int_{\mathbb{R}^{d}}\frac{|\tau_{x}\tilde f(v)||\tilde\phi(v')|}{|v-v_{*}|^{1-\gamma}}|\tilde g(v_{*})|b(\hat{u}\cdot\sigma)\text{d}v_{*}\text{d}v\text{d}\sigma\bigg)\text{d}x\,.
\end{split}
\end{align}
where $\tilde\varphi=e^{r\langle\cdot\rangle^{\alpha}}\nabla\varphi$, $\tilde f = e^{r\langle\cdot\rangle^{\alpha}}f$ and $\tilde g = e^{r\langle\cdot\rangle^{\alpha}}g$.  For the integral in parenthesis in \eqref{L1Harde4} we use Cauchy-Schwarz inequality
\begin{align}\label{L1Harde4.5}
\begin{split}
\int_{\mathbb{S}^{d-1}}\int_{\mathbb{R}^{d}}\int_{\mathbb{R}^{d}}&\frac{|\tau_{x}\tilde f(v)||\tilde\phi(v')|}{|v-v_{*}|^{1-\gamma}}|\tilde g(v_{*})|b(\hat{u}\cdot\sigma)\text{d}v_{*}\text{d}v\text{d}\sigma\\
&\leq\bigg(\int_{\mathbb{S}^{d-1}}\int_{\mathbb{R}^{d}}\int_{\mathbb{R}^{d}}\frac{|\tau_{x}\tilde f(v)|^{2}}{|v-v_{*}|^{1-\gamma}}|\tilde g(v_{*})|b(\hat{u}\cdot\sigma)\text{d}v_{*}\text{d}v\text{d}\sigma\bigg)^{\frac{1}{2}}\\
&\hspace{1cm}\times\bigg(\int_{\mathbb{S}^{d-1}}\int_{\mathbb{R}^{d}}\int_{\mathbb{R}^{d}}\frac{|\tilde\phi(v')|^{2}}{|v-v_{*}|^{1-\gamma}}|\tilde g(v_{*})|b(\hat{u}\cdot\sigma)\text{d}v_{*}\text{d}v\text{d}\sigma\bigg)^{\frac{1}{2}}\,.
\end{split}
\end{align}
As for the first integral on the right side, one notes that
\begin{align*}
\int_{\mathbb{R}^{d}}\frac{|\tilde g(v_{*})|}{|v-v_{*}|^{1-\gamma}}\text{d}v_{*}
&\leq  \bigg(\int_{\mathbb{R}^{d}}|\tilde g(v_{*}) |^{2}\langle v_{*} \rangle^{d^{+}-(1-\gamma)}\text{d}v_{*}\bigg)^{\frac{1}{2}}
\times\bigg(\int_{\mathbb{R}^{d}}\frac{\langle v_{*} \rangle^{-d^{+}+(1-\gamma)}}{|v-v_{*}|^{2(1-\gamma)}}\text{d}v_{*}\bigg)^{\frac{1}{2}}\\
&\leq C_{d,\gamma}\|g\,e^{r|\cdot|^{\alpha}}\|_{L^{2}_{\frac{d^{+}-(1-\gamma)}{2}}(\mathbb{R}^{d})}\,.
\end{align*}
Therefore,
\begin{align}\label{L1Harde4.6}
\begin{split}
\bigg(\int_{\mathbb{S}^{d-1}}\int_{\mathbb{R}^{d}}\int_{\mathbb{R}^{d}}&\frac{|\tau_{x}\tilde f(v)|^{2}}{|v-v_{*}|^{1-\gamma}}|\tilde g(v_{*})|b(\hat{u}\cdot\sigma)\text{d}v_{*}\text{d}v\text{d}\sigma\bigg)^{\frac{1}{2}}\\
&\leq C_{d,\gamma}\|b\|^{\frac{1}{2}}_{L^{1}(\mathbb{S}^{d-1})}\|g\,e^{r|\cdot|^{\alpha}}\|^{\frac{1}{2}}_{L^{2}_{\frac{d^{+}-(1-\gamma)}{2}}(\mathbb{R}^{d})}\|f\,e^{r|\cdot|^{\alpha}}\|_{L^{2}(\mathbb{R}^{d})}\,.
\end{split}
\end{align}
For the second integral on the right side, just recall the definition $v'=v_{*}+u^{+}$ which gives us the identity
\begin{equation*}
|v'-v_{*}| = |u^{+}| = |u|\sqrt{\frac{1+\hat{u}\cdot\sigma}{2}}\,.
\end{equation*}
Then, using the classical change of variables $y = u^{+}$ (with fixed $\sigma$) having Jacobian $J=\frac{1}{2^{d}}(1+\hat{u}\cdot\sigma)$, it follows that
\begin{align}\label{L1Harde5}
\int_{\mathbb{R}^{d}}\frac{ |\tilde\phi(v')|^{2}} { |u|^{1-\gamma} } b(\hat{u}\cdot\sigma)\text{d}u = \int_{\mathbb{R}^{d}}\frac{ |\tilde\phi(v_{*} + y)|^{2}} { |y|^{1-\gamma} } \tilde{b}(\hat{u}(y,\sigma)\cdot\sigma)\text{d}y\,,\quad\; \tilde{b}(\cdot) := \frac{2^{d-\frac{1-\gamma}{2}}b(\cdot)}{(1+\cdot)^{ \frac{1+\gamma}{2} } }\,.
\end{align}
Furthermore, using that $\hat{u}(y,\sigma)\cdot\sigma = 2 (\hat{y}\cdot\sigma)^{2}-1$, it follows by a direct computation that
\begin{align}\label{L1Harde6}
\int_{\mathbb{S}^{d-1}}\tilde{b}(\hat{u}(y,\sigma)\cdot\sigma)\text{d}\sigma = 2^{-\frac{d}{2}}\int_{\mathbb{S}^{d-1}}\frac{\tilde{b}(e_{1}\cdot\sigma)}{(1+e_{1}\cdot\sigma)^{\frac{d - 2}{2}}}\text{d}\sigma\leq C_{d,\gamma}\|b\|_{L^{1}(\mathbb{S}^{d-1})}\,.
\end{align}
Here $e_{1}$ is any fixed unitary vector.  As a consequence of \eqref{L1Harde5} and \eqref{L1Harde6} we obtain
\begin{align}\label{L1Harde7}
\begin{split}
\bigg(\int_{\mathbb{S}^{d-1}}\int_{\mathbb{R}^{d}}\int_{\mathbb{R}^{d}}\frac{|\tilde \phi(v')|^{2}}{|v-v_{*}|^{1-\gamma}}&|\tilde g(v_{*})|b(\hat{u}\cdot\sigma)\text{d}v_{*}\text{d}v\text{d}\sigma\bigg)^{\frac{1}{2}}\\
&\hspace{-.5cm}\leq C_{d,\gamma}\|b\|^{\frac{1}{2}}_{L^{1}(\mathbb{S}^{d-1})}\|g \,e^{r\langle\cdot\rangle^{\alpha}}\|^{\frac{1}{2}}_{L^{2}_{\frac{d^{+} - (1-\gamma)}{2}}(\mathbb{R}^{d})}\|\tilde\phi\|_{L^{2}(\mathbb{R}^{d})}\,.
\end{split}
\end{align}
Therefore, after gathering \eqref{L1Harde4.5}, \eqref{L1Harde4.6} and \eqref{L1Harde7} and plugging into \eqref{L1Harde4} we have
\begin{align}\label{L1Harde8}
\begin{split}
&\bigg|\int_{\mathbb{R}^{d}}\mathcal{I}^{+}_{2}(f,g)(v)\,e^{r\langle\cdot\rangle^{\alpha}}\tilde\phi(v)\text{d}v\bigg|\\
&\hspace{-.5cm}\leq C_{d,\gamma}\,s\,\|b\|_{L^{1}(\mathbb{S}^{d-1})}\|\tilde\varphi\|_{L^{1}(\mathbb{R}^{d})}\|f\,e^{r\langle\cdot\rangle^{\alpha}}\|_{L^{2}(\mathbb{R}^{d})}\|g\,e^{r\langle\cdot\rangle^{\alpha}}\|_{L^{2}_{\frac{d^{+} - (1-\gamma)}{2}}(\mathbb{R}^{d})}\|\tilde\phi\|_{L^{2}(\mathbb{R}^{d})}\,.
\end{split}
\end{align}
Since $\|\tilde\varphi\|_{L^{1}(\mathbb{R}^{d})}=\|\nabla\varphi\,e^{r\langle\cdot\rangle^{\alpha}}\|_{L^{1}(\mathbb{R}^{d})}<\infty$, for any $\alpha\in(0,1]$ and $r\in[0,\frac{1}{2})$, the result follows by duality.

\noindent
The case $s=1$ needs special mention due to the fact that $\nabla\varphi$ is not integrable at the origin.  However, this issue presents no obstacle for the estimate.  Indeed, using Remark \ref{R1Ap} in the Appendix, one has
\begin{align*}
\bigg|&\int_{\mathbb{R}^{d}}\mathcal{I}^{+}_{2}(f,g)(v)\phi(v)\text{d}v\bigg|=\bigg|\int_{\mathbb{R}^{d}}\int_{\mathbb{R}^{d}}\int_{\mathbb{S}^{d-1}}\mathcal{R}_{v_{*}}(f)(v)g(v_{*})\phi(v')b_{s}(\hat{u}\cdot\sigma)\text{d}\sigma\text{d}v_{*}\text{d}v\bigg|\\
& \leq \gamma\int_{\mathbb{S}^{d-1}}\int_{\mathbb{R}^{d}}\int_{\mathbb{R}^{d}}\frac{|(\nabla\varphi\ast f)(v)||\phi(v')|}{|v-v_{*}|^{1-\gamma}}|g(v_{*})|b(\hat{u}\cdot\sigma)\text{d}v_{*}\text{d}v\text{d}\sigma\\
&+  C\int_{\mathbb{R}^{d}}\big||x|^{\varepsilon}\nabla\varphi(x)\big|\bigg(\int_{\mathbb{S}^{d-1}}\int_{\mathbb{R}^{d}}\int_{\mathbb{R}^{d}}\frac{|\tau_{x}f(v)||\phi(v')|}{|v-v_{*}|^{1+\varepsilon -\gamma}}|g(v_{*})|b(\hat{u}\cdot\sigma)\text{d}v_{*}\text{d}v\text{d}\sigma\bigg)\text{d}x\,.
\end{align*}
The first term on the right side can be estimated by adapting previous argument with $\nabla\varphi\ast f$ instead of $f$.  While for the second term, one simply repeats the argument with $|\cdot|^{\varepsilon}\nabla\varphi$ instead of $\nabla\varphi$.\\

\noindent
\textit{Controlling the term $\mathcal{I}^{+}_{1}$:} Using the fact that
\begin{align*}
-\bigg( \big( 1 + |\xi^{+}|^{2} \big)^{\frac{s}{2}} - \big( a(\hat{\xi}\cdot\sigma) +  |\xi^{+}|^{2} \big)^{\frac{s}{2}} \bigg) = -\frac{s}{2}\int^{1}_{0}\frac{1-a(\hat{\xi}\cdot\sigma)}{\big((1-\theta)a(\hat{\xi}\cdot\sigma)+\theta+|\xi^{+}|^{2}\big)^{\frac{2-s}{2}}}\text{d}\theta\,,
\end{align*}
one can conclude, for $a(\hat{\xi}\cdot\sigma):=\frac{1+\hat{\xi}\cdot\sigma}{2}$ and $\ell(\theta,\hat{\xi}\cdot\sigma):=(1-\theta)a(\hat{\xi}\cdot\sigma) + \theta$, that
\begin{align*}
\widehat{\mathcal{I}^{+}_{1}(f,g)}(\xi) = -\frac{s}{2}\int^{1}_{0}\int_{\mathbb{S}^{d-1}}\Big\langle \tfrac{\xi^{+}}{\sqrt{\ell(\theta,\hat{\xi}\cdot\sigma)}}\Big\rangle^{-(2-s)}\widehat{F}(\xi^{+},\xi^{-})\frac{\big(1-a(\hat{\xi}\cdot\sigma)\big)b_{s}(\hat{\xi}\cdot\sigma)}{\ell(\theta,\hat{\xi}\cdot\sigma)^{\frac{2-s}{2}}}\text{d}\sigma\text{d}\theta\,.
\end{align*}
In addition, note that
\begin{align*}
\Big\langle \tfrac{\xi^{+}}{\sqrt{\ell(\theta,\hat{\xi}\cdot\sigma)}}\Big\rangle^{-(2-s)}&\widehat{F}(\xi^{+},\xi^{-}) \\
&\hspace{-2cm}= \, \ell(\theta,\hat{\xi}\cdot\sigma)^{\frac{d}{2}}\mathcal{F}\Big\{ \varphi\big(\ell(\theta,\hat{\xi}\cdot\sigma)^{\frac{1}{2}}\,\cdot\big)\ast \big(f(\cdot)\tau_{v'_{*}}|\cdot|^{\gamma}\big)(v')\,g(v'_{*})\Big\}(\xi^{+},\xi^{-})\,.
\end{align*}
As a consequence, we deduce from formula \eqref{BD} (for the inner integral $\theta$ is fixed) that
\begin{align*}
\mathcal{I}^{+}_{1}(f,g)(v)=-\frac{s}{2}\int^{1}_{0}\int_{\mathbb{R}^{d}}&\int_{\mathbb{S}^{d-1}}\varphi\big(\ell(\theta,\hat{u}\cdot\sigma)^{\frac{1}{2}}\,\cdot\big)\ast \big(f(\cdot)\tau_{v'_{*}}|\cdot|^{\gamma}\big)(v')\,g(v'_{*})\\
&\times\ell(\theta,\hat{u}\cdot\sigma)^{\frac{d+s-2}{2}}\big(1-a(\hat{u}\cdot\sigma)\big)b_{s}(\hat{u}\cdot\sigma)\text{d}\sigma\text{d}v_{*}\text{d}\theta\,.
\end{align*}
Now, $\ell\in[\frac{1}{2},1]$ and $\varphi$ is monotone decreasing, thus $\varphi(\ell\,\cdot)\leq \varphi(\frac{\cdot}{2})$.  As a consequence,
\begin{align*}
e^{r\langle v\rangle^{\alpha}} \Big|\mathcal{I}^{+}_{1}(f,g)(v)\Big|& \leq \frac{s\,e^{r\langle v\rangle^{\alpha}}}{2}\int_{\mathbb{R}^{d}}\int_{\mathbb{S}^{d-1}}\varphi(\tfrac{\cdot}{2})\ast (f(\cdot)\tau_{v'_{*}}|\cdot|^{\gamma})(v')\,g(v'_{*})\tilde{b}(\hat{u}\cdot\sigma)\text{d}\sigma\text{d}v_{*}\\
&\leq \frac{s}{2}\int_{\mathbb{R}^{d}}\int_{\mathbb{S}^{d-1}}(\tilde\varphi\ast\tilde{f})(v')\,\tilde{g}(v'_{*})\tilde{b}(\hat{u}\cdot\sigma)\text{d}\sigma\text{d}v_{*}=\frac{s}{2}\,Q^{+}_{0,\tilde{b}}(\tilde\varphi\ast\tilde{f},\tilde{g})(v)\,,
\end{align*}
where
\begin{equation*}
\tilde\varphi:=\varphi(\frac{\cdot}{2})e^{r\langle\cdot\rangle^{\alpha}},\;\; \tilde{f}(\cdot):=f(\cdot)\,e^{r\langle\cdot\rangle^{\alpha}}\langle \cdot \rangle^{\gamma},\;\; \tilde{g}(\cdot) := g(\cdot)\,e^{r\langle\cdot\rangle^{\alpha}}\langle \cdot \rangle^{\gamma},\;\; \tilde{b}(\cdot):=(1-a(\cdot))b_{s}(\cdot).
\end{equation*}
Therefore, using Young's inequality for the $Q^{+}_{0,\tilde{b}}$ it readily follows that
\begin{align}\label{L1Harde9}
\begin{split}
\big\|\mathcal{I}^{+}_{1}(f,g)\,e^{r\langle\cdot\rangle^{\alpha}}\big\|_{L^{2}(\mathbb{R}^{d})}
&\leq c_{d}\,s\,\|\tilde{b}\|_{L^{1}(\mathbb{S}^{d-1})}\|\tilde\varphi\ast\tilde{f}\|_{L^{2}(\mathbb{R}^{d})}\|\tilde{g}\|_{L^{1}(\mathbb{R}^{d})}\\
&\leq c_{d}\,s\,\|b\|_{L^{1}(\mathbb{S}^{d-1})}\|\tilde\varphi\|_{L^{1}(\mathbb{R}^{d})}\|f\,e^{r\langle\cdot\rangle^{\alpha}}\|_{L^{2}_{\gamma}(\mathbb{R}^{d})}\|g\,e^{r\langle\cdot\rangle^{\alpha}}\|_{L^{1}_{\gamma}(\mathbb{R}^{d})}\,.
\end{split}
\end{align}
The result follows using \eqref{L1Harde8}, \eqref{L1Harde9}, the estimate $\|g\|_{L^{1}(\mathbb{R}^{d})}\leq C_{d}\|g\|_{L^{2}_{d^{+}/2}(\mathbb{R}^{d})}$ valid for general function $g$, and $\|\tilde\varphi\|_{L^{1}(\mathbb{R}^{d})}<\infty$ for $\alpha\in(0,1]$, $r\in[0,\frac{1}{4})$.
\end{proof}
\begin{remark} Although the method of proof of Lemmas \ref{L2Hard} and \ref{L1Hard} uses Fourier transform, the argument is essentially made in the velocity space.  Thus, the method is flexible and it can be modified to obtain more general $L^{p}(\mathbb{R}^{d})$ estimates.  Observe also that polynomial weights can be readily handled using the same argument and leading to analogous estimates.
\end{remark}
\subsection{Regularization of the Boltzmann gain operator}
The second main ingredient for proving the propagation of exponentially-tailed Sobolev regularity is the regularization  property of the gain operator.
This regularizing effect  was first established in \cite{Lions} and later, with a more elementary proof, polynomial weights were added in \cite{Wennberg, MV}.  The following theorem is a version of the regularizing effect using exponential weights, and the proof is a simple adaptation of the technique given in \cite[Lemma 2.3]{ALodsSS} for polynomial weights.
\begin{theorem}\label{TSmooth}
Assume that collision kernel is of the form
\begin{equation*}
B(x,y)=\Phi(x)b(y),\quad \text{with}\quad b\in\mathcal{C}^{\infty}_{0}\big([0,1)\big)\,,
\end{equation*}
and $\Phi\in\mathcal{C}^{\infty}_{0}\big((0,\infty]\big)$, with $\Phi(x)\approx x^{\gamma}$ for large $x$ ($\gamma\geq0$).  Then, for $d\geq2$, $s\geq0$, $r\geq0$, and $\alpha\in(0,1]$
\begin{equation*}
\|e^{r\langle\cdot\rangle^{\alpha}}Q^{+}(f,g)\|_{H^{s+\frac{d-1}{2}}(\mathbb{R}^{d})}\leq C\,\|e^{r\langle\cdot\rangle^{\alpha}}\langle\cdot\rangle^{\mu}f\|_{H^{s}(\mathbb{R}^{d})}\|e^{2r\langle\cdot\rangle^{\alpha}}\langle\cdot\rangle^{\mu}g\|_{L^{1}(\mathbb{R}^{d})}\,,
\end{equation*}
where $\mu:=\mu(s,\gamma)=s^{+}+\gamma+\frac{3}{2}$, and the constant $C>0$ depends, in particular, on the distance from the support of $\Phi$ and $b$ to zero and one respectively.
\end{theorem}
\begin{proof}
One uses the argument of \cite[Lemma 2.3]{ALodsSS} which is a relaxation of \cite[Theorem 3.1]{MV}; the central step is to make sure that the expression
\begin{equation*}
\max_{|\nu|\leq s+\frac{d-1}{2}}\sup_{w\in\mathbb{S}^{d-1}}\bigg\| \mathcal{B}\big(|z|,|z \cdot w|\big)\frac{e^{r\langle z\cdot w \rangle^{\alpha}}z^{\nu}}{e^{r\langle z\rangle^{\alpha}}\langle z \rangle^{\mu}}\bigg\|_{H^{S}_{z}(\mathbb{R}^{d})}
\end{equation*} 
remains finite.  Here
\begin{equation*}
\mathcal{B}(x,y):= \frac{\Phi(x)\,b\Big(2 \frac{y^{2}}{x^{2}} - 1\Big)}{x^{d-2}\;y},\qquad x,\,y>0\,.
\end{equation*}
Also, $\nu$ is a multi-index and $S:=s + \lfloor d/2 \rfloor + 1$.  This holds true if $\mu:=s^{+}+\gamma+\frac{3}{2}$.
\end{proof}
Consider now the following decomposition of the gain collision operator, in the spirit of \cite{MV}, where $n$ stand for ``nice'' and $r$ for ``remainder'' terms:
\begin{equation}\label{decomposition}
x^{\gamma} = \Phi_{n}(x) + \Phi_{r}(x),\qquad b(y)= b_{n}(y) + b_{r}(y),
\end{equation}
with nonnegative functions satisfying the following properties for $\varepsilon,\delta>0$ sufficiently small
\begin{itemize}
\item [(1)] $\Phi_{n}\in\mathcal{C}^{\infty}\big((0,\infty)\big)$ vanishing in $(0,\delta)$.
\item [(2)] $\Phi_{r}\in L^{2}\cap L^{\infty}\big((0,\infty)\big)$ vanishing in $(2\delta,\infty)$.  Note that
\begin{equation*}
\|\Phi_{r}\|_{L^{\infty}}\leq (2\delta)^{\gamma},\qquad \|\Phi_{r}\|_{L^{2}}\leq (2\delta)^{\gamma+1/2}.
\end{equation*}
\item [(4)] $b_{n}\in\mathcal{C}^{\infty}\big((0,1)\big)$ vanishing in $(1-\varepsilon,1)$.
\item [(3)] $b_{r}\in L^{1}\big((1-y^{2})^{\frac{d-3}{2}}\text{d}y\big)$, with $\|b_{r}\|_{L^{1}\big((1-y^{2})^{\frac{d-3}{2}}\text{dy}\big)}\sim \mathfrak{m}(b_{r})$ sufficiently small.
\end{itemize}
Decomposition of the collision kernel \eqref{decomposition} leads to the decomposition of the gain operator
\begin{equation*}
Q^{+}(f,g) = Q^{+}_{nn}(f,g)+Q^{+}_{nr}(f,g)+Q^{+}_{rn}(f,g)+Q^{+}_{rr}(f,g)\,,
\end{equation*}
where $nr$ stands for nice kinetic potential with the remainder in the angular scattering.  In similar fashion we denote the other terms.
\begin{lemma}[Remainder terms]\label{RL}
Fix $\varepsilon,\delta>0$ sufficiently small.  The following estimates hold for any $\gamma\in[0,1]$
\begin{align*}
\big\|\langle\cdot \rangle^{-\gamma/2} Q^{+}_{nr}(f,g)\big\|_{L^{2}(\mathbb{R}^{d})} \leq \mathfrak{m}(b_{r})\|f\langle\cdot\rangle^{\gamma/2}\|_{L^{2}(\mathbb{R}^{d})}\|g\langle\cdot\rangle^{\gamma}\|_{L^{1}(\mathbb{R}^{d})}\,,\\
\big\|\langle\cdot \rangle^{-\gamma/2} Q^{+}_{rr}(f,g)\big\|_{L^{2}(\mathbb{R}^{d})}\leq \mathfrak{m}(b_{r})\|f\langle\cdot\rangle^{\gamma/2}\|_{L^{2}(\mathbb{R}^{d})}\|g\langle\cdot\rangle^{\gamma}\|_{L^{1}(\mathbb{R}^{d})}\,.
\end{align*}
For the $rn$ term we have
\begin{align*}
\big\|\langle\cdot \rangle^{-\gamma/2} Q^{+}_{rn}(f,g)\big\|_{L^{2}(\mathbb{R}^{d})} &\leq C(b)\,\delta^{\gamma}\,\|f\langle\cdot\rangle^{\gamma/2}\|_{L^{2}(\mathbb{R}^{d})}\|g\langle\cdot\rangle^{\gamma}\|_{L^{1}(\mathbb{R}^{d})},\quad \gamma>0\,,\\
\big\| Q^{+}_{rn}(f,g)\big\|_{L^{2}(\mathbb{R}^{d})} &\leq C(b,\varepsilon)\delta^{1/2}\|f\|_{L^{2}(\mathbb{R}^{d})}\|g\|_{L^{2}(\mathbb{R}^{d})},\quad \gamma=0\,.
\end{align*}
\end{lemma}
\begin{proof}
The estimates for the terms $nr$, $rr$ (for $\gamma\geq0$), and $rn$ (for $\gamma>0$) are classical and can be found, for example, in \cite{MV}.  They can be obtained as a simple consequence of a standard application of Young's inequality for the gain collision operator Theorem \ref{ApT-1}.  The constant in front of the inequality for the terms $nr$ and $rr$ is small due to the small mass of the remainder angular kernel $b_{r}$.  For term $rn$ in the case of Maxwell molecules, i.e. $\gamma=0$, one can use a generalized version of the Young's inequality that includes the kinetic kernel potential, see \cite[Corollary 8]{ACG} 
\begin{align*}
\big\| Q^{+}_{rn}(f,g)\big\|_{L^{2}(\mathbb{R}^{d})} &\leq C(b,\varepsilon)\|\Phi_{r}\|_{L^{2}(\mathbb{R}^{d})}\|f\|_{L^{2}(\mathbb{R}^{d})}\|g\|_{L^{2}(\mathbb{R}^{d})}\\
&\leq C(b,\varepsilon) \delta^{1/2} \|f\|_{L^{2}(\mathbb{R}^{d})}\|g\|_{L^{2}(\mathbb{R}^{d})}\,.
\end{align*}
The finiteness of the constant $C(b,\varepsilon)>0$ is ensured by the truncation of $b_{n}$ near $1$.
\end{proof}
\subsection{Propagation of regularity with exponential tails}
We are now ready to address the main question of this section -  propagation of exponentially-tailed Sobolev regularity.  We start by considering  small enough regularity. This result will later form a basis for an inductive proof of the propagation of higher regularity with exponential tails.

\begin{proposition}\label{PSmooth1}
Consider the integrable assumption $b\in L^{1}(\mathbb{S}^{d-1})$.  Assume that for some $\gamma\in[0,1]$, $s\in(0,\min\{1,\frac{d-1}{2}\}]$, and $\alpha\in(0,1]$, one has
\begin{equation*}
e^{a_o\langle v \rangle^{\alpha}}\big(1 + (-\Delta) \big)^{\frac{s}{2}}f_0 \in L^{2}(\mathbb{R}^{d}). 
\end{equation*} 
Then, for the solution $f(t,v)$ of the Boltzmann equation, there exists $a\in(0, a_o]$ such that for any $r\in(0,a)$ it holds that
\begin{equation*}
\big\|e^{r\langle\cdot\rangle^{\alpha}}\big(1 + (-\Delta) \big)^{\frac{s}{2}}f(t,\cdot)\big\|_{L^{2}(\mathbb{R}^{d})}\leq C(f_0)+\big\|e^{r\langle \cdot \rangle^{\alpha}}\big(1 + (-\Delta) \big)^{\frac{s}{2}}f_0\big\|_{L^{2}(\mathbb{R}^{d})},
\end{equation*}
where $C(f_0):=C(f_0,\gamma,s,a_o,\alpha)$ depends on lowers norms of $f_0$.
\end{proposition}
\begin{proof}
Fix $\varepsilon,\delta>0$ and write the equation as
\begin{equation}\label{PR1e0}
\partial_{t}f=Q^{+}_{nn}(f,f) + Q^{+}_{nr}(f,f) +Q^{+}_{rn}(f,f) + Q^{+}_{rr}(f,f) - Q^{-}(f,f).
\end{equation}
We will estimate each of these terms on the right side starting from the remainder terms.  Note that for the term $nr$ it follows, invoking the commutator Lemma \ref{L1Hard}, that
\begin{equation*}
e^{r\langle v \rangle^{\alpha}}\big(1+(-\Delta)\big)^{s/2}Q^{+}_{nr}(f,f) = e^{r\langle v \rangle^{\alpha}}Q^{+}_{nr}\big(\big(1+(-\Delta)\big)^{s/2}f,f\big) + e^{r\langle v \rangle^{\alpha}}\mathcal{I}^{+}_{nr}(f,f)\,.
\end{equation*}
Using Cauchy-Schwarz inequality and Lemma \ref{RL}
\begin{align*}
\int_{\mathbb{R}^{d}} & e^{r\langle v \rangle^{\alpha}}Q^{+}_{nr}\big(\big(1+(-\Delta)\big)^{s/2}f,f\big)(v)\,e^{r\langle v \rangle^{\alpha}}\big(1+(-\Delta)\big)^{s/2}f(v)\text{d}v\\
&\leq\big\|\langle\cdot\rangle^{-\gamma/2}e^{r\langle v \rangle^{\alpha}}Q^{+}_{nr}\big(\big(1+(-\Delta)\big)^{s/2}f,f\big)\big\|_{L^{2}(\mathbb{R}^{d})}  \big\|\langle\cdot\rangle^{\gamma/2}e^{r\langle v \rangle^{\alpha}}\big(\big(1+(-\Delta)\big)^{s/2}f\big\|_{L^{2}(\mathbb{R}^{d})}\\
 &\leq \mathfrak{m}(b_{r})\big\|\langle\cdot\rangle^{\gamma/2}e^{r\langle \cdot \rangle^{\alpha}}\big(1+(-\Delta)\big)^{s/2}f\big\|^{2}_{L^{2}(\mathbb{R}^{d})}\big\|\langle\cdot\rangle^{\gamma}e^{r\langle \cdot \rangle^{\alpha}} f\big\|_{L^{1}(\mathbb{R}^{d})}.
\end{align*}
Also, using Cauchy-Schwarz and Lemma \ref{L1Hard} again
\begin{align*}
\int_{\mathbb{R}^{d}}e^{r\langle v \rangle^{\alpha}}&\mathcal{I}^{+}_{nr}\big(f,f\big)(v)\,e^{r\langle v \rangle^{\alpha}}\big(1+(-\Delta)\big)^{s/2}f(v)\text{d}v\\
&\leq C(b)\big\|\langle \cdot \rangle^{\gamma+\frac{d^{+}}{2}} e^{r\langle\cdot \rangle^{\alpha}}f \big\|^{2}_{L^{2}(\mathbb{R}^{2})}\big\|e^{r\langle \cdot \rangle^{\alpha}}\big(1+(-\Delta)\big)^{s/2}f\big\|_{L^{2}(\mathbb{R}^{d})}.
\end{align*}
Thus, thanks to propagation of exponential moments and Theorems \ref{T1} and \ref{T2}, there exists $0<a\leq a_o$ such that for any $r\in(0,a)$
\begin{align}\label{PR1e1}
\begin{split}
\int_{\mathbb{R}^{d}} & e^{r\langle v \rangle^{\alpha}}\big(1+(-\Delta)\big)^{s/2}Q^{+}_{nr}\big(f,f\big)(v)\,e^{r\langle v \rangle^{\alpha}}\big(1+(-\Delta)\big)^{s/2}f(v)\text{d}v\\
&\leq \mathfrak{m}(b_{r})\big\|e^{r\langle \cdot \rangle^{\alpha}}\big(1+(-\Delta)\big)^{s/2}f\big\|^{2}_{L^{2}_{\gamma/2}(\mathbb{R}^{d})} + C\big\|e^{r\langle \cdot \rangle^{\alpha}}\big(1+(-\Delta)\big)^{s/2}f\big\|_{L^{2}_{\gamma/2}(\mathbb{R}^{d})}.
\end{split}
\end{align}
The same estimate holds for the term $rr$ as well.  Furthermore, similar argument leads to the control for the $rn$ term
\begin{align}\label{PR1e2}
\begin{split}
\int_{\mathbb{R}^{d}} & e^{r\langle v \rangle^{\alpha}}\big(1+(-\Delta)\big)^{s/2}Q^{+}_{rn}\big(f,f\big)(v)\,e^{r\langle v \rangle^{\alpha}}\big(1+(-\Delta)\big)^{s/2}f(v)\text{d}v\\
&\leq C(b,\varepsilon)\delta^{1/2}\big\|e^{r\langle \cdot \rangle^{\alpha}}\big(1+(-\Delta)\big)^{s/2}f\big\|^{2}_{L^{2}_{\gamma/2}(\mathbb{R}^{d})} + C\big\|e^{r\langle \cdot \rangle^{\alpha}}\big(1+(-\Delta)\big)^{s/2}f\big\|_{L^{2}_{\gamma/2}(\mathbb{R}^{d})}.
\end{split}
\end{align}
These last two estimates will handle the remainder term.  For the term $nn$, one uses the commutator formula of Lemma \ref{ApLMw} 
\begin{align*}
e^{r\langle v \rangle^{\alpha}}\big(1+(-\Delta)\big)^{s/2}Q^{+}_{nn}\big(f,f\big)(v)=\big(1+(-\Delta)\big)^{s/2}e^{r\langle \cdot \rangle^{\alpha}}Q^{+}_{nn}\big(f,f\big)(v) - \mathcal{R}\big( Q^{+}_{nn}(f,f)\big)\,.
\end{align*}
Thus, invoking Theorem \ref{TSmooth} and Lemma \ref{ApLMw} to estimate $Q^{+}_{nn}$ and $\mathcal{R}$ respectively
\begin{align*}
\int_{\mathbb{R}^{d}} & e^{r\langle v \rangle^{\alpha}}\big(1+(-\Delta)\big)^{s/2}Q^{+}_{nn}\big(f,f\big)(v)\,e^{r\langle v \rangle^{\alpha}}\big(1+(-\Delta)\big)^{s/2}f(v)\text{d}v\\
&=\int_{\mathbb{R}^{d}}\Big[\big(1+(-\Delta)\big)^{s/2}e^{r\langle \cdot \rangle^{\alpha}}Q^{+}_{nn}\big(f,f\big)(v) - \mathcal{R}\big( Q^{+}_{nn}(f,f)\big)\Big]e^{r\langle v \rangle^{\alpha}}\big(1+(-\Delta)\big)^{s/2}f(v)\text{d}v\\
&\leq  C(\delta,\varepsilon)\big\|e^{r\langle \cdot \rangle^{\alpha}}f\big\|_{L^{2}_{\mu}(\mathbb{R}^{d})}\big\|e^{2r\langle \cdot \rangle^{\alpha}}f\big\|_{L^{1}_{\mu}(\mathbb{R}^{d})}\big\|e^{r\langle \cdot \rangle^{\alpha}}\big(1+(-\Delta)\big)^{s/2}f\big\|_{L^{2}(\mathbb{R}^{d})}\\
&\hspace{1.cm}+C\big\|e^{r\langle \cdot \rangle^{\alpha}}Q^{+}_{nn}(f,f)\big\|_{L^{2}(\mathbb{R}^{d})}\big\|e^{r\langle \cdot \rangle^{\alpha}}\big(1+(-\Delta)\big)^{s/2}f\big\|_{L^{2}(\mathbb{R}^{d})}\,.
\end{align*}
As a consequence, one concludes that
\begin{align}\label{PR1e3}
\begin{split}
\int_{\mathbb{R}^{d}}e^{r\langle v \rangle^{\alpha}}\big(1+&(-\Delta)\big)^{s/2}Q^{+}_{nn}\big(f,f\big)(v)\,e^{r\langle v \rangle^{\alpha}}\big(1+(-\Delta)\big)^{s/2}f(v)\text{d}v\\
&\leq C(f_0,\delta,\varepsilon)\big\|e^{r\langle \cdot \rangle^{\alpha}}\big(1+(-\Delta)\big)^{s/2}f\big\|_{L^{2}(\mathbb{R}^{d})}\,.
\end{split}
\end{align}
Finally, using Lemma \ref{L2Hard} and Lemma \ref{APq-}
\begin{align}\label{PR1e4}
\begin{split}
\int_{\mathbb{R}^{d}} &e^{r\langle v \rangle^{\alpha}}\big(1+(-\Delta)\big)^{s/2}Q^{-}\big(f,f\big)(v)\,e^{r\langle v \rangle^{\alpha}}\big(1+(-\Delta)\big)^{s/2}f(v)\text{d}v\\
&\hspace{0cm}= \int_{\mathbb{R}^{d}}e^{r\langle v \rangle^{\alpha}}\Big(Q^{-}\big(\big(1+(-\Delta)\big)^{s/2}f,f\big)(v) + \mathcal{I}^{-}\big(f,f\big)\Big)e^{r\langle v \rangle^{\alpha}}\big(1+(-\Delta)\big)^{s/2}f(v)\text{d}v\\
&\geq  c_0\big\|e^{r\langle \cdot \rangle^{\alpha}}\big(1+(-\Delta)\big)^{s/2}f\big\|^{2}_{L^{2}_{\gamma/2}(\mathbb{R}^{d})} - C(f_0)\big\|e^{r\langle \cdot \rangle^{\alpha}}\big(1+(-\Delta)\big)^{s/2}f\big\|_{L^{2}(\mathbb{R}^{d})}\,.
\end{split}
\end{align}
In summary, gathering estimates \eqref{PR1e1}, \eqref{PR1e2}, \eqref{PR1e3} and \eqref{PR1e4}, it follows from \eqref{PR1e0} that
\begin{align*}
\frac{\text{d}}{\text{d}t} \big\|&e^{r\langle \cdot \rangle^{\alpha}}\big(1+(-\Delta)\big)^{s/2}f\big\|^{2}_{L^{2}(\mathbb{R}^{d})}\\
& \leq C(f_0,\varepsilon,\delta) 
+ \Big( \mathfrak{m}(b_{r}) + C(b,\varepsilon)\delta^{\frac{1}{2}} - c_0\Big)\big\|e^{r\langle \cdot \rangle^{\alpha}}\big(1+(-\Delta)\big)^{s/2}f\big\|^{2}_{L^{2}_{\gamma/2}(\mathbb{R}^{d})}\\
&\hspace{0cm}\leq C(f_0,\varepsilon,\delta) -\frac{c_0}{2}\big\|e^{r\langle \cdot \rangle^{\alpha}}\big(1+(-\Delta)\big)^{s/2}f\big\|^{2}_{L^{2}_{\gamma/2}(\mathbb{R}^{d})}\,.
\end{align*} 
We used Young's inequality to control all the linear terms, for instance,
\begin{align*}
C(f_0,\delta,\varepsilon)\big\|e^{r\langle \cdot \rangle^{\alpha}}&\big(1+(-\Delta)\big)^{s/2}f\big\|_{L^{2}(\mathbb{R}^{d})}\\
&\leq \delta^{-\frac{1}{2}}C(f_0,\delta,\varepsilon)^{2} + \delta^{\frac{1}{2}}\big\|e^{r\langle \cdot \rangle^{\alpha}}\big(1+(-\Delta)\big)^{s/2}f\big\|^{2}_{L^{2}(\mathbb{R}^{d})}\,,
\end{align*}
and consolidated all constants in $C(f_0,\varepsilon,\delta)$ and $C(b,\varepsilon)$.  Additionally, for the last inequality we chose $\varepsilon$ and, then, $\delta:=\delta(\varepsilon)$ sufficiently small.  This estimate proves the result.
\end{proof}
By an inductive argument, we next prove propagation of exponentially-tailed regularity for an arbitrary order of regularity.

\begin{theorem}[Propagation of exponentially-tailed regularity]\label{TFinal1}
Assume $\gamma\in[0,1]$, $k\geq0$, and $\alpha\in(0,1]$.  Also,
\begin{equation*}
e^{a_o\langle v \rangle^{\alpha}}\big(1 + (-\Delta) \big)^{\frac{k}{2}}f_0 \in L^{2}(\mathbb{R}^{d}). 
\end{equation*} 
Then, for the solution $f(t,v)$ of the Boltzmann equation, there exists $a\in(0,a_o]$ such that for any $r\in(0,a)$  it holds
\begin{equation*}
\big\|e^{r\langle\cdot\rangle^{\alpha}}\big(1 + (-\Delta) \big)^{\frac{k}{2}}f(t,\cdot)\big\|_{L^{2}(\mathbb{R}^{d})}\leq C(f_0) + \big\|e^{r\langle\cdot\rangle^{\alpha}}\,\big(1 + (-\Delta) \big)^{\frac{k}{2}}f_0\big\|_{L^{2}(\mathbb{R}^{d})}.
\end{equation*}
The constant depends as $C(f_0):=C(f_0,\gamma,k,r,\alpha)$ and lower order $H^{k^{-}}_{exp}(\mathbb{R}^{d})$ Sobolev norms of $f_0$.   
\end{theorem}
\begin{proof}
Here we only show the proof for $d\geq3$.  The same argument with slight modifications will do the job for the case $d=2$.  The proof uses induction by first considering the case $k\in\mathbb{N}$.  When $k=\{0,1\}$ the result follows using Theorem \ref{T1} or \ref{T2} for $k=0$, and Proposition \ref{PSmooth1} for $k=1$.  For $k\geq2$, assume the validity of the result for $k$ and conclude it for $k+1$.  Write $k=2n+i$, with $n\in\mathbb{N}$ and $i\in\{0,1\}$, and consider the operator $\text{D}^{k}:=(-\Delta)^{n}\nabla^{i}$.  Using classic Leibniz formula for integer differentiation, it is not difficult to check that
\begin{align*}
\text{D}^{k}Q(f,f) = & Q(\text{D}^{k}f,f) + Q(f,\text{D}^{k}f)\\
&\hspace{0cm}+ \mathop{\sum\sum}_{|j_1 +j_{2}|\leq (k-i-1)} C_{j_1,j_2}\Big(Q\big(\partial^{j_{1}}\nabla^{i}f, \partial^{j_{2}}f \big)+Q\big(\partial^{j_{1}}f, \partial^{j_{2}}\nabla^{i}f \big)\Big),
\end{align*}
for some coefficients $C_{j_1,j_2}$ and with $j_1$ and $j_2$ multi-indexes with order ranging as described in the sums.  Thus, applying $\big(1+(-\Delta)\big)^{\frac{1}{2}}\text{D}^{k}$ to the Boltzmann equation one gets
\begin{align}\label{SRe1}
\begin{split}
\partial_{t}\big(1&+(-\Delta)\big)^{\frac{1}{2}}\text{D}^{k}f = \big(1+(-\Delta)\big)^{\frac{1}{2}}Q(\text{D}^{k}f,f)+\big(1+(-\Delta)\big)^{\frac{1}{2}}Q(f,\text{D}^{k}f)\\
&\hspace{-.5cm}+ \mathop{\sum\sum}_{|j_1 +j_{2}|\leq (k-i-1)} C_{j_1,j_2}\big(1+(-\Delta)\big)^{\frac{1}{2}}\Big(Q\big(\partial^{j_{1}}\nabla^{i}f, \partial^{j_{2}}f \big)+ Q\big(\partial^{j_{1}}f, \partial^{j_{2}}\nabla^{i}f \big)\Big).
\end{split}
\end{align}
Let us control the terms on the right side of \eqref{SRe1} starting with the sum.  Using the commutator Lemmas \ref{L2Hard} and \ref{L1Hard}
\begin{align*}
\big(1+(-\Delta)\big)^{\frac{1}{2}}Q\big(\partial^{j_{1}}\nabla^{i}f, \partial^{j_{2}}f \big) = &Q\big(\big(1+(-\Delta)\big)^{\frac{1}{2}}\partial^{j_{1}}\nabla^{i}f, \partial^{j_{2}}f \big)\\
& + \mathcal{I}^{+}\big( \partial^{j_{1}}\nabla^{i}f, \partial^{j_{2}}f \big) - \mathcal{I}^{-}\big( \partial^{j_{1}}\nabla^{i}f, \partial^{j_{2}}f \big).
\end{align*}
Therefore, for $r\in(0,a)$ and $\alpha\in(0,1]$ it holds that
\begin{align*}
\Big\|e^{r\langle\cdot\rangle^{\alpha}}&\big(1+(-\Delta)\big)^{\frac{1}{2}}Q\big(\partial^{j_{1}}\nabla^{i}f, \partial^{j_{2}}f \big) \Big\|_{L^{2}(\mathbb{R}^{d})} \\
&\leq \Big\|e^{r\langle\cdot\rangle^{\alpha}}Q\big(\big(1+(-\Delta)\big)^{\frac{1}{2}}\partial^{j_{1}}\nabla^{i}f, \partial^{j_{2}}f \big) \Big\|_{L^{2}(\mathbb{R}^{d})}
+\Big\|e^{r\langle\cdot\rangle^{\alpha}}\mathcal{I}^{+}\big(\partial^{j_{1}}\nabla^{i}f, \partial^{j_{2}}f \big) \Big\|_{L^{2}(\mathbb{R}^{d})}\\
& \qquad+\Big\|e^{r\langle\cdot\rangle^{\alpha}}\mathcal{I}^{-}\big(\partial^{j_{1}}\nabla^{i}f, \partial^{j_{2}}f \big) \Big\|_{L^{2}(\mathbb{R}^{d})}\\
&\hspace{0cm}\leq C\|b\|_{L^{1}(\mathbb{S}^{d-1})}\big\|e^{r^{+}\langle\cdot\rangle^{\alpha}} \big(1+(-\Delta)\big)^{\frac{1}{2}}\partial^{j_{1}}\nabla^{i}f \big\|_{L^{2}(\mathbb{R}^{d})}\big\|e^{r^{+}\langle\cdot\rangle^{\alpha}}\partial^{j_{2}}f \big\|_{L^{2}(\mathbb{R}^{d})}.
\end{align*}
Observe that in this $L^{2}-L^{2}$ estimate we paid a small price in the weight $e^{r\langle\cdot\rangle^{\alpha}}\rightarrow e^{r^{+}\langle\cdot\rangle^{\alpha}}$ since the natural estimate from Young's inequality is an $L^{2}-L^{1}$ estimate.  Furthermore, using the commutator Lemma \ref{ApLMw}, it follows for $\alpha\in(0,1]$ that
\begin{align*}
\big\|e^{r^{+}\langle\cdot\rangle^{\alpha}} \big(1+(-\Delta)\big)^{\frac{1}{2}}\partial^{j_{1}}\nabla^{i}f \big\|_{L^{2}(\mathbb{R}^{d})} &\lesssim \big\|e^{r^{+}\langle\cdot\rangle^{\alpha}} \big(1+(-\Delta)\big)^{\frac{|j_1|+i+1}{2}} f \big\|_{L^{2}(\mathbb{R}^{d})}\\
&\lesssim \big\|e^{r^{+}\langle\cdot\rangle^{\alpha}} \big(1+(-\Delta)\big)^{\frac{k}{2}}f \big\|_{L^{2}(\mathbb{R}^{d})}\,.
\end{align*}
For the last inequality we used that $|j_{1}|\leq k-i-1$.  Choosing $r^{+}<a$ one has, by induction hypothesis, that
\begin{equation}\label{SRe2}
\Big\|e^{r\langle\cdot\rangle^{\alpha}}\big(1+(-\Delta)\big)^{\frac{1}{2}}Q\big(\partial^{j_{1}}\nabla^{i}f, \partial^{j_{2}}f \big) \Big\|_{L^{2}(\mathbb{R}^{d})} \leq C(f_0),
\end{equation}
with constant depending on $k^{th}$-Sobolev regularity of $f_0$.  This controls the sum.  In the same fashion, using the commutator lemmas, it follows for the second term
\begin{equation*}
\big(1+(-\Delta)\big)^{\frac{1}{2}}Q(f,\text{D}^{k}f) = Q\big(\big(1+(-\Delta)\big)^{\frac{1}{2}}f,\text{D}^{k}f\big) + \mathcal{I}^{+}\big(f,\text{D}^{k}f\big) - \mathcal{I}^{-}\big(f,\text{D}^{k}f\big)\,.
\end{equation*}
As a consequence,
\begin{align}\label{SRe3}
\begin{split}
\big\|e^{r\langle\cdot\rangle^{\alpha}}&\big(1+(-\Delta)\big)^{\frac{1}{2}}Q(f,\text{D}^{k}f)\big\|_{L^{2}(\mathbb{R}^{d})}\\
&\lesssim C\|b\|_{L^{1}(\mathbb{S}^{d-1})}\big\|e^{r^{+}\langle\cdot\rangle^{\alpha}} \big(1+(-\Delta)\big)^{\frac{1}{2}} f \big\|_{L^{2}(\mathbb{R}^{d})}\big\|e^{r^{+}\langle\cdot\rangle^{\alpha}}\text{D}^{k}f \big\|_{L^{2}(\mathbb{R}^{d})}\leq C(f_0).
\end{split}
\end{align}
In the last inequality we used the induction hypothesis again (valid for $r^{+}<a$).  Finally, following the same argument as in the proof of Proposition \ref{PSmooth1} (with $s=1$)\footnote{Indeed, the careful reader observes that such argument is bilinear.}, the first term on the right side can be estimated as
\begin{align}\label{SRe4}
\begin{split}
\int_{\mathbb{R}^{d}}e^{r\langle v \rangle^{\alpha}}&\big(1+(-\Delta)\big)^{\frac{1}{2}}Q\big(\text{D}^{k}f,f\big)(v)\,e^{r\langle v \rangle^{\alpha}}\big(1+(-\Delta)\big)^{\frac{1}{2}}\text{D}^{k}f(v)\text{d}v\\
&\leq  C(f_0) - \frac{c_o}{2}\big\|e^{r\langle \cdot \rangle^{\alpha}}\big(1+(-\Delta)\big)^{\frac{1}{2}}\text{D}^{k}f\big\|^{2}_{L^{2}(\mathbb{R}^{d})}.
\end{split}
\end{align}
As a consequence, multiplying equation \eqref{SRe1} by $e^{2r\langle\cdot\rangle}\big(1+(-\Delta)\big)^{\frac{1}{2}}\text{D}^{k}f$ and using Cauchy-Schwarz inequality together with estimates \eqref{SRe2}, \eqref{SRe3} and \eqref{SRe4}, one finds a constant $C(f_0)$ depending on lower order norms of $f_0$ such that
\begin{equation*}
\frac{\text{d}}{\text{d}t}\big\|e^{r\langle \cdot \rangle^{\alpha}}\big(1+(-\Delta)\big)^{\frac{1}{2}}\text{D}^{k}f\big\|^{2}_{L^{2}(\mathbb{R}^{d})}\leq C(f_0) - \frac{c_o}{2}\big\|e^{r\langle \cdot \rangle^{\alpha}}\big(1+(-\Delta)\big)^{\frac{1}{2}}\text{D}^{k}f\big\|^{2}_{L^{2}(\mathbb{R}^{d})}.
\end{equation*}
Then,
\begin{equation*}
\big\|e^{r\langle \cdot \rangle^{\alpha}}\big(1+(-\Delta)\big)^{\frac{1}{2}}\text{D}^{k}f(t)\big\|_{L^{2}(\mathbb{R}^{d})}\leq \max\Big\{\sqrt{\frac{2}{c_o}C(f_0)},\big\|e^{r\langle \cdot \rangle^{\alpha}}\big(1+(-\Delta)\big)^{\frac{1}{2}}\text{D}^{k}f_0\big\|_{L^{2}(\mathbb{R}^{d})}\Big\}.
\end{equation*}
This proves the case $k\in\mathbb{N}$ after observing that
\begin{equation*}
\big\|e^{r\langle \cdot \rangle^{\alpha}}\big(1+(-\Delta)\big)^{\frac{k+1}{2}}f(t)\big\|_{L^{2}(\mathbb{R}^{d})}\sim \big\|e^{r\langle \cdot \rangle^{\alpha}}\big(1+(-\Delta)\big)^{\frac{1}{2}}\text{D}^{k}f(t)\big\|_{L^{2}(\mathbb{R}^{d})} +\big\|e^{r\langle \cdot \rangle^{\alpha}}f(t)\big\|_{L^{2}(\mathbb{R}^{d})}.
\end{equation*}
When $k\in\mathbb{R}^{+}\setminus\mathbb{N}$, write $k=\lfloor k \rfloor + s $ with $s\in(0,1)$.  Then, by previous argument
\begin{equation*}
\big\|e^{r\langle \cdot \rangle^{\alpha}}\text{D}^{\lfloor k \rfloor}f(t)\big\|_{L^{2}(\mathbb{R}^{d})}\leq C(f_0)+\big\|e^{r\langle \cdot \rangle^{\alpha}}\text{D}^{\lfloor k \rfloor}f_0\big\|_{L^{2}(\mathbb{R}^{d})}.
\end{equation*}
Perform, again, previous argument for the operator $\big(1+(-\Delta)\big)^{\frac{s}{2}}\text{D}^{\lfloor k \rfloor}$ to conclude the proof.
\end{proof}
We complete this section with a result in $L^\infty$ based spaces.

\begin{corollary}\label{CFinal2}
Assume $\gamma\in[0,1]$, $k\geq0$, and $\alpha\in(0,1]$.  Also,
\begin{equation*}
e^{a_o\langle v \rangle^{\alpha}}\big(1 + (-\Delta) \big)^{\frac{k}{2}}f_0 \in L^{\infty}(\mathbb{R}^{d}). 
\end{equation*} 
Then, for the solution $f(t,v)$ of the Boltzmann equation, there exists $a\in(0,a_o]$ such that for any $r\in(0,a)$ it holds that
\begin{equation*}
\big\|e^{r\langle\cdot\rangle^{\alpha}}\big(1 + (-\Delta) \big)^{\frac{k}{2}}f(t,\cdot)\big\|_{L^{\infty}(\mathbb{R}^{d})}\leq C(f_0) + \big\|e^{r\langle\cdot\rangle^{\alpha}}\,\big(1 + (-\Delta) \big)^{\frac{k}{2}}f_0\big\|_{L^{\infty}(\mathbb{R}^{d})}.
\end{equation*}
The constant depends as $C(f_0):=C(f_0,\gamma,k,r,\alpha)$ and lower order $H^{k^{-}}_{exp}(\mathbb{R}^{d})$ Sobolev norms of $f_0$.   
\end{corollary}
\begin{proof}
This is a direct consequence of Theorem \ref{TFinal1} and the techniques presented in Theorems \ref{T1} and \ref{T2}, for hard potentials and Maxwell molecules respectively, to show $L^{\infty}$-norm propagation based on the $L^{2}$-norm propagation.
\end{proof}
\section{Convergence towards equilibrium and decomposition theorem}
In this last section we are interested in studying the convergence of the solution of the homogeneous Boltzmann equation with hard potentials towards the Maxwellian distribution and its propagation of smoothness and singularities.

Let us start by introducing the relative entropy
\begin{equation*}
\mathcal{H}(f | \mathcal{M}_f ) = \int_{\mathbb{R}^{d}}f \log\Big(\frac{f}{\mathcal{M}_f}\Big)\text{d}v\,,
\end{equation*}
and the entropy production
\begin{equation*}
\mathcal{D}(f)=\tfrac{1}{4}\int_{\mathbb{R}^{2d}}\int_{\mathbb{S}^{d-1}}\big(f' f'_{*} - f f_{*}\big)\log\Big(\frac{f' f'_*}{f f_*}\Big)B(u,\widehat{u}\cdot\sigma)\text{d}\sigma\text{d}v_*\text{d}v\,.
\end{equation*}
The function $\mathcal{M}_{f}$ is the thermodynamical equilibrium
\begin{equation*}
\mathcal{M}_{f}(v):=\frac{\rho_f}{(2\pi\, T_f)^{\frac{d}{2}}}e^{-\frac{|v-\mu_f|^{2}}{2T_f}}\,,
\end{equation*}
where $\rho_f:=\int f$ is the density, $\mu_f:=\frac{1}{\rho}\int f v$ is the momentum, and $T_f:=\frac{1}{d\rho}\int f |v-\mu|^{2}$ is the temperature associated to $f$.  Clearly, for the Boltzmann flow $\mathcal{M}_{f}=\mathcal{M}_{f_0}$ and $$0\leq\mathcal{H}(f | \mathcal{M}_{f})\leq \mathcal{H}(f_0 | \mathcal{M}_{f_0})\,.$$
In the sequel, and without loss of generality, we assume $\rho=1$ and $\mu=0$.  We refer the reader to \cite{V} for additional details and references.

The proof of the convergence towards the equilibrium consists of two steps. First, we use the Csisz\'ar-Kullback-Pinsker inequality  and the  relative entropy estimate (Proposition \ref{transitory-relaxation}) to show that after some time, the solution becomes close to the equilibrium. Then in the second step, we use properties of the linearized operator (Proposition \ref{linearized-spectrum})  to prove that the solution converges to the equilibrium exponentially fast.

We start by addressing the first main ingredient mentioned above - the relative entropy estimate (Proposition \ref{transitory-relaxation}). Such estimate is proved thanks to an entropy dissipation estimate (Theorem \ref{Dissipation-entropy}) and an estimate of entropic moments (Proposition \ref{entropic-moments}).

Then later, we will focus on the second ingredient of the strategy, which is the spectral gap result for the linearized operator on an enlarged space (Proposition \ref{linearized-spectrum}).

We start by proving the dissipation of entropy estimate related to the nonlinear dynamics.
\subsection{Dissipation of Entropy}

In the following theorem, we relax conditions of Theorem 4.1 in \cite{V} with respect to the regularity constraints.

\begin{theorem}\label{Dissipation-entropy}
Let the scattering kernel satisfy
\begin{equation}\label{K_B}
B( u ,\widehat{u}\cdot\sigma)\geq K_{B}\min\big\{|u|^{\gamma},|u|^{-\beta}\big\}\,,\quad \gamma\geq0,\;\; \beta\geq 0,
\end{equation}
and let $f\geq0$ be a function with sufficiently high number of moments and entropic moments, and such that
\begin{equation*}
f(v)\geq K_o\, e^{-A_o|v|^{q_o}}\,,\quad K_o>0,\,A_o>0,\,q_o\geq2\,.
\end{equation*}
Then, for any $\varepsilon\in(0,1)$ we have:\\

\textit{(i)} If $f\in L^{p}(\mathbb{R}^{d})$, with $p\in(1,\infty]$,
\begin{equation*}
\mathcal{D}(f)\geq A_{\varepsilon,p}(f)\mathcal{H}(f | \mathcal{M}_{f})^{(1+\varepsilon)(1+\frac{\gamma p'}{d})}\,,
\end{equation*}
where the constant $A_{\varepsilon,p}(f)$ is given in \eqref{EPe5}.\\

\textit{(ii)} If $f\in L\log L(\mathbb{R}^{d})$, then
\begin{equation*}
\mathcal{D}(f)\geq A_{\varepsilon,L\log L}(f)\,\mathcal{H}( f | \mathcal{M}_{f})^{(1+\varepsilon)(1+\frac{\gamma}{d})}e^{-\frac{2\gamma\widetilde{K}_{L\log L}(f)}{d K_{\varepsilon}(f)\mathcal{H}(f|\mathcal{M}_{f})^{1+\varepsilon}}}\,,
\end{equation*}
where the constant $A_{\varepsilon, L\log L}(f)$ is given in \eqref{EPe8}.
\end{theorem}
\begin{remark}
The fact that for all $t\geq t_o>0$
\begin{equation*}
f(t,v)\geq K_o\, e^{-A_o|v|^2}\,,\quad K_o>0,\;\;A_o>0\,,
\end{equation*}
for positive solutions $f(t,v)$ of the homogeneous Boltzmann equation satisfying only \eqref{Initial-Data} was proven in the classical paper \cite{PW}.  The constants $K_o$ and $A_o$ depend on the mass, energy, entropy of the initial data, as well as $t_o$, but are uniform for all $t\geq t_o$.  Although these constants are explicit in terms of these quantities, the dependence is quite complicated.
\end{remark} 
\begin{proof}
Note that
\begin{align}\label{EPe1}
\begin{split}
B(u,\widehat{u}\cdot\sigma) &= \Big(K_{B}\,R^{\gamma}\,\text{1}_{|u|\leq R} + B(u,\widehat{u}\cdot\sigma)\Big) - K_{B}\,R^{\gamma}\,\text{1}_{|u|\leq R}\\
&\geq K_{B}\Big(R^{\gamma}\,\text{1}_{|u|\leq R} + \min\big\{|u|^{\gamma},|u|^{-\beta}\big\} \Big) - K_{B}\,R^{\gamma}\,\text{1}_{|u|\leq R}\\
&\geq K_{B}\,R^{\gamma}\langle u \rangle^{-\beta} - K_{B}\,R^{\gamma}\,\text{1}_{|u|\leq R}\,.
\end{split}
\end{align}
As a consequence, $\mathcal{D}(f) \geq \mathcal{D}_{1}(f) - \mathcal{D}_{2}(f)$ where $\mathcal{D}_{i}(f)$, with $i\in\{1,2\}$, corresponds to each term on the right side of \eqref{EPe1} respectively.  Using \cite[Theorem 3.1]{V} one has
\begin{equation}\label{EPe2}
\mathcal{D}_{1}(f)\geq R^{\gamma} K_{\varepsilon}(f) \mathcal{H}(f | \mathcal{M}_{f} )^{1+\varepsilon}\,.
\end{equation}
An explicit form for $K_{\varepsilon}(f)$ can be found in \cite{V}.  We just mention here that it depends on mass and temperature (energy) of $f$, the parameters $A_o, K_o, q_o$, and
\begin{equation*}
\int_{\mathbb{R}^{d}} f(v)\langle v \rangle^{2+\frac{2+\beta}{\varepsilon}}\big|\log(f)\big|\text{d}v\,,\qquad \int_{\mathbb{R}^{d}} f(v)\langle v \rangle^{2+q_o+\frac{2+\beta}{\varepsilon}}\text{d}v\,.
\end{equation*}
For $\mathcal{D}_{2}(f)$ one can proceed in similar fashion to the proof of \cite[Theorem 3.1]{V} to obtain
\begin{align}\label{EPe3}
\begin{split}
\mathcal{D}_{2}(f)\leq 2K_{B}&|\mathbb{S}^{d-1}|R^{\gamma}\bigg(4\int_{\{|u|\leq R\}}f \log(f) f_{*}\text{d}v_{*}\text{d}v \\
&\hspace{2cm} + 2^{\frac{q_o}{2}+1}\Big(\log\Big(\frac{1}{K_o}\Big)+A_o\Big)\int_{\{|u|\leq R\}}\langle v \rangle^{q_0} f f_{*}\text{d}v_{*}\text{d}v\bigg)\,. 
\end{split}
\end{align}
\textit{Case $f \in L^{p}(\mathbb{R}^{d})$:}  Since, by H\"{o}lders inequality,
\begin{equation*}
\int_{\{|u|\leq R\}}f(v)\text{d}v \leq \big|\tfrac{1}{d}\,\mathbb{S}^{d-1}\big|^{\frac{1}{p'}}R^{\frac{d}{p'}}\|f\|_{L^{p}(\mathbb{R}^{d})}
\end{equation*}
one concludes from \eqref{EPe3} that
\begin{equation}\label{EPe4}
\mathcal{D}_{2}(f)\leq CK_{B}R^{\gamma+\frac{d}{p'}}\|f\|_{L^{p}(\mathbb{R}^{d})}\Big(\int_{\mathbb{R}^{d}}f\big|\log(f)\big|\text{d}v + \int_{\mathbb{R}^{d}}f\langle v \rangle^{q_o}\text{d}v\Big) =: R^{\gamma+\frac{d}{p'}}\widetilde{K}_p(f)\,,
\end{equation}
where $C>0$ is a universal constant.  Gathering \eqref{EPe2} and \eqref{EPe4} it follows that
\begin{equation*}
\mathcal{D}(f)\geq R^{\gamma}\Big(K_{\varepsilon}(f)\mathcal{H}(f | \mathcal{M}_{f})^{1+\varepsilon} - R^{\frac{d}{p'}}\widetilde{K}_p(f)\Big). 
\end{equation*}
Choosing $R^{\frac{d}{p'}}:=\frac{K_{\varepsilon}(f)}{2\widetilde{K}_p(f)}\mathcal{H}(f | \mathcal{M}_{f})^{1+\varepsilon}$ it follows that
\begin{equation}\label{EPe5}
\mathcal{D}(f)\geq \frac{K_{\varepsilon}(f)^{1+\frac{\gamma p'}{d}}}{2^{1+\frac{\gamma p'}{d}}\widetilde{K}_p(f)^{\frac{\gamma p'}{d}}}\mathcal{H}( f | \mathcal{M}_{f})^{(1+\varepsilon)(1+\frac{\gamma p'}{d})}=:A_{\varepsilon,p}(f)\,\mathcal{H}( f | \mathcal{M}_{f})^{(1+\varepsilon)(1+\frac{\gamma p'}{d})}\,.
\end{equation}
\textit{Case $f\in L\log L (\mathbb{R}^{d})$}:  Using the generalized Young's inequality, see for example the proof of \cite[Proposition A.1]{ALods}, one concludes that
\begin{equation}
\int_{\{|u|\leq R\}} f(v) \text{d}v \leq \frac{\int_{\mathbb{R}^{d}} f |\log(f)|\text{d}v}{\mathcal{W}\Big(\frac{\int_{\mathbb{R}^{d}} f|\log(f)|\text{d}v}{|\{|v|\leq R\}|}\Big)}\,,
\end{equation}
where $\mathcal{W}(x)$ is the Lambert function, that is, $\mathcal{W}^{-1}(x)=x\,e^{x}$.  Therefore, recalling \eqref{EPe3}
\begin{align}\label{EPe6}
\begin{split}
\mathcal{D}_{2}(f) 
&\leq CK_{B}R^{\gamma}\Big(\int_{\mathbb{R}^{d}}f\big|\log(f)\big|\text{d}v + \int_{\mathbb{R}^{d}}f\langle v \rangle^{q_o}\text{d}v\Big)\frac{\int_{\mathbb{R}^{d}}f\big|\log(f)\big|\text{d}v}{\mathcal{W}\Big(\frac{\int_{\mathbb{R}^{d}} f|\log(f)|\text{d}v}{|\{|v|\leq R\}|}\Big)} \\
&=: \frac{R^{\gamma}\,\widetilde{K}_{L\log L}(f)}{\mathcal{W}\Big(\frac{\int_{\mathbb{R}^{d}} f|\log(f)|\text{d}v}{|\{|v|\leq R\}|}\Big)}\,.
\end{split}
\end{align}
Using \eqref{EPe2} and \eqref{EPe6}
\begin{equation}\label{EPe7}
\mathcal{D}(f)\geq R^{\gamma}\bigg(K_{\varepsilon}(f)\mathcal{H}(f | \mathcal{M}_{f})^{1+\varepsilon} - \frac{\widetilde{K}_{L\log L}(f)}{\mathcal{W}\Big(\frac{\int_{\mathbb{R}^{d}} f|\log(f)|\text{d}v}{|\{|v|\leq R\}|}\Big)}\bigg). 
\end{equation}
Now choose, $R>0$ such that
\begin{equation*}
\frac{\widetilde{K}_{L\log L}(f)}{\mathcal{W}\Big(\frac{\int_{\mathbb{R}^{d}} f|\log(f)|\text{d}v}{|\{|v|\leq R\}|}\Big)} = \frac{K_{\varepsilon}(f)}{2}\mathcal{H}(f | \mathcal{M}_{f})^{1+\varepsilon}\,,
\end{equation*}
or more precisely,
\begin{equation*}
R^{d} := \frac{d\int_{\mathbb{R}^{d}}f|\log(f)|\text{d}v}{2\big|\mathbb{S}^{d-1}\big|\widetilde{K}_{L\log L}(f)}\,K_{\varepsilon}(f)\mathcal{H}(f | \mathcal{M}_{f})^{1+\varepsilon}e^{-\frac{2\tilde{K}_{L\log L}(f)}{K_{\varepsilon}(f)\mathcal{H}( f | \mathcal{M}_{f})^{1+\varepsilon}}}
\end{equation*}
we obtain
\begin{align}\label{EPe8}
\begin{split}
\mathcal{D}(f)
& \geq \frac{R^{\gamma}}{2}K_{\varepsilon}(f)\mathcal{H}( f | \mathcal{M}_{f})^{1+\varepsilon}\\
& = \frac{K_{\varepsilon}(f)^{1+\frac{\gamma}{d}}}{2^{1+\frac{\gamma}{d}}}\left(\frac{d\int_{\mathbb{R}^{d}}f|\log(f)|\text{d}v}{|\mathbb{S}^{d-1}|\widetilde{K}_{L\log L}(f)}\right)^{\frac{\gamma}{d}}\mathcal{H}( f | \mathcal{M}_{f})^{(1+\varepsilon)(1+\frac{\gamma}{d})}e^{-\frac{2\gamma\widetilde{K}_{L\log L}(f)}{d K_{\varepsilon}(f)\mathcal{H}(f|\mathcal{M}_{f})^{1+\varepsilon}}}\\
& =: A_{\varepsilon,L\log L}(f)\,\mathcal{H}( f | \mathcal{M}_{f})^{(1+\varepsilon)(1+\frac{\gamma}{d})}e^{-\frac{2\gamma\widetilde{K}_{L\log L}(f)}{d K_{\varepsilon}(f)\mathcal{H}(f|\mathcal{M}_{f})^{1+\varepsilon}}}\,.
\end{split}
\end{align}
\end{proof}

\subsection{Entropic moments and transitory relaxation}
We next  prove a priori estimates needed for the Proposition \ref{transitory-relaxation}), we first need to obtain propagation of entropic moments.

\begin{proposition}\label{entropic-moments}
Assume $f$ is solution of the homogeneous Boltzmann problem for hard potentials with initial data satisfying \eqref{Initial-Data}.  Then, for any $t_o>0$ and $s\in[0,\infty)$
\begin{equation*}
\sup_{t\geq t_o}\int_{\mathbb{R}^{d}}f(v)\langle v \rangle^{s}|\log(f)| dv < \infty\,.
\end{equation*}
\end{proposition}
\begin{proof}
The case $s=0$ is clear.  Thus, multiply the equation by $\langle{v}\rangle^{s}|\log(f)|$, with $s>0$, to obtain
\begin{align*}
\frac{\text{d}}{\text{d}t}\int_{\mathbb{R}^{d}}f(v)\langle{v}\rangle^{s}|\log(f)|\text{d}v = \int_{\mathbb{R}^{d}}&Q(f,f)(v)\langle{v}\rangle^{s}|\log(f)|\text{d}v\\
&+\int_{\mathbb{R}^{d}}Q(f,f)(v)\langle{v}\rangle^{s}\text{sgn}\big(\log(f)\big)\text{d}v\,.
\end{align*}
For the second term on the right side it readily follows that
\begin{equation}\label{EMe1}
\Big|\int_{\mathbb{R}^{d}}Q(f,f)(v)\langle{v}\rangle^{s}\text{sgn}\big(\log(f)\big)\text{d}v\Big|\leq 2\|b\|_{L^{1}(\mathbb{S}^{d-1})}\|f\langle v \rangle^{\gamma+s}\|_{L^{1}(\mathbb{R}^{d})}\|f\langle v \rangle^{\gamma+s}\|_{L^{1}(\mathbb{R}^{d})}\,.
\end{equation}
The first term has a positive and negative parts.  For the negative, one concludes using the uniform boundedness of entropy that  
\begin{equation}\label{EMe2}
\int_{\mathbb{R}^{d}}Q^{-}(f,f)(v)\langle{v}\rangle^{s}|\log(f)|\text{d}v\geq c_{o}\int_{\mathbb{R}^{d}}f(v)\langle{v}\rangle^{s+\gamma}|\log(f)|\text{d}v\,,
\end{equation}
where $c_o>0$ is a constant depending of the initial entropy, mass and energy.  For the positive part, one has
\begin{align}\label{EMe3}
\begin{split}
\int_{\mathbb{R}^{d}}Q^{+}(f,f)(v)\langle{v}\rangle^{s}|& \log(f)|\text{d}v=\int_{\mathbb{R}^{2d}}f(v)f(v_*)|u|^{\gamma}\int_{\mathbb{S}^{d-1}}\big|\log(f(v'))\big|\langle v ' \rangle^{s}b(\widehat{u}\cdot\sigma)\text{d}\sigma\text{d}v_{*}\text{d}v\\
&\leq C_{s}\int_{\mathbb{R}^{2d}}\sum_{(i,j)}f(v)\langle v \rangle^{i}f(v_*)\langle v_* \rangle^{j}\int_{\mathbb{S}^{d-1}}\big|\log(f(v'))\big|b(\widehat{u}\cdot\sigma)\text{d}\sigma\text{d}v_{*}\text{d}v\,.\\
\end{split}
\end{align}
The sum is performed on $(i,j)\in\{(s+\gamma,0),(s,\gamma),(\gamma,s),(0,s+\gamma)\}$.  We estimate the right side in \eqref{EMe3} controlling the integral in the integration sets $\{f'\leq 1\}$ and $\{f'>1\}$ separately.  For the former recall the classical result proved in \cite{PW}: for any $t_o>0$ there exists positive $K_o, A_o$ depending only on the initial mass, energy, entropy and $t_o$ such that
\begin{equation*}
f(t,v)\geq K_o e^{-A_o|v|^{2}}\,,\qquad t\geq t_o>0\,.
\end{equation*}
As a consequence,
\begin{equation}\label{EMe4}
\int_{\mathbb{S}^{d-1}}\big|\log(f(v'))\big|\text{1}_{f'\leq1}b(\widehat{u}\cdot\sigma)\text{d}\sigma\leq \|b\|_{L^{1}(\mathbb{S}^{d-1})}\Big(\log\Big(\frac{1}{K_o}\Big)+A_o\Big)\langle v \rangle^{2}\langle v_{*}\rangle^{2}\,.
\end{equation}
The latter set $\{f'>1\}$ is trickier.  We concentrate first on the combination $(i,j)=(0,s + \gamma)$ since the others follow a simpler argument.  First, we fix $\varepsilon>0$ and use the usual angular split \eqref{DAK}.  In each component we use the generalized Young's inequality
\begin{equation*}
xy\leq x\log x - x + e^{y}\,,\qquad x\geq0,\; y\in\mathbb{R}\,,
\end{equation*}
in slightly, but crucially, different way.  For the good part, the one with $b^{\varepsilon}_{1}\big(\cos\theta \big)$, we choose $x=f(v)$ and $y=|\log(f')|\text{1}_{f'>1}$ to get
\begin{equation*}
f\big|\log(f')\big|\text{1}_{f'>1} \leq f\log(f) - f + f'\,.
\end{equation*}
Thus,
\begin{align}\label{EMe5}
\begin{split}
\int_{\mathbb{R}^{2d}}f(v_*)&\langle v_* \rangle^{s+\gamma}f(v)\int_{\mathbb{S}^{d-1}}\big|\log(f(v'))\big|\text{1}_{f'>1}b^{\varepsilon}_1(\widehat{u}\cdot\sigma)\text{d}\sigma\text{d}v_{*}\text{d}v\\
&\leq\|b^{\varepsilon}_{1}\|_{L^{1}(\mathbb{S}^{d-1})}\|f\langle v \rangle ^{s+\gamma}\|_{L^{1}(\mathbb{R}^{d})}\|f\log(f)\|_{L^{1}(\mathbb{R}^{3})}\\
&\hspace{1cm}+2^{d}\Big\|\frac{b^{\varepsilon}_{1}(\cos(\theta))}{1-\cos(\theta)}\Big\|_{L^{1}(\mathbb{S}^{d-1})}\|f\langle v \rangle ^{s+\gamma}\|_{L^{1}(\mathbb{R}^{d})}\|f\|_{L^{1}(\mathbb{R}^{d})},
\end{split}
\end{align}
where we have used the singular change of variable $v\rightarrow v'$ in the integral containing $f'$.  For the bad part, the one with $b^{\varepsilon}_{2}\big(\cos\theta \big)$, we choose $x=\langle v_*\rangle^{s+\gamma} f(v_*)$ and $y=|\log(f')|\text{1}_{f'>1}$ to get
 \begin{equation*}
\langle v_*\rangle^{s+\gamma}f_*\big|\log(f')\big|\text{1}_{f'>1} \leq \langle v_*\rangle^{s+\gamma}f_*\log\big(\langle v_*\rangle^{s+\gamma}f_*\big) - \langle v_*\rangle^{s+\gamma}f_* + f'\,.
\end{equation*}
Therefore,
\begin{align}\label{EMe6}
\begin{split}
\int_{\mathbb{R}^{2d}} & f(v_*)\langle v_* \rangle^{s+\gamma}f(v)\int_{\mathbb{S}^{d-1}}\big|\log(f(v'))\big|\text{1}_{f'>1}b^{\varepsilon}_2(\widehat{u}\cdot\sigma)\text{d}\sigma\text{d}v_{*}\text{d}v\\
&\leq\|b^{\varepsilon}_{2}\|_{L^{1}(\mathbb{S}^{d-1})}\Big(\|f\langle v \rangle ^{s+\gamma}\log(f)\|_{L^{1}(\mathbb{R}^{d})}\|f\|_{L^{1}(\mathbb{R}^{3})}\\
&\hspace{0cm} +(s+\gamma)\|f\langle v \rangle ^{s+\gamma}\log\langle v \rangle\|_{L^{1}(\mathbb{R}^{d})}\|f\|_{L^{1}(\mathbb{R}^{3})}\Big)+2^{d}\Big\|\frac{b^{\varepsilon}_{2}(\cos(\theta))}{1+\cos(\theta)}\Big\|_{L^{1}(\mathbb{S}^{d-1})}\|f\|^{2}_{L^{1}(\mathbb{R}^{d})},
\end{split}
\end{align}
where we used the change of variable $v_*\rightarrow v'$ in the integral containing $f'$.  Of course, in this case it is harmless since $b(\cdot)$ is supported in $[0,1]$.  Furthermore, we recall that $\|b^{\varepsilon}_{2}\|_{L^{1}(\mathbb{S}^{d-1})}\sim\mathfrak{m}(b^{\varepsilon}_{2})$ can be made as small as desired.

For the rest of the cases one does not need to split the kernel in two.  It suffices to choose $x=f(v_*)\langle v_* \rangle^{s}$ when $(i,j)=(\gamma,s)$, $x=f(v_*)\langle v_* \rangle^{\gamma}$ when $(i,j)=(s,\gamma)$, and $x=f(v_*)$ when $(i,j)=(s+\gamma,0)$.  In all cases $y=|\log(f')|\text{1}_{f'>1}$.  Furthermore, the resulting lower order entropic moments can be controlled as 
\begin{equation}\label{EMe7}
\|f\log(f)\langle v \rangle^{\gamma}\|_{L^{1}(\mathbb{R}^{d})}\leq \|f\log(f)\langle v \rangle^{s+\gamma}\|^{\frac{\gamma}{s+\gamma}}_{L^{1}(\mathbb{R}^{d})}\|f\log(f)\|^{\frac{s}{s+\gamma}}_{L^{1}(\mathbb{R}^{d})}\,.
\end{equation} 
In summary, after gathering \eqref{EMe3}, \eqref{EMe4}, \eqref{EMe5}, \eqref{EMe6} and \eqref{EMe7} one concludes
\begin{align}\label{EMe8}
\begin{split}
\int_{\mathbb{R}^{d}} &Q^{+}(f,f)(v)  \langle{v}\rangle^{s}|\log(f)|\text{d}v\leq \mathfrak{m}(b^{\varepsilon}_{2})\|f\log(f)\langle v \rangle^{s+\gamma}\|_{L^{1}(\mathbb{R}^{d})}\\
&+C_1(f)\|f\log(f)\langle v \rangle^{s+\gamma}\|^{\frac{\gamma}{s+\gamma}}_{L^{1}(\mathbb{R}^{d})}+C_{2}(f)\|f\log(f)\langle v \rangle^{s+\gamma}\|^{\frac{s}{s+\gamma}}_{L^{1}(\mathbb{R}^{d})}+C^{\varepsilon}_{3}(f)\,,
\end{split}
\end{align}
where the constants $C_{i}(f)$, with $i=1,2,3$, depend on the mass, temperature, initial entropy and $\|f\langle v \rangle^{s+\gamma+2}\|_{L^{1}(\mathbb{R}^{d})}$.  Now, choose $\varepsilon$ such that $\mathfrak{m}(b^{\varepsilon}_{2})\leq \frac{c_o}{2}$ and use the estimates \eqref{EMe1}, \eqref{EMe2} and \eqref{EMe3} to get
\begin{align*}
\frac{\text{d}}{\text{d}t}&\|f\log(f)\langle v \rangle^{s}\|_{L^{1}(\mathbb{R}^{d})}\leq C_1(f)\|f\log(f)\langle v \rangle^{s+\gamma}\|^{\frac{\gamma}{s+\gamma}}_{L^{1}(\mathbb{R}^{d})}\\
&+C_{2}(f)\|f\log(f)\langle v \rangle^{s+\gamma}\|^{\frac{s}{s+\gamma}}_{L^{1}(\mathbb{R}^{d})}+C^{\varepsilon}_{3}(f) - \frac{c_o}{2}\|f\log(f)\langle v \rangle^{s+
\gamma}\|_{L^{1}(\mathbb{R}^{d})}\,.
\end{align*}
The result follows from here after invoking the instantaneous appearance of moments proven in \cite{WennbergMP}.
\end{proof}
We next prove estimates on the relative entropy, and note that the rate of the decay  depends on how rough initial data is. If initial data has finite mass, energy and entropy, the decay rate is almost logarithmic (see \eqref{ln decay}). If, in addition, initial data has a finite moment of order $p>1$, then the decay rate is algebraic (see \eqref{poly decay}).

\begin{proposition}\label{transitory-relaxation}
Assume $f$ is solution of the homogeneous Boltzmann problem for hard potentials with initial data satisfying \eqref{Initial-Data}.  Then, for any $\varepsilon>0$
\begin{equation}\label{ln decay}
\mathcal{H}(f | \mathcal{M}_{f}) \leq C_{\varepsilon}(f_0)\big(\ln(e+t)\big)^{-\frac{1}{1+\varepsilon}}\,,
\end{equation}
with constant $C_{\varepsilon}(f_0)$ depending on mass, temperature, and entropy of $f_0$.  Furthermore, if additionally $f_0\in L^{p}(\mathbb{R}^{d})$, with $p\in(1,\infty]$, then
\begin{equation}\label{poly decay}
\mathcal{H}(f | \mathcal{M}_{f}) \leq C_{\varepsilon}(f_0)(1+t)^{-\frac{1}{\frac{\gamma p'}{d}+\varepsilon\big(1+\frac{\gamma p'}{d}\big)}}\,,
\end{equation}
with constant $C_{\varepsilon}(f_0)$ depending additionally of the $L^{p}$-norm of $f_0$.
\end{proposition}

\begin{proof}
Let us prove the first statement.  Thanks to the appearance of moments \cite{WennbergMP, AGbams},  the  pointwise lower Gaussian bound \cite{PW} and the estimate on entropic moments  from Proposition \ref{entropic-moments},  we can use the entropy dissipation estimate from Theorem \ref{Dissipation-entropy} to conclude that there exist constants $\mathcal{A}_o(f_0)$ and $\mathcal{B}_o(f_0)$ only depending on initial mass, temperature, entropy, and $t_o>0$ (and $\varepsilon>0$) such that
\begin{equation*}
\frac{\text{d}}{\text{d}t}\mathcal{H}( f | \mathcal{M}_{f}) + \mathcal{A}_o(f_0)\mathcal{H}( f | \mathcal{M}_{f})^{(1+\varepsilon)(1+\frac{\gamma}{d})}e^{-\frac{\mathcal{B}_o(f_0)}{\mathcal{H}( f | \mathcal{M}_{f})^{1+\varepsilon}}}\leq 0\,,\qquad t\geq t_o\,.
\end{equation*}
It is not difficult to check that $X(t)= C\big(\ln(e+t)\big)^{-\frac{1}{1+\varepsilon}}$ satisfies
\begin{equation*}
\frac{\text{d}}{\text{d}t}X + \mathcal{A}_o(f_0)X^{(1+\varepsilon)(1+\frac{\gamma}{d})}e^{-\frac{\mathcal{B}_o(f_0)}{X^{1+\varepsilon}}}\geq 0\,,
\end{equation*}
provided $C>0$ is taken large enough depending on $\mathcal{A}_o(f_0)$ and $\mathcal{B}_o(f_0)$.  Denote any such value as $C_*$ and choose $C=C(f_0):=\max\{\mathcal{H}(f_0 | \mathcal{M}_{f_0} ), C_*\}$.  Since $\mathcal{H}(f(t_o) | \mathcal{M}_{f(t_o)})\leq \mathcal{H}(f_0 | \mathcal{M}_{f_0} )$ the result follows by a comparison principle.  The second statement follows similar argument and it is left to the reader.
\end{proof}

\subsection{Exponential convergence}
After fixing the initial mass, momentum, and temperature, one can rewrite the Boltzmann equation \eqref{HBE} by taking $f=\mathcal{M}_{f_0} + h$, where $h$ is understood as a perturbation of the thermodynamical equilibrium.  The equation for $h$ reads
\begin{equation}\label{HBE-pf}
\partial_{t}h(v) = \mathcal{L}(h)(v) + Q(h,h)(v)\,,\qquad (t,v)\in\mathbb{R}^{+}\times\mathbb{R}^{d}\,.
\end{equation}
Here, the linear component of the dynamics is generated by the operator
\begin{equation}\label{linear-operator}
\mathcal{L}(h)(v) = Q(\mathcal{M}_{f_0},h) + Q(h,\mathcal{M}_{f_0})\,.
\end{equation}
This operator was shown to be self-adjoint non-positive with a spectral gap in $L^{2}(\mathcal{M}^{-1/2}_{f_0})$ in the references \cite{C}, \cite{G1} and \cite{G2} in the cut-off case.  Later, allowing more general kernels, an explicit estimate of the spectral gap in this same space was made in \cite{Mo1}.  This is the starting point at which the spectral enlargement technique works \cite{Mo,GMM} to obtain spectral gap in more general spaces.

Based on \cite[Theorem 2.3]{GMM} or \cite[Theorem 3.1]{CL}, if a spectral enlargement from $L^{2}(\mathcal{M}^{-1/2}_{f_0})$ to $L^{1}_{k}$ is desired, we need to decompose the linearized operator $\mathcal{L}$ into two operators with suitable properties.  More precisely, $\mathcal{L} = \mathcal{A} + \mathcal{B}$ where $\mathcal{B}:L^{1}_{k}(\mathbb{R}^{d})\rightarrow L^{1}_{k}(\mathbb{R}^{d})$ is dissipative \footnote{That is, the operator $\mathcal{B}$ is closed with domain $L^{1}_{k+\gamma}(\mathbb{R}^{d})$ and satisfying $\langle \mathcal{B}f, f \rangle\leq0$.} and $\mathcal{A}:L^{1}_{k}(\mathbb{R}^{d})\rightarrow L^{2}(\mathcal{M}^{-1/2}_{f_0};\mathbb{R}^{d})$ is bounded.  Here we stress that the space $L^{2}(\mathcal{M}^{-1/2}_{f_0};\mathbb{R}^{d})$ is taken as baseline space since a detailed quantification of the spectrum is available in the aforementioned references.

The decomposition is based on truncation of small and large velocities, and glancing angles (similar but simpler to the decomposition given in \cite[subsection 4.3.3]{GMM} or \cite{Mo}).  For the scattering kernel, recall the decomposition \eqref{DAK} we have used in this paper
\begin{equation}\label{DAK1}
b\big(\cos(\theta)\big) = b\big(\cos(\theta)\big) \big(\text{1}_{|\sin(\theta)| \geq \varepsilon} + \text{1}_{|\sin(\theta)| < \varepsilon}\big)=:b^{\varepsilon}_{1}\big(\cos(\theta)\big) + b^{\varepsilon}_{2}\big(\cos(\theta)\big)\,.
\end{equation}
For the kinetic potential write $|\cdot|^{\gamma}=:\Phi_{1}+\Phi_{2}$, where
\begin{equation}\label{DAK2}
\Phi_{1}(|u|):= |u|^{\gamma}\text{1}_{\delta\leq |u|\leq \delta^{-1}}\,,\qquad \Phi_{2}(|u|):= |u|^{\gamma}\big(1-\text{1}_{\delta\leq |u|\leq \delta^{-1}}\big)\,.
\end{equation}
With the notation $\mathcal{L}_{\Phi,b}$ to express the dependence of the collision kernel, one can write
\begin{align}\label{Decomposition}
\begin{split}
\mathcal{L}_{x^{\gamma},b}& h = \mathcal{L}_{x^{\gamma},b^{\varepsilon}_{1}}h + \mathcal{L}_{x^{\gamma},b^{\varepsilon}_2}h = \mathcal{L}^{o}_{x^{\gamma},b^{\varepsilon}_{1}}h + \mathcal{L}_{x^{\gamma},b^{\varepsilon}_2}h - h\int_{\mathbb{R}^{d}}\mathcal{M}_{f_0}(v_*)|u|^{\gamma}\text{d}v_{*}\\
& =  \mathcal{L}^{o}_{\Phi_{1},b^{\varepsilon}_{1}}h + \Big(\mathcal{L}^{o}_{\Phi_{2},b^{\varepsilon}_{1}}h + \mathcal{L}_{x^{\gamma},b^{\varepsilon}_2}h - h\int_{\mathbb{R}^{d}}\mathcal{M}_{f_0}(v_*)|u|^{\gamma}\text{d}v_{*}\Big)=: \mathcal{A}_{\delta,\varepsilon}h + \mathcal{B}_{\delta,\varepsilon}h\,.
\end{split}
\end{align}
Of course,
\begin{equation}\label{Lo}
\mathcal{L}^{o}_{\Phi_{1},b^{\varepsilon}_{1}}h := Q^{+}_{\Phi_{1},b^{\varepsilon}_{1}}(\mathcal{M}_{f_0},h) + Q^{+}_{\Phi_{1},b^{\varepsilon}_{1}}(h,\mathcal{M}_{f_0}) - Q^{-}_{\Phi_{1},b^{\varepsilon}_{1}}(\mathcal{M}_{f_0},h)\,.
\end{equation}
Let us prove that $\mathcal{B}_{\delta,\varepsilon}$ is dissipative for sufficiently small parameters $\delta>0,\, \varepsilon>0$ and that $\mathcal{A}_{\delta,\varepsilon}$ has the stated ``regularizing'' property for \text{any} $\delta>0,\, \varepsilon>0$.  For the first statement, one essentially needs the following lemma controlling moments of the linearized operator.
\begin{lemma}\label{control-linear-moments}
Consider angular kernel $b\in L^{1}(\mathbb{S}^{d})$ and potential $0\leq \Phi(|u|) \leq |u|^{\gamma}$.  Then for any $h\in L^{1}_{k+\gamma}(\mathbb{R}^{d})$ and any $k\geq2$
\begin{align*}
\int_{\mathbb{R}^{d}}sgn(h)\big(1+ |v|^{k}\big)\mathcal{L}_{\Phi,b}(h)(v)&dv \leq C_{k}\big\|b\|_{L^{1}(\mathbb{S}^{d-1})}\|\langle v \rangle^{k-1}\beta_{\Phi}(v)h\|_{L^{1}(\mathbb{R}^{d})}\\
&-(1-\gamma_{k})\int_{\mathbb{R}^{2d}}|h(v)|\mathcal{M}_{f_0}(v_*)\Phi(|u|)\big(|v|^{k}+|v_*|^{k}\big)dv_{*}dv\,.
\end{align*}
Here above, $0<\gamma_{k}<\big\|b\|_{L^{1}(\mathbb{S}^{d-1})}=1$ is such that $\gamma_{k}\searrow0$ as $k\rightarrow\infty$.  And,
\begin{equation*}
\beta_{\phi}(v):=\int_{\mathbb{R}^{d}}\mathcal{M}_{f_0}(v_*)\langle v_* \rangle^{k}\Phi(|u|)dv_{*}\,.
\end{equation*}
\end{lemma}
\begin{proof}
Using the weak representation
\begin{equation*}
\int_{\mathbb{R}^{d}}\varphi\,\mathcal{L}_{\Phi,b}h(v)\text{d}v = \int_{\mathbb{R}^{2d}}\int_{\mathbb{S}^{d-1}}\mathcal{M}_{f_0}(v_*)h\big(\varphi'+\varphi'_{*}-\varphi_{*}-\varphi\big)\Phi(|u|)b(\widehat{u}\cdot\sigma)\text{d}\sigma\text{d}v_*\text{d}v\,.
\end{equation*}
Thus, for $\varphi(v)=\text{sgn}(h)\big(1 + |v|^{k}\big)$ one readily checks that
\begin{align*}
\int_{\mathbb{R}^{d}}\text{sgn}(h)\big(1 + |v|^{k}\big)\mathcal{L}_{\Phi,b}h(v)&\text{d}v\leq \int_{\mathbb{R}^{d}} |v|^{k} \mathcal{L}_{\Phi,b}|h|(v)\text{d}v + 2\big\|b\|_{L^{1}(\mathbb{S}^{d-1})}\|\beta_{\Phi}(v)h\|_{L^{1}}\,.
\end{align*}
The result follows using a Povzner angular averaging lemma in the first term of the right side, see for example \cite[Lemma 2.6]{AL}.
\end{proof}
Thanks to  the above lemma, we now prove an estimate on moments of the  operator $\mathcal{B_{\delta,\varepsilon}}$.

\begin{lemma}\label{dissipativeop}
Consider angular kernel $b\in L^{1}(\mathbb{S}^{d})$ and potentials with $\gamma\in(0,1]$.  For $\varepsilon\in(0,1)$, $\delta\in(0,1)$ and $k\geq2$, it follows that 
\begin{align*}
\int_{\mathbb{R}^{d}}sgn(h)\big(1+|v|^{k}\big)\mathcal{B_{\delta,\varepsilon}}(h)\text{d}v &\leq C_{k}\mathfrak{m}(b^{\varepsilon}_{2})\,\|h\langle v \rangle ^{k+\gamma}\|_{L^{1}(\mathbb{R}^{d})}\\
&+C_{k}\,\delta^{\gamma}\,\|h\langle v \rangle ^{k+\gamma}\|_{L^{1}(\mathbb{R}^{d})} - c_o\|(1+|v|^{k})\langle v \rangle ^{\gamma}h\|_{L^{1}(\mathbb{R}^{d})}\,.
\end{align*}
The constant $C_{k}>0$ in addition to $k$, depends on mass and temperature.  The constant $c_o$ depends only on mass and temperature.
\end{lemma}
\begin{proof}
First, let us estimate $\mathcal{L}_{x^{\gamma},b^{\varepsilon}_{2}}$ and $\mathcal{L}^{o}_{\Phi_{2},b^{\varepsilon}_{1}}$ separately.  For the former, using Lemma \ref{control-linear-moments} one readily concludes
\begin{align}\label{disspativeop-e1}
\begin{split}
\int_{\mathbb{R}^{d}}sgn(h)\big(1+ |v|^{k}\big)\mathcal{L}_{x^{\gamma},b^{\varepsilon}_{2}}(h)(v)dv &\leq C_{k}\big\|b^{\varepsilon}_{2}\|_{L^{1}(\mathbb{S}^{d-1})}\|\langle v \rangle^{k-1}\beta_{x^{\gamma}}(v)h\|_{L^{1}(\mathbb{R}^{d})}\\
&\leq C_{k}\mathfrak{m}(b^{\varepsilon}_{2})\|\langle v \rangle^{k+\gamma-1}h\|_{L^{1}(\mathbb{R}^{d})}\,.
\end{split}
\end{align}
We used, for the second inequality, that $\beta_{x^{\gamma}}(v)\leq \widetilde{C}_{k}\langle v \rangle^{\gamma}$ and $\| b^{\varepsilon}_{2}\|_{L^{1}(\mathbb{S}^{d-1})} \sim \mathfrak{m}(b^{\varepsilon}_{2})$\,.  Similarly for the latter, note that $\mathcal{L}^{o}_{\Phi_{2},b^{\varepsilon}_{1}} = \mathcal{L}_{\Phi_{2},b^{\varepsilon}_{1}} + Q^{-}_{ \Phi_2 , b^{\varepsilon}_{1}} (\mathcal{M}_{f_0},h)$, thus, using Lemma \ref{control-linear-moments} it follows that
\begin{align}\label{dissipativeop-e2}
\begin{split}
\int_{\mathbb{R}^{d}}sgn(h)\big(1+ |v|^{k}\big)\mathcal{L}_{\Phi_{2},b^{\varepsilon}_{1}}(h)(v)dv &\leq C_{k}\big\|b^{\varepsilon}_{1}\|_{L^{1}(\mathbb{S}^{d-1})}\|\langle v \rangle^{k-1}\beta_{\Phi_{2}}(v)h\|_{L^{1}(\mathbb{R}^{d})}\\
&\leq C_{k}\,\delta^{\gamma}\,\|b\|_{L^{1}(\mathbb{S}^{d-1})}\|\langle v \rangle^{k+\gamma}h\|_{L^{1}(\mathbb{R}^{d})}\,.
\end{split}
\end{align}
We used, in the second inequality, the fact that $b^{\varepsilon}_{2}\leq b$ and that $\beta_{\Phi_{2}}(v)\leq \widetilde{C}_{k}\delta^{\gamma}\langle v \rangle^{1+\gamma}$.  Second, for the dissipation term, one uses that
\begin{equation}\label{spectral-gap}
\int_{\mathbb{R}^{d}}\mathcal{M}_{f_0}(v_*)|u|^{\gamma}\text{d}v_{*}\geq c_o\langle v \rangle^{\gamma}\,,
\end{equation} 
for constant $c_o$ depending only on mass and temperature.  Then,
\begin{equation}\label{disspativeop-e3}
\int_{\mathbb{R}^{d}}\big(1+|v|^{k}\big)|h|\int_{\mathbb{R}^{d}}\mathcal{M}_{f_0}(v_*)|u|^{\gamma}\text{d}v_{*}\geq c_o\|(1+|v|^{k})\langle v \rangle^{\gamma}h\|_{L^{1}(\mathbb{R}^{d})}\,.
\end{equation}
The result follows after gathering \eqref{disspativeop-e1}, \eqref{dissipativeop-e2}, \eqref{disspativeop-e3}.
\end{proof}

As a corollary of the previous lemma, we prove that $\mathcal{B}_{\delta,\varepsilon}$ is dissipative for sufficiently small parameters $\delta>0,\, \varepsilon>0$.

\begin{corollary}\label{dissipativeop1}
There exists positive $(\delta_{k},\varepsilon_k)$ depending on mass and temperature such that $\mathcal{B}_{\delta,\varepsilon}$ is dissipative in $L^{1}_{k}(\mathbb{R}^{d})$ for any $\delta\in(0,\delta_{k})$ and $\varepsilon\in(0,\varepsilon_k)$.  Furthermore, for any $c^{-}_o<c_o$, given in \eqref{spectral-gap}, it is possible to find $(\delta,\varepsilon)$ in such ranges and such that
\begin{equation}\label{dissipativeopnorm}
\|e^{t \mathcal{B}_{\delta,\varepsilon}} \|_{L^{1}_{k}(\mathbb{R}^{d})}\leq 2^{\frac{k}{2}-1} e^{- c^{-}_o\,t}.
\end{equation}
\end{corollary}
\begin{proof}
Consider the problem $h' = \mathcal{B}_{\delta,\epsilon}(h)$ in $L^{1}_{k}(\mathbb{R}^{d})$.  Using Lemma \ref{dissipativeop} and the elementary inequality $\langle v \rangle^{k}\leq 2^{\frac{k}{2}-1}(1+|v|^{k})$, it follows that
\begin{equation*}
\frac{\text{d}}{\text{d}t}\|(1+|v|^{k})h\|_{L^{1}(\mathbb{R}^{d})} + c_o(1-\widetilde{\delta}-\widetilde{\varepsilon})\|(1+|v|^{k})h\|_{L^{1}(\mathbb{R}^{d})}\leq0\,,\qquad 0<\widetilde{\delta}+\widetilde{\varepsilon}<1\,,
\end{equation*}
provided we choose $\delta>0$ and $\varepsilon>0$ such that $C_{k}2^{\frac{k}{2}-1}\delta^{\gamma}\leq c_o\,\widetilde{\delta}$ and $C_{k}2^{\frac{k}{2}-1}\mathfrak{m}(b^{\varepsilon}_{2})\leq c_o\,\widetilde{\varepsilon}$.  Here $\widetilde{\delta}+\widetilde{\varepsilon}$ is allowed to be as small as desired.  As a consequence, choosing $(\widetilde{\delta},\widetilde{\varepsilon})$ such that $c^{-}_{o}= c_o(1 - \widetilde{\delta} - \widetilde{\varepsilon})$ it follows that
\begin{align*}
2^{1-\frac{k}{2}}\|&\langle v \rangle^{k} h\|_{L^{1}(\mathbb{R}^{d})}\leq\|(1+|v|^{k})h\|_{L^{1}(\mathbb{R}^{d})}\\
&\leq \|(1+|v|^{k})h_0\|_{L^{1}(\mathbb{R}^{d})}e^{-c^{-}_o\,t}\leq\|\langle v \rangle^{k} h_0\|_{L^{1}(\mathbb{R}^{d})}e^{-c^{-}_o\,t}\,,\qquad t>0\,.
\end{align*}
This is exactly \eqref{dissipativeopnorm}.
\end{proof}

The following lemma shows that the operator $\mathcal{A}_{\delta,\varepsilon}:L^{1}_{k}(\mathbb{R}^{d})\rightarrow L^{2}(\mathcal{M}^{-1/2}_{f_0})$  is bounded.

\begin{lemma}\label{smoothop}
For any $\delta\in(0,1)$ and $\varepsilon\in(0,1)$ and $k\geq2$, the operator $\mathcal{A}_{\delta,\varepsilon}:L^{1}_{k}(\mathbb{R}^{d})\rightarrow L^{2}(\mathcal{M}^{-1/2}_{f_0})$ is bounded. 
\end{lemma}
\begin{proof}
Let us estimate each term separately in $\mathcal{A}_{\delta,\varepsilon} = \mathcal{L}^{o}_{\Phi_1,b^{\varepsilon}_{1}}$, recall \eqref{Lo}.  It is direct that the last term is controlled by
\begin{equation}\label{smoothope1}
\|\mathcal{M}^{-1/2}_{f_0}Q^{-}_{\Phi_{1},b^{\varepsilon}_{1}}(\mathcal{M}_{f_0}, h)\|_{L^{2}(\mathbb{R}^{d})}\leq \|b\|_{L^{1}(\mathbb{S}^{d-1})}\|\langle v \rangle^{\gamma}\mathcal{M}^{1/2}_{f_0}\|_{L^{2}(\mathbb{R}^{d})}\|\langle v \rangle^{\gamma}h\|_{L^{1}(\mathbb{R}^{d})}\,.
\end{equation}
For the first term we use the elementary inequality $|x+y|^{2}\leq\tfrac{3}{2}|x|^{2}+3|y|^{2}$ valid for any $x,\,y\in\mathbb{R}^{d}$, and set $x=v'=v-u^{-}$ and $y=u^{-}$ to show that
\begin{equation*}
\mathcal{M}^{-1/2}_{f_0}(v)\leq \mathcal{M}^{-3/4}_{f_0}(v')\mathcal{M}^{-3/2}_{f_0}(u)\,.
\end{equation*}
Thus, we discover that
\begin{equation*}
\mathcal{M}^{-1/2}_{f_0}(v)Q^{+}_{\Phi_1,b^{\varepsilon}_{1}}\big( \mathcal{M}_{f_0}, |h|\big)(v) \leq \|\mathcal{M}^{-3/2}_{f_0}(u)\Phi_{1}\|_{L^{\infty}(\mathbb{R}^{d})}Q^{+}_{\Phi_1,b^{\varepsilon}_{1}}\big(\mathcal{M}^{1/4}_{f_0}, |h|\big)(v)\,.
\end{equation*}
As a conclusion, using Young's inequality for the $Q^{+}$,
\begin{align}\label{smoothope2}
\begin{split}
\|\mathcal{M}^{-1/2}_{f_0}&Q^{+}_{\Phi_1,b^{\varepsilon}_1}\big(\mathcal{M}_{f_0}, |h| \big)\|_{L^{2}(\mathbb{R}^{d})}\\
&\leq C(b^{\varepsilon}_{1})\|\mathcal{M}^{-3/2}_{f_0}(u)\Phi_{1}\|_{L^{\infty}(\mathbb{R}^{d})}\|\mathcal{M}^{1/4}_{f_0}\|_{L^{2}(\mathbb{R}^{d})}\|h\|_{L^{1}(\mathbb{R}^{d})}\,.
\end{split}
\end{align}
In addition, since $b^{\varepsilon}_{1}$ is cut-off near zero angle and (setting $x=v'_{*}=v-u^{+}$ and $y=u^{+}$ above)
\begin{equation*}
\mathcal{M}^{-1/2}_{f_0}(v)\leq \mathcal{M}^{-3/4}_{f_0}(v'_{*})\mathcal{M}^{-3/2}_{f_0}(u)\,,
\end{equation*}
we also have
\begin{align}\label{smoothope3}
\begin{split}
\|\mathcal{M}^{-1/2}_{f_0}&Q^{+}_{\Phi_1,b^{\varepsilon}_1}\big(|h|,\mathcal{M}_{f_0} \big)\|_{L^{2}(\mathbb{R}^{d})}\\
&\leq \widetilde{C}(b^{\varepsilon}_{1})\|\mathcal{M}^{-3/2}_{f_0}(u)\Phi_{1}\|_{L^{\infty}(\mathbb{R}^{d})}\|\mathcal{M}^{1/4}_{f_0}\|_{L^{2}(\mathbb{R}^{d})}\|h\|_{L^{1}(\mathbb{R}^{d})}\,.
\end{split}
\end{align}
The result follows, after gathering \eqref{smoothope1}, \eqref{smoothope2} and \eqref{smoothope3}.
\end{proof}

Next result is on the spectral gap of the linearized operator $\mathcal{L}$ on the enlarged space $L^1_k$, which is a direct consequence of  \cite[Theorem 2.1]{GMM} thanks Corollary \ref{dissipativeop1} and Lemma \ref{smoothop}.

\begin{proposition}\label{linearized-spectrum}
Let $b\in L^{1}(\mathbb{S}^{d-1})$ and $\gamma\in(0,1]$ and $k\geq2$.  The linear operator $\mathcal{L}$ defined in $L^{1}_{k}(\mathbb{R}^{d})$ satisfies
\begin{align*}
\text{Spec}(\mathcal{L})&\subset\{z\in\mathbb{C}\big|\mathfrak{R}e(z)\leq -\lambda\}\cup\{0\}\,,\\
\text{Ker}(\mathcal{L})=\text{Span}&\{\mathcal{M}_{f_0},v_{1}\mathcal{M}_{f_0},\cdots,v_{d}\mathcal{M}_{f_0},|v^{2}|\mathcal{M}_{f_0}\}\,,
\end{align*}
for any $\lambda<\lambda_{o}$ where $\lambda_o$ is the spectral gap of the restriction of $\mathcal{L}$ to $L^{2}(\mathcal{M}^{-1/2}_{f_0};\mathbb{R}^{d})$.  Furthermore, it generates a strongly continuous semigroup $e^{t\,\mathcal{L}}$ which satisfies 
\begin{equation*}
\|e^{t\mathcal{L}}h_0 - \pi h_0\|_{L^{1}_{k}(\mathcal{R}^{d})} \leq C_{k}e^{-\lambda\,t}\|h_0 - \pi h_0 \|_{L^{1}_{k}(\mathcal{R}^{d})}\,.
\end{equation*}
Here $\pi$ stands for the projection onto $\text{Ker}(\mathcal{L})$ defined as
\begin{equation*}
\pi g := \sum_{\varphi\in\{1,v_{1},\cdots,v_{d},|v|^{2}\}}\left(\int_{\mathbb{R}^{d}}\,g\,\varphi\,dv\right)\varphi\,\mathcal{M}_{f_0}\,.
\end{equation*}
\end{proposition}
\begin{proof}
This is a direct consequence of \cite[Theorem 2.1]{GMM} after taking the spaces
\begin{equation*}
E:=L^{2}(\mathcal{M}^{-1/2}_{f_0};\mathbb{R}^{d})\subset L^{1}_{k}(\mathbb{R}^{d})=:\mathcal{E}\,.
\end{equation*}
\end{proof}

Now we have all the ingredients for proving  the exponential convergence of the full homogeneous Boltzmann equation \eqref{HBE-pf}. 
\begin{theorem}[Exponential relaxation]\label{full-exponential}
Let the angular kernel satisfy \eqref{Grad-cutoff-lower} and potential $\gamma\in(0,1]$.  Assume the initial datum $f_0$ satisfies \eqref{Initial-Data}.  Then, for every $\lambda<\lambda_o$, there exists time $t_{k}(f_0)$ and $C_{k}(f_0)$ depending on the initial datum through its mass, temperature and entropy, such that
\begin{equation*}
\|f-\mathcal{M}_{f_0}\|_{L^{1}_{k}(\mathbb{R}^{d})} \leq C_{k}(f_0)e^{-\lambda t}\,,\qquad t\geq t_{k}(f_0)\,. 
\end{equation*}
Here $\lambda_o>0$ is the spectral gap of $\mathcal{L}$ in $L^{2}(\mathcal{M}^{-1/2}_{f_0};\mathbb{R}^{d})$.
\end{theorem}
\begin{proof}
 Let $f:=f(t,v)$ be solution of the Boltzmann equation, $\mathcal{M}_{f_0}$ its thermodynamical equilibrium, and set the perturbation $f:=\mathcal{M}_{f_0}+ h$.  The first step is to invoke Csisz\'ar-Kullback-Pinsker inequality to deduce that
\begin{equation*}
\|f-\mathcal{M}_{f_0}\|_{L^{1}(\mathbb{R}^{d})}\leq \sqrt{2\,\mathcal{H}(f | \mathcal{M}_{f_0})}\,.
\end{equation*}
As a consequence, using the estimate on the relative entropy from Proposition \ref{transitory-relaxation} and the property of generation of moments, it follows that for any $k\geq0$
\begin{align*}
\|f-\mathcal{M}_{f_0}\|_{L^{1}_{k}(\mathbb{R}^{d})} 
&\leq \|f-\mathcal{M}_{f_0}\|^{\frac{1}{4}}_{L^{1}_{4k}(\mathbb{R}^{d})}\|f - \mathcal{M}_{f_0}\|^{\frac{3}{4}}_{L^{1}(\mathbb{R}^{d})}\\
&\leq C_{k}(f_0)\|f - \mathcal{M}_{f_0}\|^{\frac{3}{4}}_{L^{1}(\mathbb{R}^{d})}\leq C_{k}(f_0)\big(\ln(e+t)\big)^{-\frac{1}{3}}\,, \qquad t \geq 1\,.
\end{align*}
Here, the constant $C_{k}(f_0)>0$ only depends on the initial datum through its mass, temperature and entropy.  In this way,  for every $\tilde{\varepsilon}\in(0,1)$ there exists time $t_{k}(f_0)$, for example  $t_k(f_0) = e^{\left( \tilde{\varepsilon}/C_k(f_0)\right)^3}$, such that
\begin{equation}\label{expcome1}
\|f-\mathcal{M}_{f_0}\|_{L^{1}_{k}(\mathbb{R}^{d})} \leq \tilde{\varepsilon}\,,\qquad t\geq t_{k}(f_0)\,.
\end{equation} 

Define $\tilde{h}(t):=h(t+t_{k+2\gamma}(f_0))$.  Then, from equation \eqref{HBE-pf}, it follows that such perturbation satisfies  
\begin{equation*}
\tilde{h}(t) = e^{t\mathcal{L}}\tilde{h}_0 + \int^{t}_0 e^{(t-s)\mathcal{L}}Q(\tilde{h},\tilde{h})(s)\text{d}s\,,\qquad t\geq0\,.
\end{equation*}
Since $\pi \tilde{h}_0 = 0$ and $\pi Q(\tilde{h},\tilde{h})(s) = 0$ for every $s\geq0$, it follows from Proposition \ref{linearized-spectrum} that
\begin{align}\label{expcome2}
\begin{split}
\|\tilde{h}(t)\|_{L^{1}_{k}(\mathbb{R}^{d})} &\leq \|e^{t\mathcal{L}}\tilde{h}_0\|_{L^{1}_{k}(\mathbb{R}^{d})} + \int^{t}_0 \|e^{(t-s)\mathcal{L}}Q(\tilde{h},\tilde{h})(s)\|_{L^{1}_{k}(\mathbb{R}^{d})}\text{d}s\\
&\leq e^{-\lambda t}\|\tilde{h}_0\|_{L^{1}_{k}(\mathbb{R}^{d})} +\int^{t}_0 e^{-\lambda(t-s)}\|Q(\tilde{h},\tilde{h})(s)\|_{L^{1}_{k}(\mathbb{R}^{d})}\text{d}s\,.
\end{split}
\end{align}
Furthermore, using \eqref{expcome1}  we have
\begin{equation*}
\|Q(\tilde{h},\tilde{h})(s)\|_{L^{1}_{k}(\mathbb{R}^{d})}\leq \|\tilde{h}(s)\|^{2}_{L^{1}_{k+\gamma}(\mathbb{R}^{d})}\leq \|\tilde{h}(s)\|_{L^{1}_{k}(\mathbb{R}^{d})}\|\tilde{h}(s)\|_{L^{1}_{k+2\gamma}(\mathbb{R}^{d})}\leq \tilde{\varepsilon} \|\tilde{h}(s)\|_{L^{1}_{k}(\mathbb{R}^{d})}\,.
\end{equation*}
As a consequence, setting $X(t):=e^{\lambda t}\|\tilde{h}(t)\|_{L^{1}_{k}(\mathbb{R}^{d})}$, we obtain from \eqref{expcome2}
\begin{equation*}
X(t)\leq X_0 + \tilde{\varepsilon}\int^{t}_0X(s)\text{d}s\,,\qquad t\geq0\,.
\end{equation*}
Using Gronwall's lemma one concludes that $X(t)\leq X_{0}\,e^{\tilde{\varepsilon} t}$, or equivalently,
\begin{equation*}
\|\tilde{h}(t)\|_{L^{1}_{k}(\mathbb{R}^{d})}\leq \|\tilde{h}_0\|_{L^{1}_{k}(\mathbb{R}^{d})}e^{-(\lambda-\tilde{\varepsilon})t}\,,\qquad t\geq0\,. 
\end{equation*}
\end{proof}

Theorem \ref{full-exponential} gives, straightforwardly, the decomposition theorem \cite[Theorem 5.5]{MV} related to the propagation of roughness for the Homogeneous Boltzmann equation.
\begin{corollary}[Decomposition theorem]\label{Decomposition}
Let the angular kernel satisfy \eqref{Grad-cutoff-lower} and potential $\gamma\in(0,1]$.  Assume the initial datum $f_0$ satisfies \eqref{Initial-Data}.  For any $s,k\geq0$ and $t_o>0$ there exist functions $f^{S}\geq0$ and $f^{R}$ such that 
\begin{equation*}
f = f^{S} + f^{R}\,,\qquad t\geq t_o>0\,,
\end{equation*}
satisfying the estimates
\begin{equation*}
\sup_{t\geq t_o}\|f^{S}(t)\|_{H^{s}_{k}(\mathbb{R}^{d})}<C_{s,k}(t_o)\,,\qquad \|f^{R}(t)\|_{L^{1}_{k}(\mathbb{R}^{d})}\leq C_{k}(t_o)e^{-\lambda t}\,,\quad \lambda<\lambda_o\,.
\end{equation*}
The constants depend on mass, energy, and entropy of $f_0$.
\end{corollary}
\begin{proof}
Choose $f^{S}:=\mathcal{M}_{f_0}$ and $f^{R}:= f - \mathcal{M}_{f_0}$, and use the exponential relaxation given in Theorem \ref{full-exponential} for $f^{R}$.
\end{proof}
\begin{remark}
In the context of the homogenous Boltzmann problem, the decomposition theorem loses its value since it is an obvious consequence of Theorem \ref{full-exponential}.  However, the techniques to prove it given in \cite[Theorem 5.5]{MV} continue being useful in other contexts.
\end{remark}

\vspace{.5cm}
\textbf{Acknowledgements.} 
R. Alonso thanks the Oden Institute  for Computational Engineering and Sciences at the University of Texas Austin for its hospitality.  R. Alonso acknowledges the support from Bolsa de Produtividade em Pesquisa CNPq 303325/2019-4 and ONR grant N000140910290.  I. M. Gamba acknowledges support from NSF grant   DMS-1715515 and DMS-2009736.  M.Taskovi\'c gratefully acknowledges support from the NSF grant DMS-2206187.  The three authors were partially funded by NSF-RNMS 1107465 as well as thank and gratefully acknowledged the hospitality and support from the Oden Institute  for Computational Engineering and Sciences and the University of Texas Austin.

\section{Appendix}
\subsection{Young's inequality for the gain collision operator}  The proof of the following theorem can be found in  \cite[Theorem 1]{ACG}.
\begin{theorem}\label{ApT-1}
Let $1\leq p, q, r\leq\infty$ with $1/p+1/q=1+1/r$.  Assume that
\begin{equation*}
B(x,y)=x^{\gamma}\,b(y)\,,\qquad \gamma\geq0, \;\; x\ge 0, \;\; y\in[0,1].
\end{equation*}
Then, for any $k\geq0$
\begin{equation*}
\|Q^{+}(f,g)\|_{L^{r}_{k}}\leq C\|f\|_{L^{p}_{k+\gamma}}\|g\|_{L^{q}_{k+\gamma}}\,,
\end{equation*}
where, whenever finite, the constant $C:=C(b)$ can be taken as
\begin{equation}\label{C-1}
C=K\left(\int^{1}_{0}\Big(\frac{1-s}{2}\Big)^{-\frac{d}{2r'}}\big(1-s^{2}\big)^{\frac{d-3}{2}}b(s)\text{d}s\right)^{\frac{r'}{q'}}\left(\int^{1}_{0}\Big(\frac{1+s}{2}\Big)^{-\frac{d}{2r'}}\big(1-s^{2}\big)^{\frac{d-3}{2}}b(s)\text{d}s\right)^{\frac{r'}{p'}}\,,
\end{equation}
with $K:=2^{k+\gamma+3}\big|\mathbb{S}^{d-2}\big|$.  In the particular cases when $p=1$ or $q=1$, the constant is interpreted as
\begin{align}
\begin{split}
\label{C-2}
C&=K\int^{1}_{0}\Big(\frac{1-s}{2}\Big)^{-\frac{d}{2q'}}\big(1-s^{2}\big)^{\frac{d-3}{2}}b(s)\text{d}s\,,\quad \text{when }\,p=1\\
C&=K\int^{1}_{0}\Big(\frac{1+s}{2}\Big)^{-\frac{d}{2p'}}\big(1-s^{2}\big)^{\frac{d-3}{2}}b(s)\text{d}s\,,\quad \text{when }\,q=1\,.
\end{split}
\end{align}
\end{theorem}

\begin{remark}\label{remark on C}
We emphasize that this theorem is applied for those kernels $b$ for which the constant $C$ is finite. In particular, a special care is needed around singular points $s=1$ and $s=-1$. Since in this paper we work with kernels $b$ that are supported on $[0,1]$, the only singular point  is $s=1$, which corresponds to grazing collisions $\theta =0$ in the $\theta$-parametrization \eqref{DAK}. This is the main reason why the kernel $b$ is split into two parts in \eqref{DAK} - close to this singularity point ($b^{\varepsilon}_{2}$)  and away from it ($b^{\varepsilon}_{1}$). The singular part $b^{\varepsilon}_{2}$ will require estimates with $q=1$ in Theorem \ref{ApT-1}, see Lemma \ref{APq+}. 
\end{remark}

\subsection{Fractional differentiation lemmas}
In this appendix some identities and estimates are presented when operating with fractional differentiation of products of functions.  They will be handy when proving commutator formulas for the collision operator and dealing with exponential or polynomial weights in the Sobolev estimates.      
\begin{lemma}\label{ApL-1}
For any $\varepsilon\in(0,1]$ and $a\geq0$ it holds that
\begin{equation*}
\bigg|\frac{x}{|x|^{1+a}} - \frac{y}{|y|^{1+a}}\bigg| \leq 2(1+a)|x-y|^{\varepsilon}\bigg(\frac{1}{|x|^{a+\varepsilon}} + \frac{1}{|y|^{a+\varepsilon}}\bigg)\,,\qquad x,\, y \in \mathbb{R}^{d}\,.
\end{equation*}
\begin{proof}
Since the role of the vectors $x$ and $y$ is interchangeable, we assume $|y|\geq|x|$ without loss of generality.  Thus, estimating directly the difference yields
\begin{align}\label{ApL-1e1}
\begin{split}
\bigg|\frac{x}{|x|^{1+a}} - \frac{y}{|y|^{1+a}}\bigg| &= \bigg|\frac{(|y|^{1+a} - |x|^{1+a})x}{|x|^{1+a}|y|^{1+a}} + \frac{|x|^{1+a}(x-y)}{|x|^{1+a}|y|^{1+a}}\bigg| \\
&\leq \frac{\big| |y|^{1+a} - |x|^{1+a}\big|}{|x|^{a}|y|^{1+a}} + \frac{|x-y|}{|y|^{1+a}}\,.
\end{split}
\end{align}
Note that
\begin{equation*}
|x-y|=|x-y|^{\varepsilon}|x-y|^{1-\varepsilon}\leq |x-y|^{\varepsilon}\big(|x|^{1-\varepsilon}+|y|^{1-\varepsilon}\big),\qquad \varepsilon\in(0,1]\,.
\end{equation*}
As a consequence, the second term on the right side in \eqref{ApL-1e1} can be readily estimated as 
\begin{equation}\label{ApL-1e2}
\frac{|x-y|}{|y|^{1+a}} \leq |x-y|^{\varepsilon}\frac{|x|^{1-\varepsilon}+|y|^{1-\varepsilon}}{|y|^{1+a}}\leq 2\frac{|x-y|^{\varepsilon}}{|y|^{a+\varepsilon}}\,.
\end{equation}
For the first term on the right side in \eqref{ApL-1e1}, note that
\begin{equation*}
\big| |y|^{1+a} - |x|^{1+a}\big| \leq (1+a)\max\{|x|^{a},|y|^{a}\}|x-y|\leq (1+a)|y|^{a}|x-y|\,.
\end{equation*}
Thus, bearing in mind that $|y|\geq|x|$, we conclude that
\begin{align}\label{ApL-1e3}
\begin{split}
\frac{\big| |y|^{1+a} - |x|^{1+a}\big|}{|x|^{a}|y|^{1+a}} &\leq (1+a)\frac{|x-y|}{|x|^{a}|y|}\leq (1+a)|x-y|^{\varepsilon}\frac{|x|^{1-\varepsilon}+|y|^{1-\varepsilon}}{|x|^{a}|y|}\\
&\leq 2(1+a)\frac{|x-y|^{\varepsilon}}{|x|^{a}|y|^{\varepsilon}}\leq 2(1+a)\frac{|x-y|^{\varepsilon}}{|x|^{a+\varepsilon}}\,.
\end{split}
\end{align}
The lemma follows by gathering \eqref{ApL-1e1}, \eqref{ApL-1e2} and \eqref{ApL-1e3}.
\end{proof}
\end{lemma}

\begin{lemma}\label{ApL0}
Let $0\leq A < B$ and $a\in[0,1)$.  Then,
\begin{align}\label{ApL0e1}
\int^{1}_{0}\big| -A + \theta B \big|^{-a}d\theta &= \frac{ A^{1-a} + (B-A)^{1-a} }{(1-a)B}\,.
\end{align}
In particular, for any $\gamma\in(0,1]$ and $v , v_{*}, x \in\mathbb{R}^{d}$
\begin{align}\label{ApL0e2}
\gamma \int^{1}_{0}\big| v - v_{*} - \theta\, x \big|^{-(1-\gamma)}d\theta \leq (2-\gamma)|v-v_{*}|^{-(1-\gamma)}\,.
\end{align}
\end{lemma}
\begin{proof}
Estimate \eqref{ApL0e1} is clear.  Indeed, break the integral where the integrand is nonpositive and nonnegative, namely, in the intervals $(0,\tfrac{A}{B})$ and $(\tfrac{A}{B},1)$.  Then, perform the integration.\\

Regarding estimate \eqref{ApL0e2}, since $-(1-\gamma)\leq0$, triangle inequality implies
\begin{align*}
\int^{1}_{0}\big| v - v_{*} - \theta\, x \big|^{-(1-\gamma)}\text{d}\theta 
&= \int^{1}_{0}\big| (1-\theta)(v - v_{*}) + \theta\, (v-v_*- x) \big|^{-(1-\gamma)}\text{d}\theta\\
& \leq  \int^{1}_{0}\big| \theta\, |v-v_*- x|  - (1-\theta) |v - v_{*}| \big|^{-(1-\gamma)}\text{d}\theta\\
& =   \int^{1}_{0}\big| - |v-v_*|  + \theta|\big(|v - v_{*}-x| + |v - v_{*}| \big) \big|^{-(1-\gamma)}\text{d}\theta.
\end{align*}
Using \eqref{ApL0e1} with $A=|v-v_*|$ and $B=| v - v_* - x | + | v - v_* |$, it follows that
\begin{align*}
\gamma \int^{1}_{0}\big| v - v_{*} - &\theta\, x \big|^{-(1-\gamma)}d\theta\leq \frac{|v-v_{*}|^{\gamma} + |v-v_{*}-x|^{\gamma} } { |v-v_{*}| + |v-v_{*}-x| }\\
&\leq \bigg(\sup_{X\geq0}\tfrac{1+X^{\gamma}}{1+X}\bigg)\times |v-v_{*}|^{-(1-\gamma)}\leq (2-\gamma)|v-v_{*}|^{-(1-\gamma)}\,.
\end{align*}
This proves estimate \eqref{ApL0e2}.
\end{proof}

\begin{lemma}\label{ApL1}
Let $d\geq 2$, $\gamma\geq0$, $s\in(0,1]$ and $v_{*}\in\mathbb{R}^{d}$.  Then, for any suitable function $f$ it follows that
\begin{equation*}
\big(1+(-\Delta)\big)^{\frac{s}{2}}(f\, \tau_{v_*}|\cdot|^{\gamma}) = \big(1+(-\Delta)\big)^{\frac{s}{2}}f \times \tau_{v_*}|\cdot|^{\gamma} + \mathcal{R}_{v_*}(f)\,,
\end{equation*}
where $\tau$ is the translation operator.  The remainder term is given by
\begin{equation*}
\mathcal{R}_{v_*}(f)(v) := s\int_{\mathbb{R}^{d}}\Bigg(\int^{1}_{0}\nabla |\cdot|^{\gamma} (v-v_{*} - \theta\, x)d\theta\Bigg)\cdot\nabla \varphi(x) f(v-x) \,dx\,,
\end{equation*}
where $\varphi := \mathcal{F}^{-1}\big\{ \langle \cdot \rangle^{s-2} \big\}$  is the Bessel kernel of order $2-s$.
\end{lemma}
\begin{remark}\label{R1Ap}
In Lemma \ref{ApL1} the case $\gamma\in(0,1]$ and $s=1$ is special.  In this case write
\begin{align}\label{R1Ape1}
\begin{split}
&\mathcal{R}_{v_*}(f)(v) = \nabla |\cdot|^{\gamma} (v-v_{*})\cdot \big(\nabla \varphi\ast f\big)(v)\\
+ \int_{\mathbb{R}^{d}}&\Bigg(\int^{1}_{0}\Big(\nabla |\cdot|^{\gamma} (v-v_{*} - \theta\, x) - \nabla |\cdot|^{\gamma} (v-v_{*})\Big)d\theta\Bigg)\cdot\nabla \varphi(x) f(v-x) \,dx\,.
\end{split}
\end{align}
The first term on the right side is a singular integral.  Indeed, the conditions
\begin{equation*}
\widehat{\nabla\varphi}(\xi)= \frac{i\xi}{\langle \xi \rangle}\in L^{\infty}(\mathbb{R}^{d})\,,\quad \nabla\varphi \in \mathcal{C}^{1}\big(\mathbb{R}^{d}\setminus\{0\}\big)\,,\quad |\Delta\varphi(x)|\leq \frac{C}{|x|^{d+1}}\,,
\end{equation*}
satisfied by the convolution kernel are sufficient to properly define the convolution as a bounded operator in $L^{p}(\mathbb{R}^{d})$ for any $p\in(1,\infty)$, see \cite[Chapter 4]{GF} for details.  We just notice that for the case in question here, the boundedness in $L^{2}(\mathbb{R}^{d})$ is trivial since $\widehat{\nabla\varphi}\in L^{\infty}(\mathbb{R}^{d})$, thus
\begin{equation*}
\|f\ast\nabla\varphi\|_{L^{2}(\mathbb{R}^{d})} = \big\|\,\widehat{\nabla\varphi}\;\widehat{f}\,\big\|_{L^{2}(\mathbb{R}^{d})} \leq \| f \|_{L^{2}(\mathbb{R}^{d})}\,.
\end{equation*}
The second term is properly defined as well.  Using Lemma \ref{ApL-1} with $x = v - v_{*} - \theta x$ and $y = v - v_{*}$ one has
\begin{align*}
\Big|\nabla |\cdot|^{\gamma}& (v - v_{*} - \theta\, x) - \nabla |\cdot|^{\gamma} (v-v_{*})\Big|\\
&\leq 2\gamma(2-\gamma)|\theta x|^{\varepsilon}\bigg( \frac{1}{|v - v_{*} - \theta x|^{1-\gamma + \varepsilon}}
+\frac{1}{|v - v_{*}|^{1-\gamma + \varepsilon}}\bigg)\,,\qquad \varepsilon \in (0,1].
 \end{align*}
Choosing $\varepsilon\in(0,\gamma)$ one also has $1-\gamma+\varepsilon < 1 $, thus, applying \eqref{ApL0e2} it follows that
\begin{align}\label{R1Ape2}
\begin{split}
\bigg|\int_{\mathbb{R}^{d}}\Bigg(\int^{1}_{0}&\Big(\nabla |\cdot|^{\gamma} (v-v_{*} - \theta\, x) - \nabla |\cdot|^{\gamma} (v-v_{*})\Big)d\theta\Bigg)\cdot\nabla \varphi(x) f(v-x) \,dx\bigg|\\
&\leq \frac{C}{|v-v_{*}|^{1-\gamma+\varepsilon}} \Big(\big| |\cdot|^{\varepsilon}\nabla\varphi\big|\ast |f|\Big)(v)\,.
\end{split}
\end{align}
Then, this term is well-defined since $\big| |\cdot|^{\varepsilon}\nabla \varphi\big|\in L^{1}(\mathbb{R}^{d})$.
\end{remark}

\begin{proof}
Let $g(\cdot)=\tau_{v_*}|\cdot|^{\gamma}$ and note that
\begin{align*}
\mathcal{F}\big\{(1+(-\Delta))^{\frac{s}{2}}(f\,g)\big\}(\xi) &= \langle \xi \rangle^{s}\big(\widehat{f}\ast\widehat{g}\big)(\xi)\\
&\hspace{-2.5cm}= \mathcal{F}\big\{(1+(-\Delta))^{\frac{s}{2}}f \times g\big\}(\xi) +\int_{\mathbb{R}^{d}}\big(\langle \xi \rangle^{s} - \langle \xi - x \rangle^{s} \big)\widehat{g}(x)\widehat{f}(\xi-x)\text{d}x\,.
\end{align*}
Now, the identity
\begin{equation*}
\langle \xi \rangle^{s} - \langle \xi - x \rangle^{s} = -\int^{1}_{0}\frac{\text{d}}{\text{d}\theta}\langle \xi - \theta x \rangle^{s} \text{d}\theta = s\int^{1}_{0}\frac{(\xi-\theta x)\cdot x}{\langle \xi - \theta x \rangle^{2-s}}\text{d}\theta\,
\end{equation*}
leads to the definition of the remainder 
\begin{align*}
\int_{\mathbb{R}^{d}}\big(\langle \xi  \rangle^{s} - &\langle \xi  - x \rangle^{s} \big)\widehat{g}(x)\widehat{f}(\xi-x)\text{d}x\\
&=s\int^{1}_{0}\int_{\mathbb{R}^{d}} \widehat{\nabla g}(x)\cdot\frac{-i(\xi-\theta x)}{\langle \xi - \theta x \rangle^{2-s}}\widehat{f}(\xi-x) \text{d}x\text{d}\theta=:\widehat{\mathcal{R}_{v_{*}}(f)}(\xi)\,.
\end{align*}
In addition, using properties of the Fourier transform yields
\begin{align*}
\mathcal{F}^{-1}\Big\{\frac{-i(\cdot - \theta x)\widehat{f}(\cdot - x)}{\langle \cdot - \theta x \rangle^{2-s}}\Big\}(v)=e^{ix\cdot v }\big(f\ast\phi)(v)\,,\quad\quad \phi: = - e^{-i(1-\theta)x\cdot }\,\nabla\mathcal{F}^{-1}\big\{ \langle \cdot \rangle^{s-2} \big\}\,. 
\end{align*}
Thus, plugging in the definition of $\mathcal{R}_{v_*}(f)$ one gets
\begin{align}\label{PR}
\begin{split}
\mathcal{R}_{v_*}(f)(v) = s\int^{1}_{0}&\int_{\mathbb{R}^{d}}\widehat{\nabla g}(x)\cdot (f\ast \phi)(v) e^{ix\cdot v}\text{d}x\text{d}\theta\\
&= s\int_{\mathbb{R}^{d}}\Bigg(\int^{1}_{0}\nabla g (v - \theta\,x)\text{d}\theta\Bigg)\cdot\nabla \varphi(x) f(v-x)\text{d}x\,,
\end{split}
\end{align}
where $\varphi := \mathcal{F}^{-1}\big\{ \langle \cdot \rangle^{s-2} \big\}$.  For the last equality in \eqref{PR}, we simply used the definition of convolution for the term $(f\ast \phi)(v)$.
\end{proof}
\begin{lemma}\label{ApLMw}
Let $d\geq 2$, $s\in(0,1]$, $r\in(0,\frac{1}{2})$, $\alpha\in(0,1]$.  Then, for any suitable function $f$, it follows that
\begin{equation*}
\big(1+(-\Delta)\big)^{\frac{s}{2}}(f\, e^{r \langle \cdot \rangle^{\alpha}}) = \big(1+(-\Delta)\big)^{\frac{s}{2}}f \times  e^{r \langle \cdot \rangle^{\alpha}} + \mathcal{R}(f)\,.
\end{equation*}
The remainder term is controled by
\begin{equation*}
\big\|\mathcal{R}(f)\big\|_{ L^{2}(\mathbb{R}^{d}) } \leq C(r,\varphi)\|e^{r \langle \cdot \rangle^{\alpha}}f\|_{ L^{2}(\mathbb{R}^{d}) }\,.
\end{equation*}
\end{lemma}
\begin{proof}
Use formula \eqref{PR} with $g:= e^{r \langle \cdot \rangle^{\alpha}}$.  The validity of such formula for this choice of $g$ is shown by standard approximation procedure.  Note that, for $\alpha\in(0,1]$, one has 
$|\nabla g(v)| \leq r\alpha e^{r\langle v \rangle^\alpha} \langle v \rangle^{\alpha -1} \leq C  e^{r\langle v \rangle^\alpha}$ . Therefore,
\begin{equation*}
\int^{1}_{0}\left| \nabla g (v - \theta\,x) \right| \text{d} \theta
\leq C \int_0^1 e^{r \langle v - \theta x\rangle^\alpha}  \text{d} \theta
\leq C e^{r\langle x \rangle^{\alpha}} e^{r \langle v-x\rangle^{\alpha}}\,.
\end{equation*}
Thus, when $s\in(0,1)$
\begin{equation*}
\big|\mathcal{R}(f)(v)\big| \leq s\,C\int_{\mathbb{R}^{d}}\big|e^{r\langle x \rangle^{\alpha}} \nabla \varphi(x)\big| \big|e^{r\langle v-x \rangle^{\alpha}}f(v-x)\big|\text{d}x\,.
\end{equation*}
As a consequence,
\begin{equation*}
\big\|\mathcal{R}(f)\big\|_{ L^{2}(\mathbb{R}^{d}) } \leq C\|e^{r\langle \cdot \rangle^{\alpha}} \nabla \varphi\|_{L^{1}(\mathbb{R}^{d})}\|e^{r \langle \cdot \rangle^{\alpha}}f\|_{ L^{2}(\mathbb{R}^{d}) }\,.
\end{equation*}
Note that $e^{r\langle \cdot \rangle^{\alpha}} \nabla \varphi\in L^{1}(\mathbb{R}^{d})$ for $r\in[0,\frac{1}{2})$.  For the case $s=1$, follow the  argument of Remark \ref{R1Ap}
\end{proof}

\end{document}